\documentclass[format=acmsmall, review=false, screen=true]{acmart}

\usepackage{amssymb,latexsym,wasysym,dsfont}
\usepackage{stmaryrd} 
\usepackage{amsthm}
\usepackage{calc} 
\usepackage{xspace}
\usepackage{macros,sectionningmacros}
\usepackage{bm} 
\usepackage{changebar}

\definecolor{cinnamon}{rgb}{0.82, 0.41, 0.12}
\definecolor{chocolate}{rgb}{0.48, 0.25, 0.0}
\definecolor{cadmiumgreen}{rgb}{0.0, 0.42, 0.24}

\newif\ifdraft\draftfalse

\ifdraft
\newcommand\modedraft[1]{#1}

\newcommand\nelcomment[1]{\marginpar[{\color{olive}\small\dbend}]{\color{cyan}\small\dbend}{\footnotesize \color{olive}[#1 - \textbf{Nello}]}}
\newcommand\bascomment[1]{\marginpar[{\color{blue}\small\dbend}]{\color{blue}\small\dbend}{\footnotesize \color{blue}[#1 - \textbf{Bastien}]}}
\newcommand\sascomment[1]{\marginpar[{\color{cinnamon}\small\dbend}]{\color{cinnamon}\small\dbend}{\footnotesize
    \color{cinnamon}[#1 - \textbf{Sasha}]}}
\newcommand\rapcomment[1]{\marginpar[{\color{chocolate}\small\dbend}]{\color{chocolate}\small\dbend}{\footnotesize \color{chocolate}[#1 - \textbf{Rapha\"el}]}}
\newcommand{\bas}[1]{{\color{blue}{#1}}}

\newcommand\sr[1]{{\color{blue}{SR: #1}}}
\newcommand{\basl}[1]{{\color{red}{#1}}}
\else
\newcommand\modedraft[1]{}
\newcommand\nelcomment[1]{}
\newcommand\bascomment[1]{}
\newcommand\sascomment[1]{}
\newcommand\rapcomment[1]{}

\newcommand\bas[1]{#1}

\newcommand\sr[1]{}
\newcommand\basl[1]{#1}

\fi

\acmJournal{TOCL}

\setcopyright{acmlicensed}

\acmDOI{0000001.0000001}

\received{September 2018}
\received[revised]{}
\received[accepted]{}

\begin{document}
\title{Strategy Logic with Imperfect Information}

\author{Rapha\"el Berthon}
\affiliation{%
  \institution{\'Ecole Normale Sup\'erieure
    de Rennes}
  \department{Computer Science and Telecommunication}
  \city{Rennes}
  \country{France}
}
\email{raphael.berthon@ens-rennes.fr}

\author{Bastien Maubert}
\orcid{0000-0002-9081-2920}
\affiliation{%
  \institution{Universit\`a degli Studi di Napoli ``Federico II''}
  \department{DIETI}
  \city{Naples}
  \country{Italy}}
\email{bastien.maubert@gmail.com}

\author{Aniello Murano}
\affiliation{%
  \institution{Universit\`a degli Studi di Napoli ``Federico II''}
  \department{DIETI}
  \city{Naples}
  \country{Italy}}
\email{murano@na.infn.it}

\author{Sasha Rubin}
\affiliation{%
  \institution{Universit\`a degli Studi di Napoli ``Federico II''}
  \department{DIETI}
  \city{Naples}
  \country{Italy}}
\email{sasha.rubin@unina.it}

\author{Moshe Y. Vardi}
\affiliation{%
  \institution{Rice University}
  \city{Houston}
  \state{Texas}
  \country{USA}
}
\email{vardi@cs.rice.edu}

\begin{abstract}

We introduce an extension of Strategy Logic for the imperfect-information 
setting, called \SLi, and study its model-checking problem. As this logic
naturally captures multi-player games with imperfect information, this problem
is  undecidable; but we introduce a syntactical class of ``hierarchical
instances'' for which, intuitively, as one goes down the syntactic tree of the
formula, strategy quantifications are concerned with finer observations of the
model, and we prove that model-checking \SLi restricted to hierarchical instances
is decidable.  This result, because it allows for complex patterns of
existential and universal quantification on strategies, greatly generalises
the  decidability of distributed
synthesis for  systems with hierarchical information. It
allows us to easily derive new
decidability results concerning strategic problems under imperfect information
such as the existence of Nash equilibria,  or rational synthesis.

To establish this result we go through an intermediary, ``low-level''
logic much more adapted to automata techniques.  \QCTLs is an extension
of \CTLs with second-order quantification over atomic propositions that
has been used to study strategic logics with perfect information. We
extend it to the imperfect information setting by parameterising
second-order quantifiers with observations.  The simple syntax of the
resulting logic, \QCTLsi, allows us to provide a conceptually neat
reduction of \SLi to \QCTLsi that separates concerns, allowing one to
forget about strategies and players and focus solely on second-order
quantification. While the model-checking problem of \QCTLsi is, in
general, undecidable, we identify a syntactic fragment of hierarchical
formulas and prove, using an automata-theoretic approach, that it is
decidable.  

We apply our result to solve complex strategic problems in the
imperfect-information setting. We first show that the existence of
Nash equilibria for deterministic strategies is decidable in games
with hierarchical information. We also introduce distributed rational synthesis, a
generalisation of rational synthesis to the imperfect-information
setting. Because it can easily be expressed in our
logic, our main result provides a solution to this problem in the case
of hierarchical information.
\end{abstract}
  

\acmSubmissionID{}

%
%


%
%


\maketitle

\renewcommand{\shortauthors}{R. Berthon, B. Maubert, A. Murano,
  S. Rubin and M. Y. Vardi}


\section{Introduction}

Temporal logics such as \LTL~\cite{DBLP:conf/focs/Pnueli77} or
\CTLs~\cite{emerson1986sometimes} are extremely successful logics that
have been studied in great detail and extended in many directions
along the past decades, notably in relation with the development of the
model-checking approach to program
verification~\cite{clarke1999model}.
When considering systems with multiple components such as multi-agent
systems or distributed programs, popular extensions of temporal logics
are the family of so-called \emph{logics for strategic reasoning}, or
\emph{strategic logics}, which introduce operators that can express
the existence of strategies for components to ensure that the system's
executions satisfy certain temporal properties.

A fundational logic in this family is Alternating-time Temporal Logic
(\ATL)~\cite{DBLP:journals/jacm/AlurHK02}. It extends \CTLs with 
a coalition operator $\langle A\rangle\phi$, where $A$ is a subset of
components/agents of the system,  which reads as ``coalition $A$ has
a strategy to enforce property $\phi$ no matter what the other
components/agents do''. This logic is thus quite expressive, as it
allows for instance to express the existence of winning strategies in
games played on graphs. However it is not well suited to reason about other important solution concepts in
game theory, such as Nash equilibria. 
To address this problem Strategy Logic (\SL) was introduced~\cite{chatterjee2010strategy,DBLP:journals/tocl/MogaveroMPV14}. In \SL
strategies are treated as first-order objects, thanks to strategy variables $x$
that can be quantified upon and bound to players: $\Estrato{}$ reads as ``there exists a
strategy $x$'', and $(a,x)$ reads as ``strategy $x$ is assigned to
player $a$''.
This leads to a very  expressive logic that can express many
solution concepts from game-theory such as best response, existence of
Nash equilibria or subgame-perfect equilibria. 

\halfline
\head{Imperfect information}
An essential property of realistic multi-player games is that players
often have a limited view of the  system. Such imperfect
information, or partial observation, is usually captured by equipping
the models with equivalence relations $\obs$ (called
\emph{observations}) over the state space, that specify
indistinguishable states. Strategies are then required to be
\emph{uniform}, \ie, they cannot assign different moves to
indistinguishable situations. Imperfect information is known to make
games computationally harder to solve. For two-player reachability
games, Reif showed in~\cite{reif1984complexity} that deciding the
existence of winning strategies is \EXPTIME-complete for imperfect
information, while it is in \PTIME for perfect information. This
result has later been generalised to omega-regular
objectives~\cite{berwanger2010strategy,doyen2011games}, and adapted to
the setting of program synthesis from temporal
specifications~\cite{pnueli1989synthesisshort,kupferman1999church}. In
the case of multiple players/components/agents, which interests us
here, the situation is even worse: the existence of distributed
winning strategies is undecidable already for two players with
incomparable observation trying to enforce some reachability objective
in the presence of an adversarial third
player~\cite{DBLP:conf/focs/PetersonR79}, and a similar result was
also proved in the framework of distributed synthesis~\cite{PR90}.
Since then, the formal-methods community has spent much effort finding
restrictions and variations that ensure
decidability~\cite{kupermann2001synthesizing,PR90,DBLP:journals/fmsd/GastinSZ09,
  peterson2002decision,DBLP:conf/lics/FinkbeinerS05,pinchinat2005decidable,DBLP:conf/atva/ScheweF07,DBLP:journals/acta/BerwangerMB18}.
The common thread in these approaches is hierarchical information:  players
can be totally ordered according to how well they observe the game. 
Another line of works establishes that decidability can
be retained by forbidding private communication, \ie, by considering
variants around the idea that all new information should be
public~\cite{van1998synthesis,DBLP:conf/concur/MeydenW05,ramanujam2010communication,BelardinelliLMR17a,DBLP:conf/ijcai/BelardinelliLMR17,DBLP:conf/fossacs/Bouyer18}.

\halfline
\head{Strategy Logic with imperfect information}
We propose an extension of Strategy Logic to the
imperfect-information setting, which we call \SLi. 
The first step is to choose how to introduce imperfect information in
the logic. In the formal-methods literature it is typical to associate
observations to players. In \SLi, instead, we associate observations
to strategies: the strategy quantifier $\Estrat{}$ from \SL is now
parameterised by observation $\obs$, written $\Estrato{\obs}$. 
This
novelty allows one to express, in the logic, that a player's
observation changes over time, to capture for instance the loss of a
sensor resulting in a diminished observation power.  We also add to
our logic \SLi the outcome quantifier $\A$ from Branching-time
Strategy Logic (\BSL)~\cite{1908.02488}, which quantifies on
outcomes of strategies currently used by the agents, and the unbinding
operator $\unbind$, which frees an agent from her current
strategy. This does not increase the expressivity of the logic but
presents advantages that we discuss in
Section~\ref{sec-SLi-definition}. For instance it allows us to
naturally consider nondeterministic strategies (Strategy Logic only
considers deterministic ones), which in turn allows us to capture
module checking, the extension of model checking to open
systems~\cite{kupferman2001module,jamroga2014module,jamroga2015module}.


 
The logic \SLi is very powerful: it is an extension of \SL (which
considers perfect information), and of the
imperfect-information strategic logics \ATLsi~\cite{BJ14} and
\ATLssci~\cite{DBLP:journals/corr/LaroussinieMS15}.  As already
mentioned, \SLi can express the distributed synthesis
problem~\cite{PR90}. This problem asks whether there are strategies for
components $a_1, \dots, a_n$ of a distributed system to enforce  some
 property given as an \LTL formula $\psi$ against all behaviours of
 the environment. This can be expressed by the \SLi formula
$\Phi_{\textsc{Synth}} \egdef \Estrato[x_1]{\obs_1} \dots
\Estrato[x_n]{\obs_n}  (a_1, x_1) \dots (a_n, x_n)
 \A\psi$, where $\obs_i$ represents 
the local view of component $a_i$. 
Also, \SLi
can express more complicated specifications by alternating
quantifiers, binding the same strategy to different agents and
rebinding (these are inherited from \SL), as well as changing
observations. For instance, it can express the existence of Nash
equilibria.


\halfline
\head{Main result}
Of course, the high expressivity of \SLi comes at a cost from a
computational complexity point of view.  Its satisfiability problem is undecidable (this is already true
of \SL), and so is its model-checking problem (this is
already true of \ATLsi even for the single formula
$\EstratATL[\{a,b\}] \F p$~\cite{CT11}, which means that agents $a$
and $b$ have a strategy profile to reach a situation where $p$ holds). 
We mentioned that the two main settings in which  decidability is
retrieved for 
distributed synthesis are hierarchical information and public
actions. We extend the first approach to the setting of strategic logics
by introducing a syntactic class of ``hierarchical instances'' of \SLi,
i.e., formula/model pairs, and proving that the model-checking problem
on this class of instances is decidable. Intuitively, an instance of \SLi is hierarchical if, as one goes
down the syntactic tree of the formula, the observations annotating
strategy quantifications can only become finer. Although the class of
hierarchical instances refers not only to the syntax of the logic but
also to the model, the class is syntactical in the sense that it depends
only on the structure of the formula and the observations in the
model. Moreover, it is straightforward to check (in linear time)
whether an
instance is hierarchical or not.
 

\halfline
\head{Applications}
Because the syntax of \SLi allows for arbitrary
alternations of quantifiers in the formulas, our decidability result
for hierarchical instances allows one to decide strategic problems
more involved than module checking and distributed synthesis. For
instance, we show in
Section~\ref{sec-applications} how one can apply our result to
establish that the existence of Nash equilibria is decidable in games
with imperfect information, in the case of hierarchical observations
and deterministic strategies.  This problem is relevant as Nash
equilibria do not always exist in games with imperfect
information~\cite{filiot2018rational}. 
We then consider the problem of rational
synthesis~\cite{fisman2010rational,DBLP:journals/amai/KupfermanPV16,DBLP:conf/icalp/ConduracheFGR16,filiot2018rational},
both in its cooperative and non-cooperative variants. We
introduce the generalisations of these problems to the case of
imperfect information, and  call them cooperative and non-cooperative \emph{rational distributed
   synthesis}. We then apply again our main result to establish that
they are decidable in hierarchical systems for deterministic strategies. For the non-cooperative
variant, we need the additional assumption that the environment
is at least as informed as the system. This is the case for example
when one ignores the actual observation power  of the 
 environment, and considers that it plays with perfect
information.  Doing so yields systems that are
robust to any observation power the environment may have. As Reif puts it, this amounts
to synthesising strategies that are winning even if the opponent
``cheats'' and uses information it is not supposed to have access to~\cite{reif1984complexity}.



\halfline
\head{Approach}
In order to solve the model-checking problem for \SLi we introduce an intermediate logic
\QCTLsi, an extension to the imperfect-information 
setting of \QCTLs~\cite{DBLP:journals/corr/LaroussinieM14}, 
itself an extension of \CTLs by  second-order
quantifiers over atoms.  This is a low-level logic that does not
mention strategies and into which one can effectively compile
instances of \SLi.  States of the models of the logic \QCTLsi have
internal structure, much like the multi-player game structures from~\cite{peterson2001lower} and distributed
systems~\cite{halpern1989complexity}. Model-checking \QCTLsi is also
undecidable (indeed, we show how to reduce from the MSO-theory of the
binary tree extended with the equal-length predicate, known to be
undecidable~\cite{lauchli1987monadic}).  We introduce the syntactical
class \QCTLsih of hierarchical formulas as those in which innermost
quantifiers observe more than outermost quantifiers, and prove that
model-checking is decidable using an extension of the
automata-theoretic approach for branching-time logics.
We provide a reduction from model checking \SLi to model checking \QCTLsi that preserves
being hierarchical, thus  establishing our main
contribution, i.e., that model checking the hierarchical instances of
\SLi is decidable. 

\halfline
\head{Complexity}
To establish the precise complexity of  the problems we solve,  we introduce a new measure on
formulas called \emph{simulation depth}. This measure 
resembles the notion of alternation depth (see, \eg, \cite{DBLP:journals/tocl/MogaveroMPV14}), which counts alternations between
existential and universal strategy (or second-order) quantifications.
But instead of merely counting alternations between such operators,
 simulation depth reflects the
underlying automata operations required to treat formulas, while
remaining a purely syntactical notion. 
We prove that the model-checking
problem for the hierarchical fragment of \QCTLsi and \SLi
 are both
\kEXPTIME[(k+1)]-complete for formulas of simulation depth at most $k$.
Already for the perfect-information fragment, this result is more
precise than what was previously known. Indeed, precise upper bounds
based on alternation depth were known for  syntactic fragments of \SL
but not for the full logic~\cite{DBLP:journals/tocl/MogaveroMPV14}. 



\halfline
\head{Related work} 
\basl{The literature on imperfect information in formal methods and
artificial intelligence is very vast. Imperfect information has been
considered in two-player
games~\cite{reif1984complexity,doyen2011games,berwanger2010strategy},
module checking~\cite{kupferman2001module,jamroga2015module},
distributed
synthesis  of reactive systems~\cite{PR90,kupermann2001synthesizing,DBLP:conf/lics/FinkbeinerS05}
and strategies in multiplayer
games~\cite{DBLP:conf/focs/PetersonR79,peterson2002decision,DBLP:journals/acta/BerwangerMB18},
Nash
 equilibria~\cite{ramanujam2010communication,bouyer2017nash,DBLP:conf/fossacs/Bouyer18},
 rational
 synthesis~\cite{filiot2018rational,DBLP:journals/iandc/GutierrezPW18},
 doomsday equilibria~\cite{DBLP:journals/iandc/ChatterjeeDFR17},
 admissible strategies~\cite{brenguier2017admissibility},
 quantitative objectives~\cite{degorre2010energy,perez2017fixed}, and
 more, some of which we detail below.}

Limited alternation of strategy quantification was studied in
\cite{DBLP:conf/icalp/Chatterjee014}, in which several decidability results are
proved for two and three alternations of existential and universal quantifiers. Except for one where the first
player has perfect information, all the problems solved in this work
are  hierarchical instances, and are thus
particular cases of our main result.



 Quantified $\mu$-Calculus with partial observation is studied
 in~\cite{pinchinat2005decidable}, where the model-checking problem is
 solved by considering a syntactic constraint based on hierarchical
 information, as we do for \QCTLsi. However they consider asynchronous
 perfect recall, and the automata techniques they use to deal with
 imperfect information cannot be used in the synchronous
 perfect-recall setting that we consider in this work. Similarly
the narrowing operation on tree automata (see Section~\ref{sec-ATA}), which is crucial in our
model-checking procedure, considers synchronous perfect recall and
does not seem easy to adapt to the asynchronous setting.

\basl{A number of works have considered strategic logics with imperfect
information.  Various semantics for \ATL with imperfect information
have been studied in,
\eg,~\cite{jamroga2011comparing,DBLP:journals/fuin/JamrogaH04}.  The
model-checking problem for these logics, which is undecidable for
agents with perfect recall~\cite{CT11}, has been studied for agents
with bounded memory, for which decidability is
recovered~\cite{DBLP:journals/entcs/Schobbens04,Lomuscio06mcmas}.  An
epistemic strategic logic with original operators different from those
of \ATL and \SL is proposed in~\cite{huang2014temporal}.  It considers
imperfect information strategies, but only for agents without memory.
Concerning perfect recall, which interest us in this work,
decidability results have also been obtained for
\ATL~\cite{DBLP:journals/jancl/GuelevDE11} and \ATL with strategy
context~\cite{DBLP:journals/corr/LaroussinieMS15} when agents  have the same
information.}

In \cite{1908.02488}, a branching-time variant of  \SL is  extended with epistemic operators and
agents with perfect recall.
Strategies are not required to be uniform in the semantics, but 
 this requirement can be expressed in the language. However no
 decidability result is provided. 
Another variant of \SL extended 
with epistemic operators and imperfect-information, perfect-recall
strategies is presented in~\cite{epistemicSL}, but
 model checking is not studied.  
The latter logic is extended in~\cite{DBLP:conf/ijcai/BelardinelliLMR17}, in
which its model-checking problem is solved on the class of 
systems where all agents' actions are public, which is an assumption
orthogonal to hierarchical information.

  
The work closest to ours is~\cite{DBLP:conf/csl/FinkbeinerS10} which
introduces a  logic $\CL$ in which one can encode many
distributed synthesis problems. In this logic, hierarchical information is 
a necessary consequence of the syntax and semantics, 
and as a result its
model-checking problem is decidable. However, \CL is close in spirit to our
\QCTLsih, and its semantics is less intuitive than that of \SLi. Furthermore, by means of a natural
translation we derive that \CL is strictly included in the
hierarchical instances of \SLi (Section~\ref{subsec:SLi-comparison-CL}). 
In particular,  hierarchical instances of \SLi can express
non-observable goals, while \CL cannot.  When considering  players
that choose their own goals it may be natural to assume that they can observe the
facts that define whether their
objectives are satisfied or not. But when synthesising programs for
instance, it may be enough that their behaviours enforce the desired
properties, without them having the knowledge that it is enforced. Such non-observable winning conditions  have
been studied in, \eg,~\cite{chatterjee2010complexity,degorre2010energy,DBLP:journals/acta/BerwangerMB18}.


\halfline
\head{Outline} In
Section~\ref{sec-SLi} we define \SLi and hierarchical instances, and
present some examples. In Section~\ref{sec-QCTL-imp-inf} we define \QCTLsi and its hierarchical
fragment \QCTLsih.  The proof
that model checking \QCTLsih is decidable, including the required
automata preliminaries, is in Section~\ref{sec-decidable}.  The
hierarchy-preserving translation of \SLi into \QCTLsi is in
Section~\ref{sec-modelcheckingSL}.
In Section~\ref{sec-SLi-comparison} we compare \SLi with related
logics, and in Section~\ref{sec-applications} we apply our main result
to obtain decidability results for various strategic
problems under imperfect information. Finally we conclude and discuss future work in Section~\ref{sec-outlook}.


\section{\SL with imperfect information}
\label{sec-SLi}

In this section we introduce \SLi, an extension of \SL to the
imperfect-information setting with synchronous
perfect-recall. Our
logic presents several original features compared to \SL, which we
discuss in detail in Section~\ref{sec-disc-syntax}:  
we introduce an \emph{outcome quantifier} akin to
the path quantifier in branching-time temporal logics,
we allow for nondeterministic strategies and unbinding agents from
their strategies, 
and we annotate  strategy
quantifiers with  observation symbols which denote the information
available to  strategies. 
We first fix some basic notations.

\subsection{Notations}
Let $\Sigma$ be an alphabet. A \emph{finite} (resp. \emph{infinite}) \emph{word} over $\Sigma$ is an element
of $\Sigma^{*}$ (resp. $\Sigma^{\omega}$). Words are written $w = w_0 w_1 w_2 \ldots$, i.e., indexing begins with $0$. 
The \emph{length} of a finite word $w=w_{0}w_{1}\ldots
w_{n}$ is $|w|\egdef n+1$, and $\last(w)\egdef w_{n}$ is its last
letter.
Given a finite (resp. infinite) word $w$ and $0 \leq i < |w|$  (resp. $i\in\setn$), we let $w_{i}$ be the
letter at position $i$ in $w$, $w_{\leq i}$ is the prefix of $w$ that
ends at position $i$ and $w_{\geq i}$ is the suffix of $w$ that starts
at position $i$.
We write $w\pref w'$ if $w$ is a prefix of $w'$, and $\FPref{w}$ is
the set of finite prefixes of word $w$. 
Finally, 
the domain of a mapping $f$ is written $\dom(f)$, its codomain $\codom(f)$, and for $n\in\setn$
we let $[n]\egdef\{i \in \setn: 1 \leq i \leq n\}$. 
   
\subsection{Syntax}
\label{sec-SLi-definition}

For the rest of the paper, for convenience we fix a number of
parameters for our logics and models: $\APf$ is a finite non-empty set of
\emph{atomic propositions}, $\Agf$ is a finite non-empty set of \emph{agents} or
\emph{players}, and
$\Varf$ is a finite non-empty set of \emph{variables}.
The main novelty of our logic is that we specify which information is available to a strategy,
by annotating strategy quantifiers $\Estrato{}$ with
\emph{observation symbols} $\obs$ from a finite
set  $\Obsf$, that we also fix for the rest of the paper.
When we consider model-checking problems, these data are implicitly
part of the input.

%

\begin{definition}[\SLi Syntax]
  \label{def-SLi}
    The syntax of \SLi is defined by the following grammar:
    \begin{align*}
  \phi\egdef &\; p 
  \mid \neg \phi 
  \mid \phi\ou\phi 
  \mid \Estrato{\obs}\phi 
               \mid \bind{\var}\phi
               \mid \unbind\phi
  \mid \bas{\Eout\psi}               
      \\
      \psi\egdef &\; \phi
                   \mid \neg \psi
                   \mid \psi\ou \psi
                   \mid \X \psi
                   \mid  \psi \until \psi
    \end{align*}
     where 
  $p\in\APf$, $\var\in\Varf$, $\obs\in\Obsf$ and $a\in\Agf$.
\end{definition}

Formulas of type $\phi$ are called \emph{state formulas}, those of type $\psi$
are called \emph{path formulas}, and \SLi consists of all the state formulas
defined by the grammar.

Boolean operators and temporal operators, $\X$ (read ``next'') and
$\until$ (read ``until''), have the usual meaning. The \emph{strategy
  quantifier}  $\Estrato{\obs}$ 
 is a first-order-like quantification on
strategies: $\Estrato{\obs}\phi$ reads as ``there exists a strategy $\var$
that takes decisions based on observation $\obs$  such that $\phi$
holds'', where $\var$ is a strategy variable.
The \emph{binding operator} $\bind{\var}$ assigns a strategy to an
agent, and
 $\bind{\var}\phi$ reads as ``when agent $\ag$ plays strategy $\var$,
 $\phi$ holds''. The \emph{unbinding operator} $\unbind$
 instead releases agent
$\ag$ from her current strategy, if she has one, and
$\bind{\unb}\phi$ reads as ``when
agent $\ag$ is not assigned any strategy, $\phi$ holds''. 
\bas{Finally, the \emph{outcome quantifier} $\Eout$ quantifies on
  outcomes of strategies currently in use: $\Eout\psi$ reads as ``$\psi$
holds in some
outcome of the strategies currently used by the players''.}

We use  abbreviations $\top\egdef p\ou\neg p$, $\perp\egdef\neg\top$, $\phi\impl\phi'\egdef \neg \phi \ou \phi'$,
$\phi\equivaut\phi'\egdef \phi\impl\phi'\et \phi'\impl\phi$ for
boolean connectives,
 $\F\phi \egdef \top \until \phi$ (read ``eventually $\phi$''),  $\always\phi \egdef \neg \F
\neg \phi$ (read ``globally $\phi$'') for temporal operators,
$\Astrato{\obs}\phi\egdef\neg\Estrato{\obs}\neg\phi$ (read ``for all
strategies $\var$ based on observation $\obs$, $\phi$ holds'') and
\bas{$\Aout\psi\egdef\neg\Eout\neg\psi$ (read ``all outcomes of
the current strategies satisfy $\psi$'').}

For every formula $\phi\in\SLi$, we let  $\free(\phi)$ be the set of variables that appear
free in $\phi$, \ie, that
appear out of the scope of a strategy quantifier. A formula $\phi$ is a \emph{sentence} if $\free(\phi)$ is empty.
Finally, we let the \emph{size} $|\phi|$ of a formula $\phi$ be the
number of symbols in $\phi$.



\subsection{Discussion on the syntax}
\label{sec-disc-syntax}

We discuss the syntactic differences between our logic and usual
Strategy Logic.

\halfline
\head{Outcome quantifier}
This quantifier was
 introduced in Branching-time Strategy Logic
 (\BSL)~\cite{1908.02488}, which corresponds to the
 perfect-information fragment of the logic we define here.
It removes a quirk of previous definitions, in which temporal operators
could only be evaluated in contexts where all agents were assigned a
strategy. The outcome quantifier, instead, allows for evaluation of
 temporal properties on partial assignments.
As a result, the notions of free agents and agent-complete assignments from previous
definitions of Strategy Logic are no longer needed (see, \eg, \cite{DBLP:journals/tocl/MogaveroMPV14}).
In
addition, the outcome quantifier highlights the inherent branching-time nature of Strategy
Logic: indeed, in \SL,
branching-time properties can be expressed by resorting to artificial
strategy quantifications for all agents.
It will also make the correspondence with \QCTLsi tighter,
which will allow us to establish the precise complexity of the problem we
solve, while the exact complexity of model checking classic \SL with perfect
information is still not known.
Finally, since the usual definition of
\SL requires that the current strategies define a unique outcome on
which linear-time temporal operators are evaluated, only deterministic
strategies were considered. The introduction of the outcome quantifier
allows us to consider nondeterministic strategies.

\halfline
\head{Unbinding}
With the possibility to evaluate temporal
operators even when some agents are  not bound to any strategy, it becomes
interesting to include the unbinding operator $\unbind$, introduced
in~\cite{DBLP:journals/iandc/LaroussinieM15} for \ATL with strategy
context and also present in \BSL. Note that the outcome quantifier
and unbinding operator do not increase the expressivity of \SL, at the
level of sentences~\cite{1908.02488}.

\halfline
\head{Observations}
In games with imperfect information and \ATL-like logics
with imperfect information, a strategy is always bound to some player,
and thus it is clear with regards to what observations it should be
defined. In \SL on the other hand, strategy quantification and binding
are separate. This adds expressive power with regards to \ATL
 by allowing, for instance,
to assign the same strategy to two different players, but it also
entails that when a quantification is made on a strategy, one does not
know with regards to which observation this strategy should be
defined. We know of three ways to solve this.
One is the approach followed here, which consists in associating with
strategy quantifiers an observation power. 
The second solution is
 to abandon the separation between quantification and binding and to
 use instead
quantifiers of the form $\exists_{a}$, meaning ``there exists a
strategy for player $a$'', like
in~\cite{DBLP:journals/iandc/ChatterjeeHP10,DBLP:journals/corr/Belardinelli14}: with this
operator, the strategy is immediately bound to player $a$, which
indicates with regards to which observation the strategy should be compatible. The third one, adopted
in~\cite{DBLP:conf/ijcai/BelardinelliLMR17}, consists in requiring
that a strategy be uniform for all agents to whom it will be
bound in the formula.
We chose to adopt the first solution for its simplicity and
expressiveness. Indeed the second solution limits expressiveness  by disallowing, for instance, binding the same strategy to
different agents. The third solution leads to a logic that is more
expressive than the second one, but less than the first one. Indeed,
the logic that we study here
can capture the logic from~\cite{DBLP:conf/ijcai/BelardinelliLMR17} (assuming that models
contain observations corresponding to unions of individual
observations), and  in addition \SLi  can express  changes
of agents' observation power.

\subsection{Semantics}
\label{sec-SLmodels}

The models of \SLi are classic concurrent game structures
extended by an interpretation for observation symbols in $\Obsf$. 


\begin{definition}[\CGSi]
  \label{def-CGSi}
  A \emph{concurrent game structure with imperfect information} (or
  \CGSi for short) is a tuple
  $\CGSi=(\Act,\setpos,\trans,\val,\pos_\init,\obsint)$ where
   \begin{itemize}
    \item $\Act$ is a finite non-empty set of \emph{actions},
    \item $\setpos$ is a finite non-empty set of \emph{positions},
   \item $\trans:\setpos\times \Mov^{\Agf}\to \setpos$ is a \emph{transition function}, 
  \item $\val:\setpos\to 2^{\APf}$ is a \emph{labelling function}, 
  \item $\pos_\init \in \setpos$ is an \emph{initial position}, and
  \item  $\obsint:\Obsf\to 2^{\setpos\times\setpos}$ is an 
  \emph{observation interpretation}.
  \end{itemize}
\end{definition}

For $\obs\in\Obsf$,  $\obsint(\obs)$
  is an equivalence relation on positions, that we may write
  $\obseq$.  It represents what 
a strategy with  observation $\obs$ can see: $\obsint(\obs)$-equivalent positions
are indistinguishable to such a  strategy. Also, $\val(\pos)$ is the set of
atomic propositions that hold in position $\pos$.

We define the size $|\CGSi|$ of a \CGSi
$\CGSi=(\Act,\setpos,\trans,\val,\pos_\init,\obsint)$ as the size of
an explicit encoding of the transition function: $|\CGSi|\egdef
|\setpos|\times |\Act|^{|\Agf|}\times \lceil \log(|\setpos|)\rceil$.
We may  write $\pos\in\CGS$ for $\pos\in\setpos$.

We now introduce a number of
notions involved in the semantics of \SLi.
 Consider a \CGSi $\CGSi=(\Act,\setpos,\trans,\val,\pos_\init,\obsint)$.

 \halfline
 \head{Joint actions}
In a position $\pos\in\setpos$, each player $\ag$ chooses an action $\mova\in\Mov$, 
and the game proceeds to position
$\trans(\pos, \jmov)$, where $\jmov\in \Mov^{\Agf}$ stands for the \emph{joint action}
$(\mova)_{\ag\in\Agf}$. Given a joint action
$\jmov=(\mova)_{\ag\in\Agf}$ and $\ag\in\Agf$, we let
$\jmov_{\ag}$ denote $\mova$.

\halfline
\head{Plays}
A \emph{finite} (resp. \emph{infinite}) \emph{play} is a finite (resp. infinite)
word $\fplay=\pos_{0}\ldots \pos_{n}$ (resp. $\iplay=\pos_{0} \pos_{1}\ldots$)
such that $\pos_0=\pos_\init$ and for every $i$ such that $0\leq i<|\fplay|-1$ (resp. $i\geq 0$), there exists a joint action $\jmov$
such that $\trans(\pos_{i}, \jmov)=\pos_{i+1}$.

\halfline
\head{Strategies}
A (nondeterministic) \emph{strategy} is a  function $\strat:\setpos^+
\to 2^\Mov\setminus\emptyset$ that maps each finite play to a nonempty finite set of
actions that the player may play.  A strategy $\strat$ is \emph{deterministic} if for all $\fplay$,
$\strat(\fplay)$ is a singleton.
We let $\setstrat$ denote the set of all strategies.

\halfline
\head{Assignments} 
An \emph{assignment} is
a partial function $\assign:\Agf\union\Varf \partialto \setstrat$, assigning to
each  player and variable in its domain a strategy.
For an assignment
$\assign$, a player $a$ and a strategy $\strat$,
$\assign[a\mapsto\strat]$ is the assignment of domain
$\dom(\assign)\union\{a\}$ that maps $a$ to $\strat$ and is equal to
$\assign$ on the rest of its domain, and 
$\assign[\var\mapsto \strat]$ is defined similarly, where $\var$ is a
variable; also, $\assign[a\mapsto\unb]$ is
 the restriction of $\assign$ to domain $\dom(\assign)\setminus\{a\}$.
In addition, given a formula $\phi\in\SLi$, an assignment is
\emph{variable-complete for $\phi$} if
its domain contains all free variables of $\phi$.

\halfline
\head{Outcomes}
For an assignment $\assign$ and a finite play $\fplay$, we let
$\out(\assign,\fplay)$ be the set of infinite plays that start with
$\fplay$ and are then extended by letting players follow the strategies
assigned by $\assign$. Formally,
 $\out(\assign,\fplay)$ is the set of plays of the form $\fplay \cdot
 \pos_{1}\pos_{2}\ldots$ such that for all $i\geq 0$, there exists
 $\jmov$ such that for all $\ag\in\dom(\assign)\inter\Agf$,
 $\jmov_\ag\in\assign(\ag)(\fplay\cdot\pos_{1}\ldots\pos_i)$ \mbox{ and }
 $\pos_{i+1}=\trans(\pos_{i},\jmov)$, \mbox{ with }
 $\pos_{0}=\last(\fplay)$.

\halfline
\head{Synchronous perfect recall}
In this work we consider
players with \emph{synchronous perfect recall}, meaning that each player
remembers the whole history of a play, a classic assumption in games
with imperfect information and logics of knowledge and time. Each observation
 relation is thus extended to finite plays
as follows: $\fplay \obseq \fplay'$ if $|\fplay|=|\fplay'|$
and $\fplay_{i}\obseq\fplay'_{i}$ for every $i\in\{0,\ldots,
|\fplay|-1\}$.

\halfline
\head{Imperfect-information strategies}
For $\obs\in\Obsf$, a strategy $\strat$ is an \emph{$\obs$-strategy} if
 $\strat(\fplay)=\strat(\fplay')$ whenever $\fplay \obseq
\fplay'$. The latter  constraint captures the essence of
imperfect information, which is that players can base their strategic
choices only on the information available to them. 
For $\obs\in\Obsf$ we let $\setstrato$ be the set of
all $\obs$-strategies. 

\begin{definition}[\SLi semantics]
\label{def-SLi-semantics}
The semantics of a state formula is defined on a \CGSi $\CGSi$, an
assignment  $\assign$ that is variable-complete for $\phi$, and a
finite play $\fplay$. For a path formula $\psi$, the finite play is
replaced with an infinite play $\iplay$ and an index $i\in\setn$. The
definition by mutual induction is as follows:
\[
\begin{array}{lcl}
 \CGSi,\assign,\fplay\modelsSL p & \text{ if } & p\in\val(\last(\fplay))\\[1pt]
 \CGSi,\assign,\fplay\modelsSL \neg\phi & \text{ if } &
  \CGSi,\assign,\fplay\not\modelsSL\phi\\[1pt]
 \CGSi,\assign,\fplay\modelsSL \phi\ou\phi' & \text{ if } &
  \CGSi,\assign,\fplay\modelsSL\phi \;\text{ or }\;
  \CGSi,\assign,\fplay\modelsSL\phi' \\[1pt]
 \CGSi,\assign,\fplay\modelsSL\Estrato{\obs}\phi  & \text{ if } & 
\exists\,   \strat\in\setstrato \;\text{ s.t. } \;
    \CGSi,\assign[\var\mapsto\strat],\fplay\modelsSL \phi\\[1pt]
 \CGSi,\assign,\fplay\modelsSL \bind{\var}\phi & \text{ if } &
 \CGSi,\assign[\ag\mapsto\assign(\var)],\fplay\modelsSL \phi\\[1pt]  
 \CGSi,\assign,\fplay\modelsSL \bind{\unb}\phi & \text{ if } &
          \CGSi,\assign[\ag\mapsto\unb],\fplay\modelsSL \phi\\[1pt]
 \CGSi,\assign,\fplay\modelsSL \Eout\psi & \text{ if } & \text{there exists
                                                         }\iplay \in
                                                         \out(\assign,\fplay)
                                                         \text{ such
                                                         that } 
                                                         \CGSi,\assign,\iplay,|\fplay|-1\modelsSL \psi\\[5pt]
    \CGSi,\assign,\iplay,i\modelsSL \phi & \text{ if } &
                                                         \CGSi,\assign,\iplay_{\leq i}\modelsSL\phi\\[1pt]
   \CGSi,\assign,\iplay,i\modelsSL \neg\psi & \text{ if } &
  \CGSi,\assign,\iplay,i\not\modelsSL\psi\\[1pt]
 \CGSi,\assign,\iplay,i\modelsSL \psi\ou\psi' & \text{ if } &
  \CGSi,\assign,\iplay,i\modelsSL\psi \;\text{ or }\;
  \CGSi,\assign,\iplay,i\modelsSL\psi' \\[1pt]
  \CGSi,\assign,\iplay,i\modelsSL\X\psi & \text{ if } &
  \CGSi,\assign,\iplay,i+1\modelsSL\psi\\[1pt]
\CGSi,\assign,\iplay,i\modelsSL\psi\until\psi' & \text{ if } & \exists\, j\geq i
   \mbox{ s.t. }\CGSi,\assign,\iplay,j\modelsSL \psi'\\ 
   & & \text{ and } \forall\, k \text{ s.t. } i\leq k <j,
\; \CGSi,\assign,\iplay,k\modelsSL \psi
\end{array}
\]
\end{definition}

\begin{remark}
  \label{rem-sentences}
Observe that because of the
semantics of the outcome quantifier, and unlike usual definitions of
\SL, the meaning of an \SLi sentence
 depends on the assignment in which it is evaluated. For instance the \SLi formula
 $\Aout\F p$ is clearly a sentence, but whether
$\CGS,\assign,\fplay\models\Aout\F p$ holds or not depends  on which agents are
bound to a strategy in $\assign$ and what these strategies are.
However, as usual, a sentence does not require an assignment to be
evaluated,  and for an \SLi sentence $\phi$ we let
$\CGSi,\fplay\models\phi$ if $\CGSi,\emptyset,\fplay\models\phi$ for
the empty
assignment $\emptyset$, and we write $\CGSi \models \phi$ if $\CGSi,
\pos_\init \models \phi$. 
\end{remark}

\SL is the fragment of \SLi obtained by interpreting all observation
 symbols as the identity relation (which models perfect information),
 restricting to deterministic strategies,
 and considering only assignments in which each agent has a strategy
 (in this case the outcome of an assignment consists of a single
 play; one can thus get rid of the outcome quantifier and evaluate
 temporal operators in the unique outcome of the current assignment,
 as usually done in \SL). Also, \CTLs is the fragment of \SLi which
uses no  binding, unbinding or strategy quantification.

\subsection{Discussion on the semantics}
\label{sec-discussion-semantics}


We now discuss some aspects of the semantics. 

\halfline
\head{Evaluation on finite plays} Unlike previous definitions of
Strategy Logic, we evaluate formulas on
finite plays (instead of positions),  where the finite play represents
the whole history starting from the initial position of the \CGSi in
which the formula is evaluated. There are several reasons to do
so. First, it allows us to define the semantics more simply without
having to resort to the notion of assignment translations. Second, 
it makes it easier to see the correctness of the reduction to \QCTLsi,
that we present in Section~\ref{sec-modelcheckingSL}. In \SL,
a strategy only has access to the history of the game starting from
the point where the strategy quantifier from which it arises has
been evaluated. In contrast, in \SLi strategies have access to the whole history, starting
from the initial position. However this does not affect the semantics,
in the sense that the perfect-information fragment of
 \SLi with deterministic strategies corresponds to \SL. Indeed, when
 agents have perfect information, having access to the past or not does not 
affect the existence of strategies to enforce temporal properties that
only concern the future.

\halfline
\head{Players not remembering their actions} Our definition
of synchronous perfect recall only considers the sequence of positions
in finite plays, and forgets about actions taken by players. In
particular, it is possible in this definition that a player cannot
distinguish between two finite plays in which she plays
different actions. This definition is standard in games with imperfect
information~\cite{DBLP:conf/concur/MeydenW05,berwanger2010strategy,doyen2011games,DBLP:journals/acta/BerwangerMB18},
since remembering one's actions or not is indifferent for the
existence of distributed winning strategies or Nash equilibria.
However it makes a difference for some more involved solution concepts
that are expressible in strategic logics such as \SLi. 
For instance it is observed in~\cite[Appendix A]{bouyer2017games} that
some games admit subgame-perfect equilibria only if agents remember their
own past actions. Nonetheless we consider the setting where agents do
not remember their actions, as it is the most general. Indeed, as noted
in~\cite[Remark 2.1, p.8]{chatterjee2014partial}, one can simulate
agents that remember their own actions by storing in positions of the
game the information of the last joint move played (this may create
$|\Act|^{|\Agf|}$ copies of each position, but the branching degree is
unchanged). One can then adapt  indistinguishability relations to
take actions into account. For instance, for an observation symbol $\obs$ and
an agent $\ag$, one could consider a new observation symbol
$\obs_\ag$ that would be interpreted in the enriched game structure as
the refinement of $\obseq$ that considers two positions
indistinguishable if they are indistinguishable for $\obseq$ and
contain the same last action for agent $a$. Binding agent $\ag$
only to strategies that use  observation of the form $\obs_\ag$ for
some $\obs$ captures the fact that agent $a$ remembers her
actions.

\halfline
\head{Agents changing observation} In \SLi observations are not bound
to agents but to strategies. And because agents can change their
strategy thanks to the binding operator, it follows that they can
change observation, or more precisely they can successively play with
strategies that have different observations. For instance consider a controller that observes a system through a set of $n$
sensors $S=\{s_1,\ldots,s_n\}$ as in,
\eg,~\cite{bittner2012symbolic}. Let $\obs_i$ be the
observation power provided by the set of sensors $S\setminus\{s_i\}$
(one can think of a system where states are tuples of local states,
each sensor observing one component).
Also let $\obs$ be the observation power provided by the full set $S$
of sensors,
and let atom $\text{fault}_i$ represent
the fact that a fault occurs on sensor $s_i$. The formula
\[\phi\egdef\Estrato{\obs}\bind{\var}\Aout\always\left(\text{safe} \et
  \biget_{i=1}^n\text{fault}_i \impl
  \Estrato{\obs_i}\bind{\var}\Aout\always\text{\,safe}_i\right )\]
expresses that the controller $a$ has a strategy (which
uses all sensors in $S$) to maintain the system safe, and if a sensor
is lost, it can respond by switching to a strategy using the remaining
sensors to maintain some alternative, possibly weaker, security
requirement $\text{safe}_i$.

\subsection{Model checking and hierarchical instances}
\label{sec-mc-hierarch-inst}

We now introduce the main decision problem of this paper, which is
the model-checking problem for \SLi.  An \emph{\SLi-instance}
is a model together with a formula, \ie, it is a pair $(\CGSi,\Phi)$
where $\CGSi$ is a \CGSi and  $\Phi \in \SLi$.

\begin{definition}[Model checking \SLi]
  \label{def-MC-SLi}
The \emph{model-checking problem} for \SLi is the decision problem that,
given an \SLi-instance $(\CGSi,\Phi)$, returns `Yes' if
$\CGSi\models\Phi$, and `No' otherwise.  
\end{definition}

It is well known that
deciding the existence of winning  strategies in
multi-player games with
imperfect information is undecidable for reachability
objectives~\cite{peterson2001lower}. Since this problem is easily
reduced to the  model-checking problem for \SLi, we get the following result.

\begin{theorem}
  \label{theo-undecidable-SLi}
  The model-checking problem for \SLi is undecidable.
\end{theorem}

\head{Hierarchical instances}
We now isolate a sub-problem obtained by restricting attention to \emph{hierarchical instances}. 
Intuitively, an \SLi-instance $(\CGSi,\Phi)$ is hierarchical if, as one goes down a path in the syntactic tree of $\Phi$, 
the observations tied to quantifications become finer. 
%
%

\begin{definition}[Hierarchical instances]
  \label{def-hierarchical-formula}
An \SLi-instance $(\CGSi,\Phi)$ is \emph{hierarchical} if
for every
  subformula $\phi_{1}=\Estrato[y]{\obs_{1}}\phi'_{1}$ of $\Phi$ and subformula
  $\phi_{2}=\Estrato{\obs_{2}}\phi'_{2}$ of $\phi'_1$,
 it holds that 
  $\obsint(\obs_{2})\subseteq \obsint(\obs_{1})$.
\end{definition}

If $\obsint(\obs_{2})\subseteq \obsint(\obs_{1})$ we say that $\obs_{2}$ is \emph{finer} than $\obs_{1}$ in $\CGSi$, and that $\obs_{1}$ is \emph{coarser} than $\obs_{2}$ in $\CGSi$.  Intuitively, this means that a player with observation
$\obs_{2}$ observes game $\CGSi$ no worse than, i.e., knows at least
as much as
a player with observation $\obs_{1}$.

\begin{remark}
  \label{rem-visible-actions}
  If one uses the trick described in
  Section~\ref{sec-discussion-semantics} to model agents that remember
  their own actions, then for an agent $\ag$ to know at least as much as
  another agent $\agb$ it needs to be the case that,  in particular, agent $\ag$ observes all actions
  played by agent $\agb$.
\end{remark}

\begin{example}[Fault-tolerant diagnosibility]
  Consider the following formula from
  Section~\ref{sec-discussion-semantics}:
  \[\phi\egdef\Estrato{\obs}\bind{\var}\Aout\always\left(\text{safe} \et
  \biget_{i=1}^n\text{fault}_i \impl
  \Estrato{\obs_i}\bind{\var}\Aout\always\text{\,safe}_i\right )\]
As already discussed, it expresses that the controller can react to the loss
of a sensor to keep ensuring some property of the system. Clearly, the
controller's observation $\obs_i$ after the loss of sensor $i$ is
coarser than its original observation $\obs$, and thus formula $\phi$
in such a system does not form a hierarchical instance.
\end{example}

We now give an example of scenario where hierarchical instances occur naturally.

\begin{example}[Security levels]
Consider a system with different ``security levels'', where higher levels have access to more data (\ie, can
  observe more).  Assume that the \CGSi $\CGSi$ is such that
  $\obsint(\obs_{n})\subseteq \obsint(\obs_{n-1}) \subseteq \ldots\subseteq
  \obsint(\obs_{1})$: in other words, level $n$ has the highest security
  clearance, while level $1$ has the lowest. Consider that agent $a$
  wants to reach some objective marked by atom ``goal'', that it
   starts with the lowest observation clearance $\obs_{1}$, and that
   atomic formula ``$\text{promote}_i$'' means that the agent is granted
   access to level $i$ (observe that whenever we have
   $\text{promote}_i$, we should also have $\text{promote}_j$ for all $j<i$). For every $i$ we let
   \[\phi_i(\phi')\egdef
     \text{goal}\vee(\text{promote}_i\wedge\Estrato{\obs_i}\bind{\var}\A\F\phi')\]
   Now the formula
  \[\phi\egdef \phi_1(\phi_2(\ldots\phi_{n-1}(\phi_n(\text{goal}))\ldots))\]
means that agent $a$ can enforce her goal, possibly by first getting
access to higher security levels and using this additional observation
power to reach the goal. Because the strategy quantifications that are
deeper in the formula have access to more information, this formula forms a hierarchical instance in $\CGSi$.
\end{example}


Here is the main contribution of this work:

\begin{theorem}
\label{theo-SLi}
The model-checking problem for \SLi restricted to the class of hierarchical instances is decidable.
\end{theorem}

We  prove this result in Section~\ref{sec-modelcheckingSL} by reducing it to the model-checking 
problem for the hierarchical fragment of a logic called  \QCTLs with
imperfect information, which we now introduce and
study in order to use it as an
intermediate, ``low-level'' logic between tree automata and \SLi.
We then discuss some applications of this theorem in Section~\ref{sec-applications}.

\section{\QCTLs with imperfect information}
\label{sec-QCTL-imp-inf}

In this section we introduce an imperfect-information extension of
\QCTLs~\cite{Sis83,Kup95,KMTV00,french2001decidability,
DBLP:journals/corr/LaroussinieM14}, which is an extension of \CTLs
with second-order quantification on atomic propositions. In order to introduce imperfect
information, instead of considering equivalence relations between
states as in concurrent game structures, we
will enrich Kripke structures by giving internal structure to their states,
i.e., we see states as $n$-tuples of local states. This way of
modelling imperfect information is inspired from
  Reif's multi-player
  game structures~\cite{peterson2001lower} and distributed
  systems~\cite{halpern1989complexity}, and we find it very suitable to application of
  automata techniques, as discussed in
  Section~\ref{sec-discussion-QCTL}.

  The syntax of \QCTLsi is 
similar to that of \QCTLs, except that we annotate second-order quantifiers by
subsets $\cobs \subseteq [n]$. The idea is that quantifiers annotated by $\cobs$ 
can only ``observe'' the local states indexed by $i \in \cobs$. 
We define the tree-semantics of \QCTLsi: this means that we
interpret formulas on trees that are the unfoldings of Kripke structures (this
will capture the fact that players in \SLi have synchronous perfect recall). 
We then define the syntactic class of \emph{hierarchical formulas} and prove,
using an automata-theoretic approach, that  model checking this
class of formulas is decidable.

For the rest of the section we fix some natural number $n\in\setn$ which
parameterises the logic \QCTLsi, and which is the number of components
in states of the models.

\subsection{\QCTLsi Syntax}
\label{sec-syntax-QCTLi}

The syntax of \QCTLsi is very similar to that of \QCTLs: the only difference is
that we annotate quantifiers by a set of indices that defines the
``observation'' of that quantifier.

\halfline \head{Concrete observations} 
A set $\cobs \subseteq [n]$ is called a
\emph{concrete observation} (to distinguish it from observations
$\obs$ in the definitions of \SLi).

\begin{definition}[\QCTLsi Syntax]
  \label{def-syntax-QCTLsi}
  The syntax of \QCTLsi is defined by the following grammar:
  \begin{align*}
  \phi\egdef &\; p \mid \neg \phi \mid \phi\ou \phi \mid \E \psi \mid
  \existsp[p]{\cobs} \phi\\
    \psi\egdef &\; \phi \mid \neg \psi \mid \psi\ou \psi \mid \X \psi \mid
  \psi \until \psi
\end{align*}
where $p\in\APf$ 
and $\cobs\subseteq [n]$. 
\end{definition}

Formulas of type $\phi$ are called \emph{state formulas}, those of type $\psi$
are called \emph{path formulas}, and \QCTLsi consists of all the state formulas
defined by the grammar.   We use standard abbreviation 
$\A\psi \egdef \neg\E\neg\psi$.
We also use $\exists p.\,\phi$ as a shorthand for $\existsp[p]{[n]}
\phi$, 
and we let $\forall p.\, \phi\egdef \neg \exists p.\, \neg \phi$.

Given a \QCTLsi formula $\phi$, we define the set of \emph{quantified
  propositions} $\APq(\phi)\subseteq\APf$ as the set of atomic propositions $p$ such that
$\phi$ has a subformula of the form $\existsp[p]{\cobs}\phi$. We also
define the set of \emph{free propositions} $\APfree(\phi)\subseteq\APf$ as the set
of atomic propositions that have an occurrence which is not under the scope of any quantifier
of the form $\existsp[p]{\cobs}$ Observe that 
$\APq(\phi)\inter\APfree(\phi)$ may not be empty, \ie, a proposition may appear both free and quantified in (different
places of) a formula. 
 
\subsection{\QCTLsi semantics}
\label{sec-QCTLsi-semantics}


Several semantics have been considered for \QCTLs, the two
most studied being the \emph{structure semantics} and the \emph{tree
  semantics} (see \cite{DBLP:journals/corr/LaroussinieM14} for more
details). For the semantics of \QCTLsi we adapt the tree semantics,
and we explain the reasons for doing so in Section~\ref{sec-discussion-QCTL}.

As already mentioned, for \QCTLsi we consider  structures whose states are tuples
of local states. We now define these structures and related notions.

\begin{definition}[Compound Kripke structures]
A \emph{compound Kripke structure}, or \CKS, over $\APf$ is a tuple 
$\CKS=(\setstates,\relation,\lab,\sstate_\init)$ where
\begin{itemize}
\item $\setstates\subseteq \prod_{i\in [n]}\setlstates_i$  is a set of
\emph{states}, with $\{\setlstates_{i}\}_{i\in [n]}$ a family of $n$ disjoint finite sets of
\emph{local states}, 
\item $\relation\subseteq\setstates\times\setstates$ is a
left-total\footnote{\ie, for all $\sstate\in\setstates$, there exists $\sstate'$
such that $(\sstate,\sstate')\in\relation$.} \emph{transition
relation}, 
\item $\lab:\setstates\to 2^{\APf}$ is a \emph{\labeling function} and
\item $\sstate_\init \in \setstates$ is an \emph{initial state}.
\end{itemize}
\end{definition}
A \emph{path} in $\CKS$  is an infinite sequence of states
$\spath=\sstate_{0}\sstate_{1}\ldots$ such that
 for all $i\in\setn$,
$(\sstate_{i},\sstate_{i+1})\in \relation$. 
A \emph{finite path} is a finite non-empty prefix of a path.
We may write $\sstate\in\CKS$ for $\sstate\in\setstates$, and we
define the \emph{size} $|\CKS|$ of a \CKS
$\CKS=(\setstates,\relation,\sstate_\init,\lab)$ as its number of states: $|\CKS|\egdef
|\setstates|$. 

Since we will interpret \QCTLsi on unfoldings of \CKS, we now define
infinite trees.

\halfline
\head{Trees}
In many works, trees are defined as prefix-closed sets of words 
with the empty word $\epsilon$ as root. Here trees represent unfoldings of 
Kripke structures, and we find it more convenient to see a node $u$ as a sequence of
states and the root as the initial state. 
Let $\Dirtree$ be a finite set of \emph{directions} (typically a set of states). 
An \emph{$\Dirtree$-tree} $\tree$ 
 is a
nonempty set of words $\tree\subseteq \Dirtree^+$ such that:
\begin{itemize}
  \item\label{p-root} there exists $\racine\in\Dirtree$,  called the
    \emph{root} of $\tree$, such that each
    $\noeud\in\tree$ starts with $\racine$ ($\racine\pref\noeud$);
  \item if $\noeud\cdot\dir\in\tree$ and $\noeud\cdot\dir\neq\racine$, then
    $\noeud\in\tree$, 
  \item if $\noeud\in\tree$ then there exists $\dir\in\Dirtree$ such that $\noeud\cdot\dir\in\tree$.
\end{itemize}

The elements of a tree $\tree$ are called \emph{nodes}.  
  If 
 $\noeud\cdot\dir \in \tree$, we say that $\noeud\cdot\dir$ is a \emph{child} of
 $\noeud$.
 The \emph{depth} of a node $\noeud$ is $|\noeud|$.
An $\Dirtree$-tree $\tree$ is \emph{complete} if for every $\noeud \in
\tree$ and  $\dir \in \Dirtree$,  $\noeud \cdot \dir \in \tree$.
A \emph{\tpath} in $\tree$ is an infinite sequence of nodes $\tpath=\noeud_0\noeud_1\ldots$
such that for all $i\in\setn$, $\noeud_{i+1}$ is a child of
$\noeud_i$,
and $\tPaths(\noeud)$ is the set of \tpaths
 that start in node $\noeud$. 

 \halfline
 \head{Labellings}
 An \emph{$\APf$-\labeled $\Dirtree$-tree}, or
\emph{$(\APf,\Dirtree)$-tree} for short, is a pair
$\ltree=(\tree,\lab)$, where $\tree$ is an $\Dirtree$-tree called the
\emph{domain} of $\ltree$ and
$\lab:\tree \rightarrow 2^{\APf}$ is a \emph{\labeling}, which maps
each node to the set of propositions that hold there.
 For $p\in\APf$, a \emph{$p$-\labeling} for a  tree is a mapping
$\plab:\tree\to \{0,1\}$ that indicates in which nodes $p$ holds, and
for a \labeled tree $\ltree=(\tree,\lab)$, the $p$-\labeling of $\ltree$ is
the $p$-\labeling $\noeud \mapsto 1$ if $p\in\lab(\noeud)$, 0 otherwise. 
The composition of a \labeled tree $\ltree=(\tree,\lab)$ with a
$p$-\labeling $\plab$ for $\tree$ is defined as
$\ltree\prodlab\plab\egdef(\tree,\lab')$, where
$\lab'(\noeud)=\lab(\noeud)\union \{p\}$ if $\plab(\noeud)=1$, and
$\lab(\noeud)\setminus \{p\}$ otherwise.
A $p$-\labeling for a labelled tree $\ltree=(\tree,\lab)$ is a
$p$-\labeling for its domain $\tree$.
A \emph{pointed labelled tree} is a pair $(\ltree,\noeud)$ where
 $\noeud$ is a node of $\ltree$.


If $\noeud=w\cdot \dir$, the \emph{subtree} $\ltree_{\noeud}$ of $\ltree=(\tree,\lab)$ is defined
      as $\ltree_{\noeud}\egdef(\tree_{\noeud},\lab_{\noeud})$ with
      $\tree_{\noeud}=\{\dir\cdot w' \mid w\cdot\dir\cdot w' \in
      \tree\}$, and $\lab_{\noeud}(\dir\cdot w')=\lab(w\cdot\dir\cdot w')$.
A labelled tree is \emph{regular} if it has finitely many disctinct
subtrees.


In the tree semantics of \QCTLsi that we consider here, formulas are
evaluated on tree unfoldings of \CKS, which we now define.
  
\halfline
\head{Tree unfoldings} 
Let $\CKS=(\setstates,\relation,\lab,\sstate_\init)$ be a compound Kripke structure over $\APf$. 
The \emph{tree-unfolding of $\CKS$} is the $(\APf,\setstates)$-tree $\unfold{\sstate}\egdef (\tree,\lab')$, where
    $\tree$ is the set
    of all finite  paths that start in $\sstate_\init$, and
    for every $\noeud\in\tree$,
    $\lab'(\noeud)\egdef \lab(\last(\noeud))$.

    Note that a labelled tree is  regular if and only if it    
    is the unfolding of
some finite Kripke structure.


\halfline  
\head{Narrowing}
Let $\Dirtree$ and $\Dirtreea$ be two finite sets, and let $(\dir,\dira)\in\Dirtree\times\Dirtreea$.
 The
 \emph{$\Dirtree$-narrowing} of $(\dir,\dira)$ is
${\projI[\Dirtree]{(\dir,\dira)}}\egdef \dir$.
%
This definition extends naturally to words and trees over
$\Dirtree\times\Dirtreea$ (point-wise). 

Given a family of (disjoint) sets of local states
$\{\setlstates_{i}\}_{i\in [n]}$ and a subset  $I\subseteq [n]$, we let $\Dirtreei\egdef\prod_{i\in
I}\setlstates_{i}$ if $I\neq\emptyset$ and
$\Dirtreei[\emptyset]\egdef\{\blank\}$, where $\blank$ is a special symbol.
 For $I,J\subseteq [n]$ and $\dirz\in\Dirtreei[I]$,
 we also define 
 ${\projI[{J}]{\dirz}}\egdef
 \projI[{\Dirtreei[I\inter J]}]{\dirz}$,
 where $\dirz$ is seen as a pair $\dirz=(\dir,\dira)\in
 \Dirtreei[I\inter J]\times \Dirtreei[{I\setminus J}]$, \ie, we apply the
 above definition with
 $\Dirtree=\Dirtreei[I\inter J]$ and $\Dirtreea=\Dirtreei[{I\setminus J}]$. This is well defined
 because having taken sets $\setlstates_{i}$ to be disjoint, the
 ordering of local states in $\dirz$ is indifferent.
 We also extend this definition to words and trees.
  In particular, for every
 $\Dirtreei[I]$-tree $\tree$, $\projI[\emptyset]{\tree}$ is the only
 $\Dirtreei[\emptyset]$-tree, $\blank^{\omega}$.


\halfline
\head{Quantification and uniformity} 
In \QCTLsi $\existsp[p]{\cobs} \phi$ holds in a tree $\ltree$ if there is some
$\cobs$-uniform $p$-labelling of $\ltree$ such that $\ltree$ with this
$p$-labelling satisfies $\phi$.  Intuitively, a $p$-labelling of a
tree is $\cobs$-uniform  if 
 every two nodes that are indistinguishable for observation $\cobs$
agree on $p$.

\begin{definition}[$\cobs$-indistinguishability and $\cobs$-uniformity
  in $p$]
  \label{def-uniformity}
Fix $\cobs \subseteq [n]$ and $I \subseteq [n]$.

\begin{itemize}
 \item Two tuples $\dir,\dir'\in\Dirtreei[I]$ are \emph{$\cobs$-indistinguishable},
written $\dir\oequiv\dir'$, if 
$\projI[{\cobs}]{\dir}=\projI[{\cobs}]{\dir'}$.

\item Two words
   $\noeud=\noeud_{0}\ldots\noeud_{i}$ and
   $\noeud'=\noeud'_{0}\ldots\noeud'_{j}$ over alphabet $\Dirtreei[I]$ are
   \emph{$\cobs$-indistinguishable}, written $\noeud\oequivt\noeud'$, if
   $i=j$ and for all $k\in \{0,\ldots,i\}$ we have
   $\noeud_{k}\oequiv\noeud'_{k}$.
 \item A $p$-\labeling for a tree $\tree$ is \emph{$\cobs$-uniform} if for all
    $\noeud,\noeud'\in\tree$, $\noeud\oequivt\noeud'$ implies
 $\plab(\noeud)=\plab(\noeud')$.
\end{itemize}
\end{definition}


\allowdisplaybreaks
\begin{definition}[\QCTLsi semantics]
We define by induction the satisfaction relation $\modelst$ of
\QCTLsi. Let   $\ltree=(\tree,\lab)$ be
an $\APf$-\labeled $\Dirtreei$-tree, 
$\noeud$  a node and $\tpath$  a path in $\tree$:
\begingroup
  \addtolength{\jot}{-3pt}
\begin{alignat*}{3}
  \ltree,\noeud\modelst & 	\,p 			&& \mbox{ if } &&\quad p\in\lab(\noeud)\\
  \ltree,\noeud\modelst & 	\,\neg \phi		&& \mbox{ if } && \quad\ltree,\noeud\not\modelst \phi\\
  \ltree,\noeud\modelst & 	\,\phi \ou \phi'		&& \mbox{ if } &&\quad \ltree,\noeud \modelst \phi \mbox{ or    }\ltree,\noeud\modelst \phi' \\
  \ltree,\noeud\modelst & 	\,\E\psi			&& \mbox{ if } &&\quad \exists\,\tpath\in\tPaths(\noeud) \mbox{      s.t. }\ltree,\tpath\modelst \psi \\
  \ltree,\noeud\modelst & \,\existsp{\cobs} \phi && \mbox{ if }
  && \quad \exists\,\plab \mbox{ a $\cobs$-uniform $p$-\labeling for
    $\ltree$ such that
  }\ltree\prodlab\plab,\noeud\modelst\phi\\
\ltree,\tpath\modelst &		\,\phi 			&& \mbox{ if } &&\quad \ltree,\tpath_{0}\modelst\phi \\ 
\ltree,\tpath\modelst &		\,\neg \psi 		&& \mbox{ if }
&& \quad \ltree,\tpath\not\modelst \psi \\ 
\ltree,\tpath\modelst & \,\psi \ou \psi'	\quad		&& \mbox{ if } && \quad\ltree,\tpath \modelst \psi \mbox{ or }\ltree,\tpath\modelst \psi' \\ 
\ltree,\tpath\modelst & \,\X\psi 				&& \mbox{ if } && \quad\ltree,\tpath_{\geq 1}\modelst \psi \\ 
\ltree,\tpath\modelst & \,\psi\until\psi' 		&& \mbox{ if }
&& \quad\exists\, i\geq 0 \mbox{ s.t.    }\ltree,\tpath_{\geq
  i}\modelst\psi' \text{ and } \forall j \text{ s.t. }0\leq j <i,\; \ltree,\tpath_{\geq j}\modelst\psi
\end{alignat*}
\endgroup
\end{definition}


We write $\ltree\modelst\phi$ for $\ltree,\racine\modelst\phi$,
where $\racine$ is the root of $\ltree$.     Given a \CKS $\KS$ and a
\QCTLsi formula $\phi$, we also write 
$\CKS \modelst \phi$ if
$\CKS,\sstate_\init \models \phi$. 

\begin{example}
  \label{example-QCTLi}
  Consider the following \CTL formula:
\[\ligne{p}\egdef  \A \F p \wedge \A \always (p
\rightarrow \A\X \A\always \neg p).\]

This formula  holds in a labelled tree if and only if each path contains
exactly one node labelled with $p$. Now, consider the following
\QCTLsi formula:
\[\ligneb{p}\egdef\existsp{\emptyset}\ligne{p}.\]
For a blind quantifier, two nodes of a tree are
indistinguishable if and only if they have same depth. Therefore,
this formula holds on a tree iff the $p$'s label all and only the nodes at 
some fixed depth. This formula can thus be used to capture the equal
level predicate on trees. Actually, just as \QCTLs captures \MSO, 
one can prove that \QCTLsi with tree semantics subsumes \MSO with
equal level~\cite{elgot-rabin66,lauchli1987monadic,thomas-msoeqlevel}.
In Theorem~\ref{theo-undecidable} we make use of a similar observation to prove that
model-checking \QCTLsi is undecidable.
\end{example}



\subsection{Discussion on the definition of \QCTLsi}
 \label{sec-discussion-QCTL}

 We now motivate in detail some aspects of \QCTLsi.

 \halfline \head{Modelling of imperfect information} We model
 imperfect information by means of local states (rather than
 equivalence relations) because this greatly facilitates the use of
 automata techniques. More precisely, in our decision procedure of
 Section~\ref{sec-decidable} we use an operation on tree automata
 called \emph{narrowing}, which was introduced
 in~\cite{kupferman1999church} to deal with imperfect-information in
 the context of distributed synthesis for temporal
 specifications. Given an automaton $\auto$ that works on
 $\Dirtree\times \Dirtreea$-trees, where $\Dirtree$ and $\Dirtreea$
 are two finite sets, and assuming that we want to model an operation
 performed on trees while observing only the $\Dirtree$ component of
 each node, this narrowing operation allows one to build from $\auto$
 an automaton $\auto'$ that works on $\Dirtree$-trees, such that
 $\auto'$ accepts an $\Dirtree$-tree if and only if $\auto$ accepts
 its widening to $\Dirtree\times\Dirtreea$ (intuitively, this widening
 is the
 $\Dirtree\times\Dirtreea$-tree in which each node is \labeled as its
 projection on the original $\Dirtree$-tree; see
 Section~\ref{sec-decidable} for details). 

With our definition of compound Kripke structures, their
unfoldings are trees over the  Cartesian product
$\Dirtreei[{[n]}]$. To model a quantification  $\exists^{\cobs}p$ with observation
$\cobs\subseteq [n]$, we can thus use the narrowing operation to
forget about components $\setlstates_{i}$, for $i\in [n]\setminus\cobs$.
We then use the classic projection of nondeterministic tree automata to perform
existential quantification on atomic proposition $p$. Since the choice
of the $p$-labelling is made directly on $\Dirtreei[\cobs]$-trees, it
is necessarily $\cobs$-uniform. 

\halfline
\head{Choice of the tree semantics}
 The two most studied semantics for \QCTLs   are the \emph{structure
   semantics}, in which  formulas are evaluated directly on Kripke
 structures, and the \emph{tree semantics}, in which Kripke structures
 are first unfolded into infinite trees. Tree semantics thus allows
 quantifiers to
 choose the value of a quantified atomic proposition in each
 \emph{finite path} of the model, while in structure semantics the
 choice is only made in each state. 
When \QCTLs is used to express existence of strategies, existential
quantification on atomic propositions labels the structure with
strategic choices; in this kind of application, structure semantics
reflects so-called \emph{positional} or \emph{memoryless} strategies,
while tree semantics captures \emph{perfect-recall} or
\emph{memoryful} strategies. Since in this work we are interested in
perfect-recall strategies, we only consider the tree semantics.

\subsection{Model checking \QCTLsi}
\label{sec-modelchecking}


We now define the model-checking problem studied in the rest of this section.

\begin{definition}[Model checking \QCTLsi]
  \label{def-mc-QCTLi}
  The \emph{model-checking problem for \QCTLsi} is the 
following decision problem: given an instance $(\CKS,\Phi)$ where $\CKS$ is a CKS, 
and $\Phi$ is a \QCTLsi formula, return `Yes' if $\CKS \modelst \Phi$ and `No' otherwise.
\end{definition}

We now prove that the model-checking problem for \QCTLsi
 is undecidable. This comes as 
no surprise since, as we will show in Section~\ref{sec-modelcheckingSL}, \QCTLsi can  express the existence of
distributed winning strategies in imperfect-information games.
 However we propose a proof that shows the connection between \QCTLsi
 and \MSO with equal-level predicate~\cite{elgot-rabin66,lauchli1987monadic,thomas-msoeqlevel}. This proof also has the benefit
 of showing that \QCTLsi is undecidable already for formulas that
 involve only propositional quantifiers that observe either everything
 or nothing.
 
 \newcounter{theo-undecidable}
\setcounter{theo-undecidable}{\value{theorem}}

\begin{theorem}
    \label{theo-undecidable}
The model-checking problem for \QCTLsi is undecidable.
\end{theorem}

\begin{proof}
  Let \MSOeql denote the extension of the logic \MSO (without unary
  predicates) by a binary predicate symbol $\eql$.  \MSOeql is
  interpreted on the full binary tree, and the semantics of
  $\eql(x,y)$ is that $x$ and $y$ have the same depth in the tree. We
  show how to effectively translate \MSOeql into \QCTLsi, and our result follows since the \MSOeql-theory of the
  binary tree is undecidable~\cite{lauchli1987monadic}.
  The translation from \MSOeql to \QCTLsi is obtained by extending that
from \MSO to \QCTL~\cite{DBLP:journals/corr/LaroussinieM14}, using the formula $\ligneb{\cdot}$ from
  Example~\ref{example-QCTLi} to help capture the equal-length
  predicate.

  We define a translation $\transt{\quad}$ from \MSOeql to \QCTLsi such that
  for every tree $\ltree$ with root $\racine$, nodes
  $\noeud_1,\ldots,\noeud_i\in \ltree$ and sets of nodes
  $\setnodes_1,\ldots,\setnodes_j\subseteq \ltree$, and
every \MSOeql formula
  ${\phi(x,x_{1},\ldots,x_{i},X_{1},\ldots,X_{j})}$,
we  have that
  \begin{equation}
    \label{eqn-proof-MSO}
  \ltree,\racine,\liste{\noeud}{1}{i},\liste{\setnodes}{1}{j} \models
  \phi(x,x_1,\ldots,x_i,X_1,\ldots,X_j) \text{\quad if and only if \quad}
  \transt{\ltree},\racine\modelst \transt{\phi}    
  \end{equation}
 where 
  $\transt{\ltree}$ is obtained from $\ltree$ by defining
  the labelling for fresh atomic propositions $p_{x_{k}}$ and
  $p_{X_{k}}$, with $k\in [i]$, as follows:
  $p_{x_{k}}\in\transt{\lab}(\noeud)$ if $\noeud=\noeud_{k}$ and
  $p_{X_{k}}\in\transt{\lab}(\noeud)$ if $\noeud\in\setnodes_{k}$.  

The translation of \MSO to \QCTLs from \cite{DBLP:journals/corr/LaroussinieM14} can be extended
to one from \MSOeql to \QCTLsi by adding rules for the equal level predicate.
Indeed, for  $\phi(x,\liste{x}{1}{i},\liste{X}{1}{j})\in\MSOeql$, we inductively define the \QCTLsi formula 
$\transt{\phi}$ as follows, where $k\in [i]$:
\[
\begin{array}{rclcrcl}
  \transt{x = x_k} &  \egdef  & p_{x_k} & & \transt{x_k = x_l} &  \egdef  & \E\F(p_{x_k}\wedge
  p_{x_l}) \\[5pt]
\transt{x\in X_{k}} &  \egdef  & p_{X_{k}} & & \transt{x_k\in X_{l}} &  \egdef  & \E\F
(p_{x_k}\wedge p_{X_{l}})\\[5pt]
\transt{\neg \phi'} &  \egdef  & \neg \transt{\phi'} & & \transt{\phi_{1}\vee
\phi_{2}} &  \egdef  & \transt{\phi_{1}} \vee \transt{\phi_{2}}\\[5pt]
\transt{\exists x_k.\phi'} &  \egdef  & \multicolumn{5}{l}{ \existsp[p_{x_k}]{}\uniq[p_{x_k}] \wedge
  \transt{\phi'}}\\[5pt]
\transt{\exists X_{k}.\phi'} &  \egdef  & \multicolumn{5}{l}{ \existsp[p_{X_{k}}]{}
\transt{\phi'}}\\[5pt]
  \transt{\edge(x,x_k)} &  \egdef  & \E\X p_{x_k} & & \transt{\edge(x_k,x)} &  \egdef  &
  \perp \\[5pt]
  \transt{\edge(x_k,x_l)} &  \egdef  & \multicolumn{5}{l}{  \E\F
  (p_{x_k}\wedge \E\X p_{x_l})}
\end{array}
\]
where $\uniq[p] \egdef  \E\F p \wedge \forall q.\;\left ( \E\F(p \wedge
  q) \rightarrow \A\always (p \rightarrow q)\right )$ holds in a tree iff
it has exactly one node labelled with $p$. To understand the $x = x_k$
and $x\in X_k$
cases, consider that $x$ will be interpreted as the root. 
For the $\edge(x_k,x)$ case, observe that $x$ has no incoming edge since it is interpreted as the root.
Second-order quantification $\exists X_k$ is translated into
quantification on atomic proposition $p_{X_k}$, and first-order
quantification $\exists x_k$ is treated similarly, with the additional
constraint that quantification is limited to $p_{x_k}$-\labelings that
set $p_{x_k}$ to true in one and only one node of the tree.

The rules for $\eql$ are as follows:
\begin{align*}
  \transt{\eql(x,x_{k})} &  \egdef  p_{x_{k}} \\ \transt{\eql(x_{k},x_l)} &  \egdef  
  \existsp{\emptyset}\ligne{p} \wedge \A\always (p_{x_{k}}\rightarrow p \wedge p_{x_l}\rightarrow p)  
\end{align*}
To understand the first case, observe that since $x$ is interpreted as the root, $x_{k}$ is on the same
level as $x$ if and only if it is also assigned the root. 
For the second case, recall from Example~\ref{example-QCTLi} that the
\QCTLsi formula  $\existsp{\emptyset}\ligne{p}$ places one unique horizontal line of
$p$'s in the tree, and thus requiring that $x_{k}$ and $x_{l}$ be both on this line ensures
that they are on the same level.
The correctness of the translation follows from~\eqref{eqn-proof-MSO},
which is proven by induction.

Now take an instance $(\ltree,\phi(x))$ of the model-checking problem for
\MSOeql on the full binary tree $\ltree$.
Let $\CKS$ be a \CKS with two states $\sstate_{0}$ and
$\sstate_{1}$ (local states are irrelevant here),  whose transition
relation is the complete relation, and with empty labelling function.
Clearly, $\unfold{\sstate_{0}}=\ltree$, and
applying~\eqref{eqn-proof-MSO} we get:
\[\ltree,\sstate_0\models \phi(x)
 \text{\quad iff\quad }\transt{\ltree},\sstate_0\models\transt{\phi}.\]

Observe that in the previous line, because there are no free
variables besides $x$, which stands for the root, we have that $\transt{\ltree}=\ltree=\unfold{\sstate_0}$,
hence we have indeed produced an instance of the model-checking
problem for \QCTLsi. 
\end{proof}

\section{A decidable fragment of \QCTLsi: hierarchy on observations}
\label{sec-decidable}

The main result of this section is the identification of an important decidable fragment of \QCTLsi.

\begin{definition}[Hierarchical formulas]
  \label{def-hierarchical}
  A \QCTLsi formula $\phi$ is \emph{hierarchical} if for all
  subformulas 
  $\phi_{1}=\existsp[p_{1}]{\cobs_{1}}\phi'_{1}$ and
  $\phi_{2}=\existsp[p_{2}]{\cobs_{2}}\phi'_{2}$ of $\phi$ where  
  $\phi_{2}$
  is a subformula of $\phi'_{1}$, we have $\cobs_{1}\subseteq\cobs_{2}$.
\end{definition}

In other words, a formula is hierarchical if innermore propositional
quantifiers observe at least as much as  outermore ones.

\begin{example}
  \label{ex-hier-qctl}
  Formula $\existsp[p]{\{1,2\}}\existsp[q]{\{1,2,4\}}\A\always (p\ou
  q)$ is hierarchical because $\{1,2\}\subseteq \{1,2,4\}$. On the
  other hand, formula $\existsp[p]{\{1,2\}}\big(\existsp[q]{\{1,2,4\}}\A\always (p\ou
  q)\wedge \existsp[q']{\{3\}}\E\F (p\wedge q')\big)$ is not, because
  $\{1,2\}\not\subseteq \{3\}$. Note that neither is it the case
  that $\{3\}\subseteq \{1,2\}$: the observation
  power of quantifiers $\existsp[p]{\{1,2\}}$ and
  $\existsp[q']{\{3\}}$ are incomparable.
  Finally,  formula
  $\forallp[p]{\{1,2,3\}}\existsp[q]{\{1,2\}}.\A\always( p \ou q)$ is not hierarchical
even though $\{1,2\}\subseteq \{1,2,3\}$, as the quantifier that observes best is \emph{higher} in the syntactic tree.
\end{example}

We let \QCTLsih be the set of hierarchical \QCTLsi formulas.

\begin{theorem}
  \label{theo-decidable-QCTLi}
Model checking \QCTLsih is non-elementary decidable.
\end{theorem}

Since our decision procedure for the hierarchical fragment
of \QCTLsi is based on an automata-theoretic approach, we
recall some definitions and results for alternating tree automata.

\subsection{Alternating parity tree automata}
\label{sec-ATA}

  We  recall  alternating parity tree automata. Because  their
  semantics is defined via acceptance games,
  we start with basic definitions for two-player turn-based parity
  games, or simply parity games.

  \halfline
  \head{Parity games}
  A \emph{parity game} is a structure
  $\game=(\setpos,\moves,\pos_\init,\couleur)$, where $\setpos=\setpos_E\uplus\setpos_A$ is a
  set of \emph{positions} partitioned between positions of Eve ($\setpos_E$)
  and those of Adam ($\setpos_A$),
  $\moves\subseteq\setpos\times\setpos$ is a set of \emph{moves},
  $\pos_\init$ is an initial position and $\couleur:\setpos\to \setn$
  is a colouring function of finite codomain. In positions $\setpos_E$, Eve chooses the
  next position, while Adam chooses in positions $\setpos_A$. A play is an infinite sequence of positions
$\pos_0\pos_1\pos_2\ldots$ such that $\pos_0=\pos_\init$ and for all
$i\geq 0$, $(\pos_i,\pos_{i+1})\in\moves$ (written $\pos_i\to
\pos_{i+1}$). We assume that for every $\pos\in\setpos$ there exists
$\pos'\in\setpos$ such that $\pos\to\pos'$. A strategy for Eve is a
partial function $\setpos^*\partialto\setpos$ that maps each finite
prefix of a play ending in a position $\pos\in\setpos_E$ to a next
position $\pos'$ such that $\pos \to\pos'$. A play
$\pos_0\pos_1\pos_2\ldots$ \emph{follows} a strategy $\strat$ of Eve
if for every $i\geq 0$ such that $\pos_i\in\setpos_E$,
$\pos_{i+1}=\strat(\pos_0\ldots\pos_i)$. A strategy $\strat$ is
winning if every play that follows it satisfies the parity condition,
\ie, the least colour seen infinitely often along the play is even.

\halfline  
\head{Parity tree automata}  Because it is sufficient for our needs and simplifies definitions,
  we assume that all input trees are complete trees.
For a set $Z$, $\boolp(Z)$ 
is the set of
formulas built from the elements of $Z$ as atomic propositions using the connectives $\ou$ and
$\et$, 
and with $\top,\perp \in \boolp(Z)$.
An \emph{alternating tree automaton (\ATA) on $(\APf,\Dirtree)$-trees}
is a structure $\auto=(\tQ,\tdelta,\tq_{\init},\couleur)$
where 
$\tQ$ is a finite set of states, $\tq_{\init}\in \tQ$ is an initial
state, $\tdelta : \tQ\times 2^{\APf} \rightarrow \boolp(\Dirtree\times
\tQ)$ is a transition function, and $\couleur:\tQ\to \setn$ is a
colouring function.  To ease reading we shall write atoms in
$\boolp(\Dirtree\times\tQ)$ between brackets, such as
$[x,\tq]$.  
A \emph{nondeterministic tree automaton (\NTA) on
  $(\APf,\Dirtree)$-trees} is an \ATA
$\auto=(\tQ,\tdelta,\tq_{\init},\couleur)$ such that for every $\tq\in
\tQ$ and $a\in 2^{\APf}$, $\tdelta(\tq,a)$ is written in
disjunctive normal form and for every direction $\dir\in \Dirtree$ 
each disjunct contains exactly one element of $\{\dir\}\times Q$.
An \NTA is \emph{deterministic} if for each $\tq\in\tQ$ and $a\in
2^{\APf}$, $\tdelta(\tq,a)$ consists of a single disjunct.

Acceptance of a pointed \labeled tree $(\ltree,\noeud_\init)$, where
$\ltree=(\tree,\lab)$, by an \ATA
$\ATA=(\tQ,\tdelta,\tq_\init,\couleur)$ is defined via the parity game
$\tgame{\ATA}{\ltree}{\noeud_\init}=(\setpos,\moves,\pos_\init,\couleur')$
where $\setpos=\tree\times \tQ \times \boolp (\Dirtree\times \tQ)$,
position $(\noeud,\tq,\pform)$ belongs to Eve if $\pform$ is of
the form $\pform_1\vee \pform_2$ or $[\dir,\tq']$, and to Adam
otherwise,
$\pos_{\init}=(\noeud_\init,\tq_\init,\tdelta(\tq_\init,\noeud_\init))$,
 and  $\couleur'(\noeud,\tq,\pform)=\couleur(\tq)$.
Moves in $\tgame{\tauto}{\ltree}{\noeud_\init}$ are defined by
the following rules:
\[\begin{array}{ll}
 (\noeud,\tq,\pform_1 \;\op\; \pform_2) \move (\noeud,\tq,\pform_i) &
 \mbox{where } 
 \op \in\{\vee,\wedge\} \mbox{ and } i\in\{1,2\},  \\
 \multicolumn{2}{l}{
     (\noeud,\tq,[\dir,\tq']) \move (\noeud\cdot
            \dir,\tq',\tdelta(\tq',\lab(\noeud\cdot\dir)))}
\end{array}\]


Positions of the form $(\noeud,\tq,\top)$ and  $(\noeud,\tq,\perp)$
are sinks, winning for Eve and Adam respectively.

A pointed \labeled tree $(\ltree,\noeud)$ is \emph{accepted}   by
$\tauto$ if Eve has a winning strategy in
$\tgame{\tauto}{\ltree}{\noeud}$, and the \emph{language} of $\ATA$ is the
 set  of pointed \labeled trees accepted by $\ATA$, written
 $\lang(\ATA)$.
We write $\ltree\in\lang(\ATA)$ if $(\ltree,\racine)\in\lang(\ATA)$,
where $\racine$ is the root of $\ltree$.
Finally, the \emph{size} $|\ATA|$ of an \ATA $\ATA$ is its number of states plus
the sum of the sizes of all formulas appearing in the transition function.

\halfline
\head{Word automata}
When the set of directions $X$ is a singleton,  directions can be
forgotten and infinite trees can be identified with  
infinite words.
We thus call \emph{parity word automaton} a parity tree automaton on
$(\APf,\Dirtree)$-trees where $\Dirtree$ is a singleton. In the case of
a nondeterministic parity word automaton,
transitions can be represented as usual as a mapping
    $\Delta:\tQ\times 2^{\APf}\to 2^{\tQ}$ which,
     in a state $\tq\in\tQ$, reading the label $a\in 2^{\APf}$ of
     the current position in the word,
indicates    a set of states $\Delta(\tq,a)$ from which Eve can choose to
    send in the next position of the word.
    


\halfline    
We recall four classic operations on tree automata.

\halfline \head{Complementation} Given an \ATA
$\ATA=(\tQ,\tdelta,\tq_{\init},\couleur)$, we define its
\emph{dual}
$\compl{\ATA}=(\tQ,\compl{\tdelta},\tq_{\init},\compl{\couleur})$
where, for each $\tq\in\tQ$ and $a\in 2^{\APf}$,
$\compl{\tdelta}(\tq,a)$ is the dual of $\tdelta(\tq,a)$, \ie,
conjunctions become disjunctions and vice versa, and
$\couleur(\tq)\egdef\couleur(\tq)+1$.

\begin{theorem}[Complementation~\cite{DBLP:journals/tcs/MullerS95}]
  \label{lab-complement}
  For every \labeled tree $\ltree$ and node $\noeud$ in $\ltree$,
  \[(\ltree,\noeud)\in\lang(\compl{\ATA}) \mbox{ if, and only if, }(\ltree,\noeud)\notin\lang(\ATA).\]
\end{theorem}

\halfline
\head{Projection}  The second construction is
a projection operation, 
used by
Rabin to deal with second-order monadic quantification:
\begin{theorem}[Projection \cite{rabin1969decidability}]
  \label{theo-projection}
  Given an \NTA $\NTA$ on $(\APf,\Dirtree)$-trees and an atomic
  proposition $p\in\APf$, one can build in linear time an \NTA
  $\proj{\NTA}$ on $(\APf\setminus \{p\},\Dirtree)$-trees
  such that
  \[(\ltree,\noeud)\in\lang(\proj{\NTA})\mbox{\bigiff} \mbox{ there exists a
      $p$-\labeling $\plab$ for $\ltree$ s.t. }(\ltree\prodlab\plab,\noeud)\in\lang(\NTA).\]
\end{theorem}

Intuitively,  ${\proj{\NTA}}$ is automaton $\NTA$ 
with the only difference that when it reads the label of a node, it
can choose to run as if $p$ was either true or false: if $\tdelta$ is the transition
function of $\NTA$, that of ${\proj{\NTA}}$ is
$\tdelta'(q,a)=\tdelta(q,a\union \{p\}) \ou
\tdelta(q,a\setminus\{p\})$, for any state $q$ and label $a\in 2^{\APf}$. Another way of seeing it is that
$\proj{\NTA}$ guesses a $p$-labelling for the input tree, and
 simulates $\NTA$ on this modified input.

 \halfline
\head{Simulation} 
To prevent $\proj{\NTA}$ from guessing different labels for a same
node in different executions, it is crucial that $\NTA$ be nondeterministic, 
which is the reason why we need the following result: 

\begin{theorem}[Simulation \cite{DBLP:journals/tcs/MullerS95}]
\label{theo-simulation}
Given an \ATA $\ATA$, one can build in exponential time an \NTA $\NTA$ 
 such that $\lang(\NTA)=\lang(\ATA)$.
\end{theorem}

The last construction was introduced by Kupferman and Vardi to deal with
imperfect information aspects in distributed synthesis. To describe it we need
to define a widening operation on trees which expands the directions
in a tree.

\halfline
\head{Tree widening} 
We generalise the widening operation defined
in~\cite{kupferman1999church}. In the following definitions we fix a  \CKS
$\CKS=(\setstates,\relation,\sstate_\init,\lab)$, and for $I\subseteq [n]$
we let $\SI\egdef \{\projI{\sstate}\, \mid \sstate\in\setstates\}\subseteq \Dirtreei$
 (recall that 
  $\Dirtreei=\prod_{i\in I}\setlstates_{i}$).
Let $J\subseteq I \subseteq [n]$. For every $\SI[J]$-tree $\tree$ rooted
in $\sstate_J$ and  $\sstate_I\in\SI$ such that
$\projI[J]{\sstate_I}=\sstate_J$, we define the \emph{$I$-widening} of
$\tree$ as the $\SI$-tree
\[\liftI{\sstate_I}{\tree}\egdef\{\noeud \in \sstate_I\cdot \SI^* \mid \projI[J]{\noeud}\in\tree\}.\]

For an
 $(\APf,\SI[J])$-tree
$\ltree=(\tree,\lab)$ rooted
in $\sstate_J$ and  $\sstate_I\in\SI$ such that
$\projI[J]{\sstate_I}=\sstate_J$, we let \[\liftI{\sstate_I}{\ltree}\egdef (\liftI{\sstate_I}{\tree},\lab'),
\mbox{ where }\lab'(\noeud)\egdef \lab(\projI[J]{\noeud}).\]

When clear from the context we may omit the subscript $\sstate_I$.
It is the case in particular when referring to \emph{pointed} widenings of trees:
   $(\liftI{}{\ltree},\noeud)$ stands for    $(\liftI{\noeud_0}{\ltree},\noeud)$.

\halfline
\head{Narrowing} We now state a result from~\cite{kupferman1999church} in our slightly
more general setting (the proof can be adapted straightforwardly).
The rough idea of this narrowing operation on ATA is
that,  if one just observes
 $\SI[J]$,  uniform  $p$-labellings on 
$\SI$-trees  can be obtained by choosing the
labellings directly on $\SI[J]$-trees, and then lifting them to $\SI$.

\begin{theorem}[Narrowing \cite{kupferman1999church}]
  \label{theo-narrow}
  Given an \ATA $\ATA$ on $\SI$-trees one can build in linear time an
  \ATA ${\narrow[J]{\ATA}}$ on $\SI[J]$-trees such that for every
  pointed $(\APf,\SI[J])$-tree $(\ltree,\noeud)$ and every
  $\noeud'\in\SI^+$ such that
  $\projI[J]{\noeud'}=\noeud$, 
  \[(\ltree,\noeud)\in\lang(\narrow[J]{\ATA}) \mbox{ iff }(\liftI{}{\ltree},\noeud')\in\lang(\ATA).\]
\end{theorem}

\subsection{Translating \QCTLsih to \ATA}
In order to prove Theorem~\ref{theo-decidable-QCTLi} we need some more
notations and a technical lemma that contains the automata construction.

\begin{definition}
  \label{def-Iphi}
For every $\phi\in\QCTLsi$, we let \[\Iphi\egdef
\biginter_{\cobs\in\setobs}\cobs \subseteq [n],\] where $\setobs$ is the set of
concrete observations
  that occur in $\phi$, with the intersection over the empty set 
  defined as $[n]$. For a \CKS $\CKS$ with   state set
  $\setstates\subseteq \prod_{i\in [n]}\setlstates_i$ we also let
  $\SI[\phi]\egdef \{\projI[{\Iphi}]{\sstate}\mid \sstate\in\setstates\}$. 
\end{definition}
 
Elements of $\SI[\phi]$ will be the possible directions used by the automaton we
  build for $\phi$. In other words, the automaton for $\phi$ will
  work on $\SI[\phi]$-trees.  The intuition is that the observations
  in $\phi$ determine which components of the model's states can be
  observed by the automaton.

Our construction, that transforms a \QCTLsih formula $\phi$ and a \CKS
$\CKS$ into an ATA,
builds upon the classic construction 
  from~\cite{DBLP:journals/jacm/KupfermanVW00}, which builds 
  ATA for \CTLs formulas.  In addition, we use projection of automata
  to treat second-order quantification, and to deal with imperfect
  information we resort to automata narrowing.

  Moreover, we use
tree  automata in an original way that allows us to deal with
  non-observable atomic propositions, which in turn makes it possible
  to consider non-observable winning conditions in our decidable
  fragment of \SLi.
  The classical approach to model checking via tree automata is to
  build an automaton that accepts all tree models of the input
  formula,  and check whether it accepts the unfolding of the
  model~\cite{DBLP:journals/jacm/KupfermanVW00}.
  We instead encode the model in the automata, using the input tree
  only to guess labellings for quantified propositions.
  
\halfline
  \head{Encoding the model in the automaton} Quantification on atomic
  propositions is classically performed by means of automata
  projection (see Theorem~\ref{theo-projection}). But in order to
  obtain a labelling that is uniform with regards to the observation
  of the quantifier, we need to make use of the narrowing operation (see
  Theorem~\ref{theo-narrow}).  Intuitively, to check that a formula
  $\existsp{\cobs} \phi$ holds in a tree $\ltree$, we would like to
  work on its narrowing $\ltree'\egdef\projI[\cobs]{\ltree}$, guess a
  labelling for $p$ on this tree thanks to automata projection,
thus  obtaining a tree $\ltree'_{p}$, take its widening
  $\ltree_{p}''\egdef\liftI[{[n]}]{}{\ltree'_{p}}$, obtaining a
  tree with an $\cobs$-uniform labelling for $p$,  and then check
  that $\phi$ holds on $\ltree_{p}''$.
  The problem  is that unless $\ltree=(\tree,\lab)$ is
$\cobs$-uniform in every atomic proposition in $\APf$, there is no
 way
to define the labelling of $\projI[\cobs]{\tree}$ without losing
information.
This implies that, unless we restrict to models where all
atomic propositions are observable for all observations $\cobs$,
 we cannot pass the model as input to our automata, which will work on
 narrowings of trees.  

Therefore, to model check a
\QCTLsi formula $\phi$ on a \CKS $\CKS$, each state of the
automaton that we build for $\phi$ will contain a state of $\CKS$.  
The automaton can thus guess paths in $\CKS$, and evaluate
free occurrences of atomic propositions in $\CKS$ without reading the
input tree. The input tree 
no longer represents the model, but we use it to carry
labellings for quantified atomic propositions in $\APq(\phi)$: we provide the
automaton with an input tree whose labelling is initially empty, and
the automaton, through successive narrowing and projection operations,
decorates it with uniform labellings for quantified atomic propositions.

We remark that this technique allows one to go beyond Coordination Logic~\cite{DBLP:conf/csl/FinkbeinerS10}: 
by separating between quantified atomic propositions (that need to be
uniform and are carried by the input tree) and free atomic
propositions (that state facts about the model and are coded in the automaton), we manage to remove the
restriction present in \CL, that requires  all facts about
the model to be known to every strategy (see
Proposition~\ref{prop:CLtoSLi} in Section~\ref{subsec:SLi-comparison-CL}).
To do this we  assume without loss of generality that propositions
that are quantified in $\phi$ do not appear free in $\phi$, \ie,
$\APq(\phi)\inter\APfree(\phi)=\emptyset$. 

Finally, given a formula $\phi$, a \CKS $\CKS$ and a state $\sstate\in\CKS$, the
truth value of $\phi$ in $(\CKS,\sstate)$ does not depend on the
labelling of $\CKS$ for atoms in $\APq(\phi)$, which can
thus be forgotten.
Thus, from now on we will assume that an instance $(\CKS,\Phi)$ of the
model-checking problem for \QCTLsi is such that
$\APq(\Phi)\inter\APfree(\Phi)=\emptyset$ and $\CKS$ is a \CKS over $\APfree(\Phi)$.


\halfline
\head{Merging the  decorated input tree and the model}
To state the correctness of our construction, we will need to merge
the labels for quantified propositions, carried by the input tree,
with those for free propositions, carried by \CKS $\CKS$. Because,
through successive widenings, the input tree (represented by $\ltree$
in the definition below) will necessarily be a
complete tree, its domain will always contain the domain of the
unfolding of $\CKS$ (represented by $\ltree'$ below), hence the
following definition.

\begin{definition}[Merge]
  \label{def-merge}
Let
 $\ltree=(\tree,\lab)$ be a complete
$(\APf,\Dirtree)$-tree and  $\ltree'=(\tree',\lab')$ an
$(\APf\,',\Dirtree)$-tree with same root as $\ltree$, where $\APf\inter\APf\,'=\emptyset$. We
 define the \emph{merge} of $\ltree$ and $\ltree'$
 as the $(\APf\union\APf\,',\Dirtree)$-tree \[\ltree\merge\ltree'\egdef
(\tree\cap\tree'=\tree',\lab''),\] where
$\lab''(\noeud)=\lab(\noeud) \union \lab'(\noeud)$.
\end{definition}

We now describe our automata construction.
Let $(\CKS,\Phi)$ be an instance of the model-checking problem
for \QCTLsih, where $\CKS=(\setstates,\relation,\labS,\sstate_\init)$.

\begin{lemma}[Translation]
    \label{lem-final}
    For every subformula
    $\phi$ of $\Phi$ and state $\sstate$ of $\CKS$, one can build an
    \ATA $\bigauto[\sstate]{\phi}$ on $(\APq(\Phi),\SI[\phi])$-trees
    such that for every 
    $(\APq(\Phi),\SI[\phi])$-tree $\ltree$ rooted in
    $\projI[{\Iphi}]{\sstate_\init}$, every $\noeud\in\unfold{}$
    ending in $\sstate$, it holds that
    \begin{equation*}
      (\ltree,\projI[{\Iphi}]{\noeud})\in\lang(\bigauto[\sstate]{\phi}) \mbox{\;\;\;iff\;\;\;}
      \liftI[{[n]}]{}{\ltree}\merge\;\unfold{\sstate},\noeud \modelst
      \phi. 
    \end{equation*}

  \end{lemma}


\begin{proof} 
  Let $\APq=\APq(\Phi)$ and
  $\APfree=\APfree(\Phi)$, and recall that $\CKS$ is labelled over $\APfree$.
  For each state $\sstate\in\setstates$ and each subformula
  $\phi$ of $\Phi$ (note that all subformulas of $\Phi$ are also
  hierarchical), we define by induction on $\phi$ the \ATA
  $\bigauto{\phi}$ on $(\APq,\SI[\phi])$-trees.

\halfline
\noindent$\bm{\phi=p:}$
First, by Definition~\ref{def-Iphi},
 $\SI[\phi]=\SI[{[n]}]=\setstates$. We let $\bigauto{p}$ be the \ATA over $\setstates$-trees with one unique
  state $\tq_\init$, with transition function defined as follows:
  \[\tdelta(\tq_\init,a)=
  \begin{cases}
    \top  & \mbox{if } 
    \begin{array}{c}
p\in\APfree \mbox{ and }p\in\labS(\sstate)  \\
\mbox{ or   }\\ p\in \APq \mbox{ and }p\in a
    \end{array}
\\
    \perp & \mbox{if }  
    \begin{array}{c}
p\in\APfree \mbox{ and }p\notin\labS(\sstate) \\
 \mbox{ or  } \\
p\in \APq \mbox{ and }p\notin a      
    \end{array}
  \end{cases}
\]

\noindent$\bm{\phi=\neg\phi':}$
We let  $\bigauto{\phi}\egdef\compl{\bigauto{\phi'}}$. 

\halfline
\noindent$\bm{\phi=\phi_{1}\ou\phi_{2}:}$
 Because
  $\Iphi=\Iphi[\phi_{1}]\cap\Iphi[\phi_{2}]$, and each
  $\bigauto{\phi_{i}}$ for $i\in\{1,2\}$ works on $\Dirtreei[\phi_{i}]$-trees, we first 
  narrow them so that they work on $\Dirtreei[\phi]$-trees: 
  for $i\in \{1,2\}$, we let $\ATA_{i}\egdef
  {\narrow[\Iphi]{\bigauto{\phi_{i}}}}=(\tQ^{i},\tdelta^{i},\tq_\init^{i},\couleur^{i})$.
  %
  Letting
  $\tq_\init$ be a fresh initial state we define
  $\bigauto{\phi}\egdef(\{\tq_\init\}\union\tQ^1\union\tQ^2,\tdelta,\tq_\init,\couleur)$, where
  $\tdelta$ and $\couleur$ agree with $\tdelta^i$ and $\couleur^i$,
  respectively, on states from $\tQ^i$,  and
  $\tdelta(\tq_\init,a)=\tdelta^1(\tq_\init^1,a)\ou
  \tdelta^2(\tq_\init^2,a)$. The colour of $\tq_\init$ does not matter.
  
  
\halfline
\noindent$\bm{\phi=\E\psi:}$
Let $\max(\psi)=\{\phi_1,\ldots,\phi_k\}$ be the
  set of maximal state subformulas of $\psi$.
  In a first step we
  see these maximal state subformulas as atomic propositions, 
 we see  $\psi$
 as an \LTL formula over $\max(\psi)$, and we  build 
  a nondeterministic
  parity word automaton
  $\autopsi=(\Qpsi,\Deltapsi,\qpsi_\init,\couleurpsi)$ over alphabet $2^{\max(\psi)}$
  that accepts exactly the models of $\psi$ (and uses two colours)~\cite{vardi1994reasoning}.  We define
  the \ATA $\tauto$ that, given as input a
  $(\max(\psi),\SI[\phi])$-tree $\ltree$, nondeterministically
  guesses a path $\tpath$ in
  $\liftI[{[n]}]{}{\ltree}\merge\;\unfold{\sstate}$, or equivalently
a path  in $\CKS$ starting from $\sstate$, and simulates
  $\autopsi$ on it, assuming that the labels it reads while following
  $\projI[\Iphi]{\tpath}$ in its input $\ltree$ correctly represent the truth
  value of formulas in $\max(\psi)$ along
  $\tpath$. 
Recall that $\CKS=(\setstates,\relation,\sstate_\init,\labS)$; we define
 $\tauto\egdef(\tQ,\tdelta,\tq_{\init},\couleur)$, where
\begin{itemize}
\item $\tQ=\Qpsi\times\setstates$, 
\item $\tq_{\init}=(\qpsi_{\init},\sstate)$,
\item for each $(\qpsi,\sstate')\in\tQ$, $\couleur(\qpsi,\sstate')=\couleurpsi(\qpsi)$, and
\item  for each $(\qpsi,\sstate')\in\tQ$
  and $a\in 2^{\max(\psi)}$, 
  \[\tdelta((\qpsi,\sstate'),a)=\bigvee_{\tq'\in\Deltapsi(\qpsi,a)}\bigvee_{
    \sstate''\in\relation(\sstate')}[\projI[\Iphi]{\sstate''},\left(\tq',\sstate''\right)].\]
\end{itemize}
The intuition is that $\tauto$ reads the current label in $2^{\max(\psi)}$, chooses nondeterministically
a transition  in $\autopsi$, chooses a next state $\sstate''$ in $\setstates$
and proceeds in the corresponding direction  $\projI[\Iphi]{\sstate''}\in\SI[\phi]$. 

Now from $\tauto$ we build the automaton $\bigauto{\phi}$ over
$\SI[\phi]$-trees labelled with ``real'' atomic propositions in
$\APq$.
Intuitively, in each node it visits, $\bigauto{\phi}$ guesses what should be its
labelling over $\max(\psi)$, it simulates $\tauto$
accordingly, and checks
that the guess it made is correct.
If, after having guessed a finite path $\noeud\in\unfold{\sstate}$
ending in state $\sstate'$, $\bigauto{\phi}$
 guesses that $\phi_{i}$ holds,
it checks this guess by starting a
 copy of automaton $\bigauto[\sstate']{\phi_{i}}$ from node
 $\noeuda=\projI[\Iphi]{\noeud}$ in its input $\ltree$.

Formally, for each $\sstate'\in\CKS$ and each $\phi_{i}\in\max(\psi)$ we first
build $\bigauto[\sstate']{\phi_i}$, which works on $\SI[\phi_i]$-trees.
Observe that $\Iphi[\phi]=\inter_{i=1}^k \Iphi[\phi_i]$, so that we need to
narrow down these automata\footnote{In the conference version of this
  work~\cite{BMMRV17} we made a mistake here: we wrote that $\Iphi[\phi]=\Iphi[\phi_i]$,
  which is not the case in general. As a consequence we do need to
  narrow down automata, unlike what was written in the conference
  version.}:
We let $\ATA^i_{\sstate'}\egdef\narrow[{\Iphi[\phi]}]{\bigauto[\sstate']{\phi_i}}
=(\tQ^{i}_{\sstate'},\tdelta^{i}_{\sstate'},\tq^{i}_{\sstate'},\couleur^{i}_{\sstate'})$.
We also let
$\compl{\ATA^{i}_{\sstate'}}=(\compl{\tQ^{i}_{\sstate'}},\compl{\delta^{i}_{\sstate'}},\compl{\tq^{i}_{\sstate'}},\compl{\couleur^{i}_{\sstate'}})$
be the dualisation of $\ATA^i_{\sstate'}$, and we assume without loss
of generality  all the state sets are
pairwise disjoint.
We define the \ATA
\[\bigauto{\phi}=(\tQ\cup
\bigcup_{i,\sstate'} \tQ^{i}_{\sstate'} \cup
\compl{\tQ^{i}_{\sstate'}},\tdelta',\tq_{\init},\couleur'),\] where the
colours of states are left as they were in their original automaton,
and $\tdelta'$ is defined as follows. For states in $\tQ^{i}_{\sstate'}$
(resp. $\compl{\tQ^{i}_{\sstate'}}$), $\tdelta'$ agrees with $\tdelta^{i}_{\sstate'}$
(resp. $\compl{\delta^{i}_{\sstate'}}$), and for $(\qpsi,\sstate')\in \tQ$ and $a\in
2^{\APq}$ we let $\tdelta'((\qpsi,\sstate'),a)$ be the disjunction over ${a'\in
    2^{\max(\psi)}}$ of 
  \begin{align}
    \label{lab-transition-automata}
    \Bigg ( \tdelta\left((\qpsi,\sstate'),a'\right)
    \et 
      \biget_{\phi_i\in a'}\tdelta^{i}_{\sstate'}(\tq^{i}_{\sstate'},a) 
    \;\et
     \biget_{\phi_i\notin
      a'}\compl{\delta^{i}_{\sstate'}}(\compl{\tq^{i}_{\sstate'}},a)
    \Bigg ).
\end{align}

\bas{Note that in general it is not possible to define a $\max(\psi)$-labelling
  of $\ltree$ that faithfully represents the truth values of formulas
  in $\max(\psi)$ for all nodes in $\unfold{\sstate}$, because a node in
$\ltree$ may correspond to  different nodes in
$\unfold{\sstate}$ that have same projection on $\SI[\phi]$ but satisfy different formulas of $\max(\psi)$. However this is
not a problem because different copies of $\bigauto{\phi}$ that
visit the same node can guess different
labellings, depending on the actual state of $\KS$ (which is part of
the state of $\bigauto{\phi}$).}

\noindent$\bm{\phi=\exists}^{\bm{\cobs}}\bm{p.\,\phi':}$
We  build automaton $\bigauto{\phi'}$ that works on $\SI[\phi']$-trees;
because $\phi$ is hierarchical, we have that $\cobs\subseteq  \Iphi[\phi']$
and we can narrow down $\bigauto{\phi'}$ to work on $\SI[\cobs]$-trees and obtain
$\ATA_{1}\egdef{\narrow[{\cobs}]{\bigauto{\phi'}}}$. By
Theorem~\ref{theo-simulation} we can
nondeterminise it to get $\ATA_{2}$, which by
Theorem~\ref{theo-projection} we can project with respect to
$p$, finally obtaining $\bigauto{\phi}\egdef \proj{\ATA_{2}}$.

\halfline
\head{Correctness}
We now prove by induction on $\phi$ that the construction is correct. 
In each case, we let $\ltree=(\tree,\lab)$ be a
complete $(\APq,\SI[\phi])$-tree rooted in $\projI[{\Iphi}]{\sstate_\init}$.

\halfline $\bm{\phi=p:}$ First, note that $I_{p}=[n]$, so that
$\ltree$ is rooted in
$\projI[{\Iphi}]{\sstate_\init}=\sstate_\init$, and
$\projI[{\Iphi}]{\noeud}=\noeud$. Also recall that $\noeud$ ends in $\sstate$. Let us
consider first the case where $p\in\APfree$. By definition of
$\bigauto{p}$, we have that $(\ltree,\noeud)\in\lang(\bigauto{p})$ if
and only if
$p\in\labS(\sstate)$. We also have 
$\liftI[{[n]}]{}{\ltree}\merge\;\unfold{\sstate},\noeud\models p$ if
and only if
$p\in\lab'(\noeud)$, where $\lab'$ is the labelling of tree
$\liftI[{[n]}]{}{\ltree}\merge\;\unfold{\sstate}$.
By definition of unfolding and merge, we have that $\lab'(\noeud)=\labS(\sstate)$, which concludes this direction. Now if $p\in\APq$: by
definition of $\bigauto{p}$, we have $(\ltree,\noeud)\in\lang(\bigauto{p})$ if
and only if
$p\in \lab(\noeud)$; also, by definition of the merge and
unfolding, 
 we have that
 $\liftI[{[n]}]{}{\ltree}\merge\;\unfold{\sstate},\noeud\models p$ if and
 only if
 $p\in\lab(\noeud)$, and we are done.

\halfline $\bm{\phi=\neg\phi':}$ 
Correctness follows from the induction hypothesis and Theorem~\ref{lab-complement}.

 \halfline $\bm{\phi_{1}\ou \phi_{2}:}$
 We have $\ATA_{i} =
\narrow[{\Iphi}]{\bigauto{\phi_{i}}}$, so by Theorem~\ref{theo-narrow}
 we have $(\ltree,\projI[{\Iphi}]{\noeud})\in\lang(\ATA_{i})$
 if and only if
$(\liftI[{\Iphi[\phi_i]}]{}{\ltree},\projI[{\Iphi[\phi_i]}]{\noeud})\in\lang(\bigauto{\phi_{i}})$, which by
induction hypothesis holds if and only if
$\liftI[{[n]}]{}{(\liftI[{\Iphi[\phi_{i}]}]{}{\ltree})}\merge\;
\unfold{\sstate},\noeud \modelst \phi_{i}$, \ie, 
$\liftI[{[n]}]{}{\ltree}\merge\;\unfold{\sstate},\noeud\modelst\phi_{i}$.
We conclude by observing that
$\lang(\bigauto{\phi})=\lang(\ATA_{1})\union\lang(\ATA_{2})$.

\halfline $\bm{\phi=\E\psi :}$ Suppose that
$\liftI[{[n]}]{}{\ltree}\merge\;\unfold{\sstate},\noeud\modelst\E\psi$. There
exists an infinite path $\tpath$ in
$\liftI[{[n]}]{}{\ltree}\merge\;\unfold{\sstate}$ starting at $\noeud$
such that
$\liftI[{[n]}]{}{\ltree}\merge\;\unfold{\sstate},\tpath\models\psi$.
Again, let $\max(\psi)$ be the set of maximal state subformulas of
$\phi$, and let $w$ be the infinite word over $2^{\max(\psi)}$ that
agrees with $\tpath$ on the state formulas in $\max(\psi)$, \ie, for
each node $\tpath_k$ of $\tpath$ and formula $\phi_i\in\max(\psi)$,
it holds that $\phi_i\in w_k$ if and only if
$\liftI[{[n]}]{}{\ltree}\merge\;\unfold{\sstate},\tpath_k \modelst
\phi_i$.  To show that
$(\ltree,\projI[\Iphi]{\noeud})\in\lang(\bigauto[\sstate]{\phi})$ we
show that Eve can win the acceptance game
$\tgame{\bigauto[\sstate]{\phi}}{\ltree}{\projI[\Iphi]{\noeud}}$.  In
this game, Eve can guess the path $\tpath$ while the automaton
follows $\projI[{\Iphi[\phi]}]{\tpath}$ in its input $\ltree$,
and she can also guess the corresponding word $w$ on
$2^{\max(\psi)}$. By construction of $\autopsi$, Eve has a winning
strategy $\strat_\psi$ in the acceptance game of $\autopsi$ on
$w$. From $\tpath$, $w$ and $\strat_\psi$
we can easily define a
strategy for Eve in
$\tgame{\bigauto[\sstate]{\phi}}{\ltree}{\projI[\Iphi]{\noeud}}$ on
all positions that can be reached while Adam does not choose to
challenge her on a guess she made for the truth value of some maximal
state subformula, and on such plays this strategy is winning because
$\strat_\psi$ is winning.

Now if Adam challenges her on one of these guesses:
     Let
     $\tpath_k\in\liftI[{[n]}]{}{\ltree}\merge\;\unfold{\sstate}$ be a node along $\tpath$, let
     $\sstate'$ be its last direction and
let     $\tpath_k'=\projI[{\Iphi[\phi]}]{\tpath_k}\in\ltree$.
Assume that in node $\tpath'_k$ of the input tree, in a state
$(\qpsi,\sstate')\in \tQ$, Adam challenges Eve on some
$\phi_i\in \max(\psi)$ that she assumes to be true in $\tpath'_k$,
\ie, such that $\phi_i\in w_k$. Formally, in the evaluation game this
means that Adam chooses the conjunct
$\tdelta^{i}_{\sstate'}(\tq^{i}_{\sstate'},a)$ in transition
formula~\ref{lab-transition-automata}, where
$a=\lab(\tpath'_k)$, thus moving to position
$(\tpath'_k,(\qpsi,\sstate'),\tdelta^{i}_{\sstate'}(\tq^{i}_{\sstate'},a))$.
We want to show that Eve wins from this position.  To do so we first
show that $(\ltree,\tpath'_k)\in\lang(\ATA^{i}_{\sstate'})$.
 
First, since
$\ATA^{i}_{\sstate'}=\narrow[I_{\phi}]{\bigauto[\sstate']{\phi_{i}}}$,
by Theorem~\ref{theo-narrow},
$(\ltree,\tpath'_k)\in\lang(\ATA^i_{\sstate'})$ if and only if
$(\liftI[{\Iphi[\phi_i]}]{}{\ltree},\projI[{\Iphi[\phi_i]}]{\tpath_k})\in\lang(\bigauto[\sstate']{\phi_{i}})$. Next,
by applying the induction hypothesis we get that
$(\liftI[{\Iphi[\phi_i]}]{}{\ltree},\projI[{\Iphi[\phi_i]}]{\tpath_k})\in\lang(\bigauto[\sstate']{\phi_{i}})$
if and only if
$\liftI[{[n]}]{}{\liftI[{\Iphi[\phi_i]}]{}{\ltree}}\merge\;
\unfold{},\tpath_k\modelst \phi_{i}$, \ie,
$\liftI[{[n]}]{}{\ltree}\merge\; \unfold{},\tpath_k\modelst
\phi_{i}$. The latter holds because $\phi_i\in w_k$, and by assumption
$w_k$ agrees with $\tpath_k$ on $\phi_i$. Thus
$(\ltree,\tpath'_k)\in\lang(\ATA^i_{\sstate'})$.

This means that Eve
has a winning strategy from the initial position
$(\tpath'_k,\tq^i_{\sstate'},\tdelta^i_{\sstate'}(\tq^i_{\sstate'},a))$
of the acceptance game of $\ATA^{i}_{\sstate'}$ on $(\ltree,\tpath'_k)$.
 Since
$(\tpath'_k,\tq^i_{\sstate'},\tdelta^i_{\sstate'}(\tq^i_{\sstate'},a))$
and
$(\tpath'_k,(\qpsi,\sstate'),\tdelta^i_{\sstate'}(\tq^i_{\sstate'},a))$
contain the same node $\tpath'_k$ and transition formula
$\tdelta^i_{\sstate'}(\tq^i_{\sstate'},a)$, the subgames that start in
these positions are isomorphic and a winning strategy in one
of these positions induces
 a winning strategy in the other, and therefore Eve wins Adam's
challenge (recall that positional strategies are
  sufficient in parity games \cite{DBLP:journals/tcs/Zielonka98}).  With a similar argument, we get that also when Adam
challenges Eve on some $\phi_i\in \max(\psi)$ assumed not to be true in node
$\tpath_k$, Eve wins the challenge. Finally,
    Eve wins the acceptance game of
    $\bigauto{\phi}$ on $(\ltree,\projI[\Iphi]{\noeud})$, and thus $(\ltree,\projI[\Iphi]{\noeud})\in\lang(\bigauto{\phi})$.

    For the other direction, assume that
    $(\ltree,\projI[\Iphi]{\noeud})\in\lang(\bigauto{\phi})$, \ie, Eve wins the evaluation game
    of $\bigauto{\phi}$ on $(\ltree,\projI[\Iphi]{\noeud})$.  A winning strategy for Eve
    describes a path $\tpath$ 
    in $\unfold{\sstate}$ from $\sstate$, which is also
 a path in 
    $\liftI[{[n]}]{}{\ltree}\merge\;\unfold{\sstate}$ from $\noeud$. This
    winning strategy also defines an infinite word $w$
    over $2^{\max(\psi)}$ such that $w$ agrees with $\tpath$ on the
    formulas in $\max(\psi)$, and it also describes a winning strategy
    for Eve in the acceptance game of 
 $\autopsi$ on $w$. Hence $\liftI[{[n]}]{}{\ltree}\merge\;\unfold{\sstate},\tpath\modelst\psi$, and
    $\liftI[{[n]}]{}{\ltree}\merge\;\unfold{\sstate},\noeud\modelst \phi$.

    \halfline $\bm{\phi=\exists}^{\bm{\cobs}}\bm{p.\,\phi':}$
    First, by definition we have $\Iphi=\cobs\inter\Iphi[\phi']$. Because
    $\phi$ is hierarchical, $\cobs\subseteq\cobs'$ for
    every $\cobs'$ that occurs in $\phi'$, and thus $\cobs\subseteq\Iphi[\phi']$.
    It follows that $\Iphi=\cobs$. Next, by
    Theorem~\ref{theo-projection} we have that
    \begin{equation}
      \label{eq:5}
      (\ltree,\projI[\Iphi]{\noeud})\in\lang(\bigauto{\phi}) \mbox{\bigiff}
      \exists\,\plab \mbox{ a $p$-\labeling for
    $\ltree$ such that
  }(\ltree\prodlab\plab,\noeud)\in\lang(\ATA_{2}).        
    \end{equation}
By Theorem~\ref{theo-simulation},
    $\lang(\ATA_{2})=\lang(\ATA_{1})$, and since $\ATA_{1}
    =\narrow[\cobs]{\bigauto{\phi'}}=\narrow[{\Iphi[\phi]}]{\bigauto{\phi'}}$
    we get by Theorem~\ref{theo-narrow} that
    \begin{equation}
      \label{eq:6}
      (\ltree\prodlab\plab,\projI[\Iphi]{\noeud})\in\lang(\ATA_{2}) \mbox{\bigiff}
      (\liftI[{\Dirtreei[\phi']}]{}{(\ltree\prodlab\plab)},\projI[{\Iphi[\phi']}]{\noeud})\in\lang(\bigauto{\phi'}).   
    \end{equation}
By induction hypothesis, 
\begin{equation}
  \label{eq:7}
(\liftI[{\Dirtreei[\phi']}]{}{(\ltree\prodlab\plab)},\projI[{\Iphi[\phi']}]{\noeud})\in\lang(\bigauto{\phi'})
\mbox{\bigiff} \liftI[{[n]}]{}{\liftI[{\Dirtreei[\phi']}]{}{(\ltree\prodlab\plab)}}\merge\;\unfold{\sstate},\noeud\modelst
\phi'.  
\end{equation}
Now, by points \eqref{eq:5}, \eqref{eq:6} and \eqref{eq:7}
and the fact that
 $\liftI[{[n]}]{}{\liftI[{\Iphi[\phi']}]{}{(\ltree\prodlab\plab)}}=   \liftI[{[n]}]{}{(\ltree\prodlab\plab)}$,
we get that
\begin{equation}
  \label{eq:8}
  (\ltree,\projI[\Iphi]{\noeud})\in\lang(\bigauto{\phi}) \mbox{\bigiff}
        \exists\,\plab \mbox{ a $p$-\labeling for
    $\ltree$ such that
  }
  \liftI[{[n]}]{}{(\ltree\prodlab\plab)}\merge\;\unfold{\sstate},\noeud\modelst\phi'.
\end{equation}

We now prove the following equation  which, together with point \eqref{eq:8},
concludes the proof:

\begin{equation}
  \label{eq:9}
  \begin{array}{c}
      \exists\,\plab \mbox{ a $p$-\labeling for
    $\ltree$ such that
  }
  \liftI[{[n]}]{}{(\ltree\prodlab\plab)}\merge\;\unfold{\sstate},\noeud\modelst\phi'
  \\
    \mbox{\bigiff}\\
  \liftI[{[n]}]{}{\ltree}\merge\;\unfold{\sstate},\noeud\modelst\existsp{\cobs}\phi'
  \end{array}
\end{equation}

Assume that there exists a $p$-\labeling $\plab$ for $\ltree$ such
that
$\liftI[{[n]}]{}{(\ltree\prodlab\plab)}\merge\;\unfold{\sstate},\noeud\modelst\phi'$.
Let $\plab'$ be the $p$-\labeling of
$\liftI[{[n]}]{}{(\ltree\prodlab\plab)}\merge\;\unfold{\sstate}$.  By
definition of the merge, $\plab'$ is equal to the $p$-labelling of
$\liftI[{[n]}]{}{(\ltree\prodlab\plab)}$, which by definition of the
widening is $\Iphi$-uniform, \ie, it is $\cobs$-uniform. In addition,
it is clear that
$\liftI[{[n]}]{}{(\ltree\prodlab\plab)}\merge\;\unfold{\sstate}=(\liftI[{[n]}]{}{\ltree}\merge\;\unfold{\sstate})\prodlab\plab'$,
which concludes this direction.

For the other direction, assume that
$\liftI[{[n]}]{}{\ltree}\merge\;\unfold{\sstate},\noeud\modelst\existsp{\cobs}\phi'$:
there exists a $\cobs$-uniform  $p$-\labeling $\plab'$ for $\liftI[{[n]}]{}{\ltree}\merge\;\unfold{\sstate}$ 
 such that $(\liftI[{[n]}]{}{\ltree}\merge\;\unfold{\sstate})\prodlab\plab',\noeud\modelst\phi'$.  We define
 a $p$-\labeling $\plab$ for $\ltree$ such that
 $\liftI[{[n]}]{}{(\ltree\prodlab\plab)}\merge\;\unfold{\sstate},\noeud\modelst\phi'$.
 First, let us write $\ltree'=\liftI[{[n]}]{}{\ltree}\merge\;\unfold{\sstate}=(\tree',\lab')$.
 For each node $\noeud$ of $\ltree$, 
 let
 \[
\plab(\noeud)=
\begin{cases}
\plab'(\noeud') &   \mbox{if there exists }\noeud'\in\tree' \mbox{
  such that }\projI[\cobs]{\noeud'}=\noeud,\\
0 & \mbox{otherwise.}
\end{cases}
   \]
   This is well defined because $\plab'$ is $\cobs$-uniform in $p$, so
   that if two nodes $\noeud',\noeuda'$ project on $\noeud$, we have
   $\noeud'\oequivt\noeuda'$ and thus
   $\plab'(\noeud')=\plab'(\noeuda')$.  In case there is no
   $\noeud'\in\tree'$ such that
   $\projI[{\Iphi[\phi]}]{\noeud'}=\noeud$, the value of
   $\plab(\noeud)$ has no impact on
   $\liftI[{[n]}]{}{(\ltree\prodlab\plab)}\merge\;\unfold{\sstate}$.
Finally,
   $\liftI[{[n]}]{}{(\ltree\prodlab\plab)}\merge\;\unfold{\sstate}=
   (\liftI[{[n]}]{}{\ltree}\merge\;\unfold{\sstate})\prodlab\plab'$,
hence the result.

%
\end{proof}

\subsection{Proof of Theorem~\ref{theo-decidable-QCTLi}}
\label{sec-proof-theo-decidable}

We  now prove Theorem~\ref{theo-decidable-QCTLi}. Let
$\CKS$ be a \CKS with initial state $\sstate_\init$, and let $\Phi\in\QCTLsih$.
By Lemma~\ref{lem-final} one can build an \ATA
$\bigauto[\sstate_\init]{\Phi}$ such that for every
 labelled $\SI[\phi]$-tree $\ltree$ rooted in
 $\projI[{\Iphi[\phi]}]{\sstate_\init}$, and every node
 $\noeud\in\unfold{\sstate_\init}$, 
     $(\ltree,\projI[{\Iphi}]{\noeud})\in\lang(\bigauto[\sstate_\init]{\phi})$
     if, and only if,
      $\liftI[{[n]}]{}{\ltree}\merge \;\unfold{\sstate_\init},\noeud \modelst
      \Phi$.      

      Let $\tree$ be the full $\SI[\phi]$-tree rooted in
      $\projI[{\Iphi[\phi]}]{\sstate_\init}$, and let
      $\ltree=(\tree,\lab_{\emptyset})$,  where $\lab_{\emptyset}$ is
       the empty labelling.
      Clearly,
$\liftI[{[n]}]{}{\ltree}\merge
\;\unfold{\sstate_\init}=\unfold{\sstate_\init}$, and because $\ltree$ is
rooted in $\projI[{\Iphi[\phi]}]{\sstate_\init}$, we get that 
      $\ltree\in\lang(\bigauto[\sstate_\init]{\phi})$ if, and only if
      $\unfold{\sstate_\init}\modelst \Phi$, \ie, $\CKS\modelst\Phi$. It remains to 
check whether tree $\ltree$, which is regular, is accepted by
$\bigauto[\sstate_\init]{\Phi}$. This can be done by solving a parity
game built from the
product of $\bigauto[\sstate_\init]{\Phi}$ with a finite Kripke structure
representing $\ltree$~\cite{loeding}.


      \subsection{Complexity}
      \label{sec-complexity-QCTL}


     To state a precise upper bound on the complexity of our
        procedure, we first introduce a syntactic notion of
        \emph{simulation depth} for formulas of \SLi. While
        alternation depth (see, \eg,
        \cite{DBLP:journals/tocl/MogaveroMPV14}) simply counts the number of
        alternations between existential and universal strategy
        quantifications, simulation depth reflects automata operations required to treat a
        formula, and  counts the
        maximum number of nested simulations of alternating
        tree automata that need to be performed when applying our automata
        construction. However, like alternation depth, it is a purely
        syntactic notion. Formally we define a function
        $\ndd:\QCTLsi\to \setn \times \{\nd,\alt\}$ which returns, for
        each formula $\phi$, a pair $\ndd(\phi)=(k,x)$ where $k$ is the
        simulation depth of $\phi$, and $x\in\{\nd,\alt\}$
        indicates whether the automaton $\bigauto{\phi}$ built from
        $\phi$ and a state $\sstate$ of a \CKS $\CKS$ is
        nondeterministic ($\nd$) or alternating ($\alt$). If
        $\ndd(\phi)=(k,x)$ we shall denote $k$ by $\ndd_k(\phi)$ and
        $x$ by $\ndd_x(\phi)$.  The inductive definition for state
        formulas is as follows:
      \[
        \begin{array}{l}
          \ndd(p) \egdef (0,\nd)\\[5pt]
          
          \ndd(\neg \phi) \egdef (\ndd_k(\phi),\alt)\\[5pt]
          
          \ndd(\phi_1\ou\phi_2) \egdef \left
          (\max_{i\in\{1,2\}}\ndd_k(\phi_i), x \right ),\\
          \hfill \mbox{where
          }x=
          \begin{cases}
            \nd & \mbox{if }\ndd_x(\phi_1)=\ndd_x(\phi_2)=\nd\\
            \alt & \mbox{otherwise}
          \end{cases}
          \\[15pt]
\bas{   \ndd(\E\psi)\egdef
          \begin{cases}
            (0,\nd) & \mbox{if }\psi\in\LTL\\
            (\max_{\phi\in\max(\psi)}\ndd_k(\phi), \alt) & \mbox{otherwise}
          \end{cases}
          }\\[15pt]
          \ndd(\existsp{\cobs}\phi)\egdef (k,\nd),\\
          \hfill\quad\quad\quad\quad\quad\mbox{where }
          k=\begin{cases}
            \ndd_k(\phi) & \mbox{if }\ndd_x(\phi)=\nd \mbox{ and
            }\cobs=\Iphi \quad \mbox{(recall Definition~\ref{def-Iphi})}\\
            \ndd_k(\phi)+1 & \mbox{otherwise}
          \end{cases}                              
        \end{array}
      \]

We explain each case. For an atomic proposition $p$, the automaton
$\bigauto{p}$ is clearly nondeterministic and no simulation is
involved in its construction. For a formula $\neg \phi$, the automaton
$\bigauto{\neg\phi}$ is obtained by dualising $\bigauto{\phi}$, an operation
that in general does not return a nondeterministic automaton but an
alternating one; also this dualisation does not involve any
simulation, hence the definition of the first component.
Now for the disjunction, the first component should be clear; for the
second one, observe that by construction of $\bigauto{\phi_1 \ou
  \phi_2}$, if both $\bigauto{\phi_1}$ and $\bigauto{\phi_2}$ are
nondeterministic, then so is $\bigauto{\phi_1\ou\phi_2}$; otherwise,
it is alternating. For the path quantifier, by construction
$\bigauto{\E\psi}$ is alternating in the general case as it starts copies of automata for
each maximal state subformula in $\psi$; for the first component, we recall that $\max(\psi)$
denotes the set of these maximal state subformulas and we observe that
no additional simulation is performed to build
$\bigauto{\E\psi}$ besides those needed to construct the automata for
the maximal state subformulas. \bas{If $\psi$ is an \LTL formula, 
 then one can
build the nondeterministic word automaton $\autopsi$ directly working on ``real''
atomic propositions in $\APq\union\APfree$. The automaton $\tauto$
can then be built working directly on $\APq$, with $\autopsi$ reading
valuations for $\APq$ in the input tree and those for atoms in $\APfree$ in the current
state of $\CKS$. Because we do not need to guess
valuations of maximal state subformulas and launch
additional automata to check that these guesses are correct, we obtain
a nondeterministic automaton.}
Finally, for a formula of the form $\existsp{\cobs}\phi$, to
build automaton $\bigauto{\existsp{\cobs}\phi}$ we first
build $\bigauto{\phi}$, which we then narrow down to work on
$\Dirtreei[\cobs]$-trees. Since the narrowing operation introduces alternation,
we need to nondeterminise the resulting automaton before projecting it
with respect to $p$. Now observe that if $\Iphi=\cobs$ we do not
need to perform this narrowing, and thus if $\bigauto{\phi}$ is a
nondeterministic automaton we can directly perform the
projection. This justifies the definition of the first component; for
the second one, observe that the projection of a nondeterministic
automaton is also nondeterministic. 

\bas{
\begin{example}
  \label{ex-ndd}
  Assume that $n=3$, \ie, states of \CKS have three components (recall
  that $[3]=\{1,2,3\}$). Let us
  consider formula $\phi=\forallp[p]{\{1,3\}}\forallp[q]{[3]}\existsp[r]{[3]}\E\always
  (p\et q\ou r)$. We
  describe how its simulation  depth is computed.
  First, let us rewrite $\phi=\neg\existsp[p]{\{1,3\}} \existsp[q]{[3]}\neg\existsp[r]{[3]}\E\always
  (p\et q\ou r)$.
  
  Since $\always (p\et q \ou r)$ is an \LTL formula,
  $\ndd(\E\always (p\et q \ou r))=(0,\nd)$. Next, because
  $\Iphi[\E\always (p\et q \ou r)]=[3]$, it follows that
  $\ndd(\existsp[r]{[3]}\E\always (p\et q \ou r))=(0,\nd)$, and
  $\ndd(\neg\existsp[r]{[3]}\E\always (p\et q \ou r))=(0,\alt)$.  Next
  we have that
  $\ndd(\existsp[q]{[3]}\neg\existsp[r]{[3]}\E\always (p\et q \ou
  r))=(1,\nd)$. This reflects the fact that the automaton obtained
  for formula $\neg\existsp[r]{[3]}\E\always (p\et q \ou r)$, which is
  alternating because of complementation, needs to be simulated
  before projecting it over $q$. Then, because $\{1,3\}\neq[3]$, it
  holds that
  $\ndd(\existsp[p]{\{1,3\}}\existsp[q]{[3]}\neg\existsp[r]{[3]}\E\always
  (p\et q \ou r))=(2,\nd)$: to project
  over $p$ we first need to narrow down the previous automaton to make
  it see only components 1 and 3, and because the narrowing
  operation introduces alternation, the resulting automaton needs to
  be simulated before projecting it. Finally, we get
  that $\ndd(\phi)=(2,\alt)$
\end{example}
}

We now introduce two additional depth measures on \QCTLsi formulas,
which help us establish more precise upper bounds  on the sizes of the automata
we build. For every \QCTLsi formula $\phi$, we let $\Ed(\phi)$ be the
maximum number of nested path quantifiers $\E$ in $\phi$, and $\rEd(\phi)$
is the maximum number of nested second-order quantifiers $\exists$ in $\phi$.
We also inductively define the function $\tower{k}{n}$, for
$k,n\in\setn$, as follows: $\tower{0}{n}\egdef n$ and
$\tower{k+1}{n}\egdef 2^{\tower{k}{n}}$.


      \begin{proposition}
        \label{prop-size-automata}
        Let $\Phi$ be a $\QCTLsih$ formula, $\CKS$ a \CKS and
        $\sstate\in\CKS$ a state.  
        \begin{itemize}
        \item If
        $\ndd_k(\Phi)=0$, $\bigauto[\sstate]{\Phi}$ has at most
        $\gsphi[\Phi]$ states  and
        2 colours, and
        \item 
       if $\ndd_k(\Phi)\geq 1$, 
        $\bigauto[\sstate]{\Phi}$ has at most
        $\tower{\ndd_k(\Phi)}{\gsphi[\Phi] \log \gsphi[\Phi]}$  states
        and its number of colours is at most
        $\tower{\ndd_k(\Phi)-1}{\gsphi[\Phi] \log \gsphi[\Phi]}$,
      \end{itemize}
              where
              $\gsphi[\Phi]=m_1^{\rEd(\Phi)}|\Phi||\CKS|^{\Ed(\Phi)}2^{m_2|\Phi|\Ed(\Phi)}$,
              with $m_1, m_2\in\setn$  constants.
              
 Also, if $\bigauto{\phi}$ has state set $\tQ$ then for each
 $\tq\in\tQ$ and $a\in 2^{\APq(\Phi)}$ we have
     $|\tdelta(\tq,a)|\leq |\CKS| |\tQ|^{|\CKS|} 2^{\paraphi|\phi|}$,
  where $\paraphi=1+\Ed(\phi)$.
      \end{proposition}

Constants $m_1$ and $m_2$ are derived from constants in the
              complexity of, respectively, the simulation procedure, and  the
              procedure that builds a nondeterministic word automaton
              for an \LTL formula. For more detail, see the proof of
              Proposition~\ref{prop-size-automata} in Appendix~\ref{sec-appendix-a}.

From this we get the following complexity result.
      
      \begin{proposition}
        \label{prop-compl-QCTL}
        The model-checking problem for \QCTLsih formulas of
        simulation depth at most $k$  is  \kEXPTIME[(k+1)]-complete. 
      \end{proposition}

      \begin{proof} 
 We start with the upper bounds. For an
        instance $(\Phi,\CKS)$, our decision
        procedure in  Section~\ref{sec-proof-theo-decidable} first builds
        automaton $\bigauto[\sstate_\init]{\Phi}$,
        and concludes by testing whether the full $\SI[\Phi]$-tree
        with empty labelling $\ltree$
        is accepted by $\bigauto[\sstate_\init]{\Phi}$. 
        This can be done in time
$O((|\bigauto[\sstate_\init]{\Phi}|\cdot|\ltree|)^l)$, where $|\ltree|$
is the size of a smallest Kripke structure representing the regular
tree $\ltree$, $|\bigauto[\sstate_\init]{\Phi}|$ is the sum of the
number of states and sizes of formulas in the transition function of $\bigauto[\sstate_\init]{\Phi}$, and $l$  the
number of colours it uses~\cite{loeding}.
Clearly $\ltree$ can be represented by  a Kripke structure of size
 $|\SI[\Phi]|$, so that
 $|\ltree|\leq|\SI[\Phi]|\leq |\CKS|$.

 By Proposition~\ref{prop-size-automata}, each formula in the
 transition function of $\bigauto[\sstate_\init]{\Phi}$ is of size at
 most $|\CKS||\tQ|^{|\CKS|}2^{H|\Phi|}$, where $\tQ$ is the set of
 states in $\bigauto[\sstate_\init]{\Phi}$ and $H=1+\Ed(\Phi)$. There are at most
 $|\tQ|2^{|\APq(\Phi)|}$ such formulas\footnote{In fact the final
   automaton $\bigauto[\sstate_\init]{\Phi}$ does not read anything in
 its input, hence the alphabet could be considered to be a
 singleton. We thus have only $|\tQ|$ different formulas in the
 transition function, at most.} and $|\APq(\Phi)|\leq |\Phi|$,
so that
 $|\bigauto[\sstate_\init]{\Phi}|\leq
 |\tQ|+|\tQ|2^{|\APq(\Phi)|}|\CKS||\tQ|^{|\CKS|}2^{H|\Phi|}\leq
 2|\CKS||\tQ|^{|\CKS|+1}2^{(H+1)|\Phi|}$. Also $H+1\leq |\Phi|$, so we
finally have $|\bigauto[\sstate_\init]{\Phi}|\leq 2|\CKS||\tQ|^{|\CKS|+1}2^{|\Phi|^2}$.

        If $k=0$, by Proposition~\ref{prop-size-automata}
        $\bigauto[\sstate_\init]{\Phi}$ has at most $\gsphi[\Phi]$ states
        and 2 colours, and $\gsphi[\Phi]$  is polynomial in $|\CKS|$
        but exponential in $|\Phi|$. Therefore 
$|\bigauto[\sstate_\init]{\Phi}|$ is exponential
in $|\Phi|$ and in $|\CKS|$, and so is the complexity of checking that
$\ltree$ is accepted by $\bigauto[\sstate_\init]{\Phi}$.

 If $k\geq 1$, by Proposition~\ref{prop-size-automata},
$|\tQ|$ is $k$-exponential
in $\gsphi[\Phi]\log \gsphi[\Phi]$,
and $\gsphi[\Phi]\log \gsphi[\Phi]$ itself is polynomial in $|\CKS|$ but exponential in $|\Phi|$.
As a result, $|\bigauto[\sstate_\init]{\Phi}|$ is $(k+1)$-exponential
in $|\Phi|$ and $k$-exponential in $|\CKS|$.
Finally, still by Proposition~\ref{prop-size-automata}, the number of
colours $l$ is
 $(k-1)$-exponential in $\gsphi[\Phi] \log \gsphi[\Phi]$, hence
 $k$-exponential in $|\Phi|$. Checking that $\ltree$ is accepted by
 $\bigauto[\sstate_\init]{\Phi}$ can thus be done in time
 $(k+1)$-exponential in $|\Phi|$, and $k$-exponential in $|\CKS|$,
 which finishes to establish the upper bounds.

        For the lower bounds, consider the fragment \EQkCTLs of \QCTLs
        (with perfect information) which  consists in formulas in
        prenex normal form, \ie, with all
        second-order quantifications at the beginning,  with at most
        $k$ alternations between existential and universal
        quantifiers, counting the first quantifier as one alternation (see~\cite[p.8]{DBLP:journals/corr/LaroussinieM14} for
        a formal definition). Clearly, \EQkCTLs is a fragment of \QCTLsi (with
        $n=1$), and formulas of \EQkCTLs have simulation depth
        at most $k$. It is proved in~\cite{DBLP:journals/corr/LaroussinieM14} that model
        checking \EQkCTLs is \kEXPTIME[(k+1)]-hard.        
      \end{proof}

      \begin{remark}
        \label{rem-lower-bounds}
        One may wonder why we do not get our lower bounds from the
        distributed synthesis problem in systems with hierarchical
        information. The reason is that this problem is \kEXPTIME-complete for
        \LTL or \CTLs
        specifications~\cite{PR90,kupermann2001synthesizing} and can
        be expressed with formulas of simulation depth $k$, and thus
        would only provide \kEXPTIME lower-bounds for simulation depth
        $k$, while our problem is \kEXPTIME[{k+1}]-complete. This may seem
        surprising, but we point out that thanks to
        alternation of existential and universal quantifiers, \QCTLsi formulas with
        simulation depth $k$ can express more complex problems than
        classic distributed synthesis, such as existence of Nash
        equilibria (see Section~\ref{sec-NE-hier}).
      \end{remark}
      
\head{Improved upper bound} We now refine the previous result by observing that 
some subformulas can be model-checked independently in a bottom-up
labelling algorithm which uses the above model-checking procedure as a
subroutine. The height of exponential of the overall procedure for a
formula $\Phi$ is thus
determined by the maximal simulation-depth of the successive
independent subformulas $\phi$ treated by the labelling algorithm, instead of
the simulation depth of the full formula $\Phi$. To make this
precise we define the \emph{simulation number} of a sentence,
akin to the alternation number introduced
in~\cite{DBLP:journals/tocl/MogaveroMPV14}.

Let $\Phi\in\QCTLsi$, and assume without loss of generality that
 $\APq(\Phi)\inter\APfree(\Phi)=\emptyset$.
A state subformula $\phi$ of
$\Phi$ is a \emph{subsentence} if no atom quantified in $\Phi$ appears
free in $\phi$, \ie, $\phi$ is a subsentence of $\Phi$ if
 $\APq(\Phi)\inter\APfree(\phi)=\emptyset$.\footnote{Observe that since
we always assume that $\APq(\Phi)\inter\APfree(\Phi)=\emptyset$,
$\Phi$ is a subsentence of itself.}
The \emph{simulation
  number} $\ndn(\Phi)$ of a \QCTLsi formula $\Phi$ is the maximal
simulation depth of $\Phi$'s subsentences, where the simulation depth
is computed by considering strict subsentences as atoms.

Note that because temporal operators of \SLi can only talk about the
future, the truth value of a subsentence in a node $\noeud$ of an
unfolding $\unfold{\sstate}$ only depends on the current state
$\last(\noeud)$. The bottom-up labelling algorithm for an instance $(\Phi,\CKS)$ thus consists in
iteratively model checking innermore subsentences of $\Phi$ in all states
of $\CKS$, marking the states where they hold with fresh atomic
propositions with which the corresponding subsentences are replaced in
$\Phi$.

      \begin{proposition}
        \label{prop-compl-QCTL-refined}
        The model-checking problem for \QCTLsih formulas of
        simulation number at most $k$  is  \kEXPTIME[(k+1)]-complete. 
      \end{proposition}




\section{Model-checking hierarchical instances of {\SLi}}
\label{sec-modelcheckingSL}

In this section we establish that the model-checking problem for \SLi
restricted to the class of hierarchical instances is decidable
(Theorem~\ref{theo-SLi}).

\subsection{Reduction to \QCTLsi}
\label{sec-reduction}

We build upon the proof in
\cite{DBLP:journals/iandc/LaroussinieM15} that
establishes the decidability of the model-checking problem for \ATLssc
by reduction to the model-checking problem for \QCTLs. The main difference is that
we reduce to the model-checking problem for \QCTLsi instead, using
quantifiers on atomic propositions parameterised with observations that reflect the ones used
in the \SLi model-checking instance. 

Let $(\CGSi,\Phi)$ be a hierarchical instance of the \SLi
model-checking problem, and assume without loss of generality that
each strategy variable is quantified at most once in $\Phi$. We define an equivalent instance of the
model-checking problem for \QCTLsih.

\halfline
 \head{Constructing the \CKS $\CKS_{\CGSi}$} We define
 $\CKS_{\CGSi}$ so that (indistinguishable) nodes in its
 tree-unfolding correspond to (indistinguishable) finite plays in
 $\CGSi$.  The \CKS will make use of atomic
 propositions $\APv\egdef\{p_{\pos}\mid\pos\in\setpos\}$ (that we
 assume to be disjoint from $\APf$).  The idea is that $p_{\pos}$
 allows the \QCTLsi formula $ \tr[\emptyset]{\Phi}$ to refer to the current position $\pos$ in
 $\CGSi$.
 Later we will see that $\tr[\emptyset]{\Phi}$ will also make use of
 atomic propositions $\APm\egdef\{p_{\mov}^{\var}\mid\mov\in\setmoves
 \mbox{ and }\var \in \Varf\}$ that we assume, again, are disjoint
 from $\APf \cup \APv$. This allows the formula to use
 $p_{\mov}^{\var}$ to refer to the actions $\mov$ {advised by
   strategies $x$. }

 Suppose $\Obsf=\{\obs_{1},\ldots,\obs_{n}\}$, and let
$\CGSi=(\Act,\setpos,\trans,\val,\pos_\init,\obsint)$. For $i \in [n]$, define
 the local states
 $\setlstates_{i}\egdef\{\eqc{\obs_{i}}\mid\pos\in\setpos\}$ where
 $\eqc{\obs}$ is the equivalence class of $\pos$ for relation
 $\obseq$. Since we need to know the actual position of the \CGSi to
 define the dynamics, we also let $\setlstates_{n+1}\egdef\setpos$.

Define the \CKS $\CKS_{\CGSi}\egdef(\setstates,\relation,\sstate_{\init},\lab')$ where
\begin{itemize}
\item $\setstates\egdef\{\sstate_{\pos} \mid \pos\in\setpos\}$,
\item $\relation\egdef\{(\sstate_{\pos},\sstate_{\pos'})\mid
  \exists\jmov\in\Mov^{\Agf} \mbox{ s.t. }\trans(\pos,\jmov)=\pos'\}
  \subseteq \setstates^2$,
  \item $\sstate_{\init}\egdef\sstate_{\pos_{\init}}$,
\item $\lab'(\sstate_{\pos})\egdef\val(\pos)\union \{p_{\pos}\} \subseteq \APf \cup \APv$,
\end{itemize}
and $\sstate_{\pos}\egdef(\eqc{\obs_{1}},\ldots,\eqc{\obs_{n}},\pos) \in \prod_{i\in [n+1]}\setlstates_{i}$.


For every finite play $\fplay=\pos_{0}\ldots\pos_{k}$, define
the node $\noeud_{\fplay}\egdef \sstate_{\pos_{0}}\ldots \sstate_{\pos_{k}}$ in
$\unfold[\CKS_{\CGSi}]{\sstate_{\pos_{0}}}$ (which exists, by definition of
$\CKS_{\CGSi}$ and of tree unfoldings).  Note that the mapping
$\fplay\mapsto\noeud_{\fplay}$ defines a bijection between the set
of finite plays and the set of
nodes in $\unfold[\CKS_{\CGSi}]{\sstate_{\init}}$.

\halfline
\head{Constructing the \QCTLsih formulas $\tr[f]{\phi}$}
 We now describe how to transform an \SLi formula $\phi$ and a partial
function $f:\Agf \partialto  \Varf$ into a \QCTLsi
formula $\tr[f]{\phi}$ (that will also depend on $\CGSi$).
Suppose that $\Mov=\{\mov_{1},\ldots,\mov_{\maxmov}\}$, and define
$\tr[f]{\phi}$ and $\trp[f]{\psi}$ by mutual induction on state and path formulas. 
The base cases are as follows:
$\tr[f]{p} 		 \egdef p$ and $\trp[f]{\phi} \egdef
\tr[f]{\phi}$. Boolean and temporal operators are simply obtained by
distributing the translation:
$\tr[f]{\neg \phi} 	 \egdef \neg \tr[f]{\phi}$, $\trp[f]{\neg
  \psi} \egdef \neg \trp[f]{\psi}$,
$\tr[f]{\phi_1\ou\phi_2}  \egdef \tr[f]{\phi_1}\ou\tr[f]{\phi_2}$,
$\trp[f]{\psi_1\ou\psi_2}  \egdef \trp[f]{\psi_1}\ou\trp[f]{\psi_2}$,
$\trp[f]{\X\psi}  \egdef \X\trp[f]{\psi}$ and $\trp[f]{\psi_1\until\psi_2}  \egdef \trp[f]{\psi_1}\until\trp[f]{\psi_2}$.

We continue with the case of the strategy quantifier:
\[
  \begin{array}{lrl}
& \tr[f]{\Estrato{\obs}\phi}	& \egdef  \exists^{\trobs{\obs}}
                                  p_{\mov_{1}}^{\var}\ldots
                                  \exists^{\trobs{\obs}}
                                  p_{\mov_{\maxmov}}^{\var}. \phistrat
                                  \et \tr[f]{\phi}\\[5pt]
  \mbox{where} &  \phistrat & \egdef \A\always
                       \bigou_{\mov\in\Mov}p_{\mov}^{\var}\\[5pt]
  \mbox{and} & \trobs{\obs_i} & \egdef \{j\mid
\obsint(\obs_{i})\subseteq\obsint(\obs_{j})\}.
\end{array}
\]

The intuition is that for each possible action $\mov\in\Mov$, an
existential quantification on the atomic proposition $p_{\mov}^{\var}$
``chooses'' for each  node $\noeud_{\fplay}$ of the tree 
$\unfold[\CKS_{\CGSi}]{\sstate_{{\pos_0}}}$ whether strategy $\var$
allows action $\mov$ in $\fplay$ or not, and it does so uniformly with
regards to observation $\trobs{\obs}$. 
$\phistrat$  checks that at least one action is allowed in each
node, and thus that atomic propositions
$p_{\mov}^{\var}$ indeed define a
strategy.

We define $\trobs{\obs_i}$ as $\{j\mid
\obsint(\obs_{i})\subseteq\obsint(\obs_{j})\}$
instead of $\{i\}$ in order to obtain a hierarchical instance. 
Note that including all coarser observations does not increase the
information accessible to the quantifier: indeed, 
two nodes are $\{i\}$-indistinguishable if and only if they are
$\trobs{\obs_{i}}$-indistinguishable.

Here are the remaining cases:
\[
\begin{array}{lrl}
& \tr[f]{\bind{\var}\phi}	& \egdef \tr[{f[\ag\mapsto \var]}]{\phi} \quad\quad
                          \text{for }\var\in\Varf\union\{\unb\}  \\[5pt]
\mbox{and} & \bas{\tr[f]{\Eout\psi}}	& \bas{\egdef \E\,(\psiout[f] \wedge \trp[f]{\psi})}\\[5pt]
\mbox{where} &
\psiout[f] & \egdef \always
  \bigou_{\pos\in\setpos}\left ( p_{\pos} \et \bigou_{\jmov\in\Mov^{\Agf}} 
  \biget_{\ag\in\dom(f)}p_{\jmov_{\ag}}^{f(\ag)}\et \X\,
                          p_{\trans(\pos,\jmov)}\right ).
\end{array}
\]

$\psiout[f]$ checks that  each player $\ag$ in the domain of $f$
follows the strategy coded by the $p_\act^{f(\ag)}$. 

\begin{remark}
  \label{rem-tr-deterministic}
  If we consider the fragment of \SLi that only allows for
  deterministic strategies, the translation can be adapted by simply
  replacing formula $\phistrat$ above with its deterministic variant
 \[\phistratdet \egdef
\A\always\bigou_{\mov\in\Mov}(p_{\mov}^{\var}\et\biget_{\mov'\neq\mov}\neg
p_{\mov'}^{\var}),\]
which ensures that \emph{exactly one} action is chosen for strategy $\var$  in each finite
play, and thus that atomic propositions
$p_{\mov}^{\var}$  characterise a
deterministic strategy.
\end{remark}
   
To prove correctness of the translation, given a strategy $\strat$ and
a strategy variable $\var$ we
let  $\stratlab{\var}\egdef\{\plab[{p_\act^\var}]\mid
\act\in\Act\}$ be the family of $p_\act^\var$-\labelings for tree
$\unfold[\CKS_{\CGSi}]{}$ defined as follows: for each
finite play $\fplay$ in $\CGSi$ and $\act\in\Act$,
we let $\plab[{p_\act^\var}](\noeud_\fplay)\egdef1$ if $\act\in\strat(\fplay)$, 0 otherwise.
For a \labeled tree $\ltree$ with same domain as
$\unfold[\CKS_{\CGSi}]{}$ we write $\ltree\prodlab \stratlab{\var}$ for
$\ltree\prodlab \plab[{p_{\act_1}^\var}]\prodlab\ldots\prodlab \plab[{p_{\act_\maxmov}^\var}]$.

Given an infinite play $\iplay$ and a point $i\in\setn$, we also let
$\tpath_{\iplay,i}$ be the infinite path in
$\unfold[\CKS_{\CGSi}]{\sstate_{\pos_\init}}$ that starts in node
$\noeud_{\iplay_{\leq i}}$ and is defined as
$\tpath_{\iplay,i}\egdef\noeud_{\iplay_{\leq i}}\noeud_{\iplay_{\leq
    i+1}}\noeud_{\iplay_{\leq i+2}}\ldots$

Finally, for an assignment $\assign$ and
a partial   function $f:\Agf\partialto\Varf$, we say that $f$ is
\emph{compatible} with $\assign$ if
 $\dom(\assign)\inter \Agf=\dom(f)$ and  for all $a \in \dom(f)$,  $\assign(a) = \assign(f(a))$.

 \begin{proposition}
   \label{prop-redux}
   For every  state subformula $\phi$ and path subformula $\psi$ of
   $\Phi$, finite play $\fplay$, infinite play $\iplay$, point
   $i\in\setn$, for every  assignment $\assign$ variable-complete for
   $\phi$ (resp. $\psi$) and
partial   function $f:\Agf\partialto\Varf$ compatible with $\assign$, assuming
   also that no $\var_i$ in $\dom(\assign)\inter \Varf=\{\var_1,\ldots,\var_k\}$ is
 quantified in $\phi$ or $\psi$, we have
   \begin{align*}
\CGSi,\assign,{\fplay}\models\phi && \mbox{ if and only if } &&
  \unfold[\CKS_{\CGSi}]{\sstate_{\pos_\init}}\prodlab
  \stratlab[\assign(\var_1)]{\var_1}\prodlab\ldots \prodlab
  \stratlab[\assign(\var_k)]{\var_k},\noeud_{\fplay} \modelst
                                                              \tr[f]{\phi}\\
\CGSi,\assign,{\iplay},i\models\psi && \mbox{ if and only if } &&
  \unfold[\CKS_{\CGSi}]{\sstate_{\pos_\init}}\prodlab
  \stratlab[\assign(\var_1)]{\var_1}\prodlab\ldots \prodlab
  \stratlab[\assign(\var_k)]{\var_k},\tpath_{\iplay,i} \modelst
  \trp[f]{\psi}     
   \end{align*}

  In addition, $\CKS_{\CGSi}$ is of size linear in
$|\CGSi|$, and $\tr[f]{\phi}$ and $\trp[f]{\psi}$ are of size linear in $|\CGSi|^2+|\phi|$.
 \end{proposition}

 \begin{proof}
   The proof is by induction on $\phi$.
We detail the cases for binding,  strategy quantification and
outcome quantification, the others follow simply by definition of
$\CKS_{\CGSi}$ for atomic propositions and induction hypothesis for
remaining cases.

\halfline
For $\phi=\bind{\var}\phi'$, we have
$\CGSi,\assign,{\fplay}\models\bind{\var}\phi'$ if and only if 
$\CGSi,\assign[\ag\mapsto\assign(\var)],{\fplay}\models\phi'$.
The result follows by using the induction hypothesis with assignment
$\assign[\ag\mapsto\var]$ and function
 $f[a\mapsto\var]$. This is possible because $f[a\mapsto\var]$ is compatible with $\assign[\ag\mapsto\var]$: indeed
 $\dom(\assign[\ag\mapsto\var])\inter\Agf$ is equal to
 $\dom(\assign)\inter\Agf \union \{a\}$ which, by assumption, is equal
 to $\dom(f) \union \{a\}=\dom(f[a\mapsto
 \var])$. Also
 by assumption, for all $a'\in\dom(f)$, $\assign(a')=\assign(f(a'))$, and 
by definition $\assign[a\mapsto \assign(\var)](a)=\assign(\var)=\assign(f[a\mapsto\var](a))$.

\halfline
For $\phi=\Estrato{\obs}\phi'$, assume first that
$\CGSi,\assign,{\fplay}\models\Estrato{\obs}\phi'$. There exists an
$\obs$-uniform strategy $\strat$ such that
\[\CGSi,\assign[\var\mapsto \strat],\fplay\models \phi'.\] Since $f$ is
compatible with $\assign$, it is also compatible with assignment
$\assign'=\assign[\var\mapsto \strat]$. By assumption, no variable in
$\{\var_1,\ldots,\var_k\}$ is quantified in $\phi$, so that $\var\neq
\var_i$ for all $i$, and thus $\assign'(\var_i)=\assign(\var_i)$ for
all $i$; and because no strategy variable is
quantified twice in a same formula,
$\var$ is not quantified in $\phi'$, so that no variable in
$\{\var_1,\ldots,\var_k,\var\}$ is quantified in $\phi'$.
By induction hypothesis 
  \[\unfold[\CKS_{\CGSi}]{\sstate_{\pos_\init}}\prodlab
  \stratlab[\assign'(\var_1)]{\var_1}\prodlab\ldots \prodlab
  \stratlab[\assign'(\var_k)]{\var_k}\prodlab   \stratlab[\assign'(\var)]{\var},\noeud_{\fplay}
  \modelst \tr[f]{\phi'}.\]

  Because $\strat$ is $\obs$-uniform, each
  $\plab[{p_\act^\var}]\in\stratlab{\var}=\stratlab[\assign'(\var)]{\var}$ is $\trobs{\obs}$-uniform,
  and it follows that   \[\unfold[\CKS_{\CGSi}]{\sstate_{\pos_\init}}\prodlab
  \stratlab[\assign'(\var_1)]{\var_1}\prodlab\ldots \prodlab
  \stratlab[\assign'(\var_k)]{\var_k},\noeud_{\fplay}
  \modelst \exists^{\trobs{\obs}} p_{\mov_{1}}^{\var}\ldots
  \exists^{\trobs{\obs}}
  p_{\mov_{\maxmov}}^{\var}. \phistrat\et\tr[f]{\phi'}.\]
Finally, since $\assign'(\var_i)=\assign(\var_i)$ for all $i$, we
conclude that 
\[\unfold[\CKS_{\CGSi}]{\sstate_{\pos_\init}}\prodlab
  \stratlab[\assign(\var_1)]{\var_1}\prodlab\ldots \prodlab
  \stratlab[\assign(\var_k)]{\var_k},\noeud_{\fplay}
  \modelst \tr[f]{\Estrato{\obs}\phi'}.\]

For the other direction, assume
that
\[\unfold[\CKS_{\CGSi}]{\sstate_{\pos_\init}}\prodlab
  \stratlab[\assign(\var_1)]{\var_1}\prodlab\ldots \prodlab
  \stratlab[\assign(\var_k)]{\var_k},\noeud_{\fplay} \modelst
  \tr[f]{\phi},\] and recall that
$\tr[f]{\phi}=\exists^{\trobs{\obs}} p_{\mov_{1}}^{\var}\ldots
\exists^{\trobs{\obs}}
p_{\mov_{\maxmov}}^{\var}. \phistrat\et\tr[f]{\phi'}$.  Write
$\ltree=\unfold[\CKS_{\CGSi}]{\sstate_{\pos_\init}}\prodlab
\stratlab[\assign(\var_1)]{\var_1}\prodlab\ldots \prodlab
\stratlab[\assign(\var_k)]{\var_k}$. There exist
$\trobs{\obs}$-uniform $\plab[{p_\act^\var}]$-\labelings such that
\[\ltree\prodlab  \plab[{p_{\act_1}^\var}]\prodlab\ldots\prodlab \plab[{p_{\act_\maxmov}^\var}]
  \modelst \phistrat\et\tr[f]{\phi'}.\]
By $\phistrat$, these \labelings  
 code for a strategy $\strat$, and because they are
 $\trobs{\obs}$-uniform, $\strat$ is $\obs$-uniform. Let
 $\assign'=\assign[\var\mapsto \strat]$. For all $1\leq i\leq k$, by
 assumption $\var\neq \var_i$, and thus $\assign'(\var_i)=\assign(\var_i)$.
 The above can thus be rewritten
 \[\unfold[\CKS_{\CGSi}]{\sstate_{\pos_\init}}\prodlab
\stratlab[\assign'(\var_1)]{\var_1}\prodlab\ldots \prodlab
\stratlab[\assign'(\var_k)]{\var_k}\prodlab  \stratlab[\assign'(\var)]{\var}
  \modelst \phistrat\et\tr[f]{\phi'}.\]
 By induction hypothesis we have
$\CGSi,\assign[\var\mapsto
\strat],\fplay\models \phi'$, hence $\CGSi,\assign,\fplay\models
\Estrato{\obs}\phi'$.

\halfline
For $\phi=\Eout\psi$,
assume first that $\CGSi,\assign,{\fplay}\models\E\psi$. 
There exists a play $\iplay\in\out(\assign,\fplay)$ such that
$\CGSi,\assign,\iplay,|\fplay|-1\modelsSL \psi$. By induction
hypothesis,
$\unfold[\CKS_{\CGSi}]{\sstate_{\pos_\init}}\prodlab
  \stratlab[\assign(\var_1)]{\var_1}\prodlab\ldots \prodlab
  \stratlab[\assign(\var_k)]{\var_k},\tpath_{\iplay,|\fplay|-1} \modelst
  \trp[f]{\psi}$. Since $\iplay$ is an outcome of $\assign$, each agent $a\in\dom(\assign)\inter\Agf$ 
follows strategy $\assign(a)$ in $\iplay$.
Because  $\dom(\assign)\inter \Agf=\dom(f)$ and for all $a \in \dom(f)$,
  $\assign(a) = \assign(f(a))$, each agent $a\in\dom(f)$ follows
the  strategy $\assign(f(a))$, which is coded by atoms
$p_\mov^{f(\ag)}$ in the translation of $\Phi$. Therefore $\tpath_{\iplay,|\fplay|-1}$ also
satisfies $\psiout$, hence $\unfold[\CKS_{\CGSi}]{\sstate_{\pos_\init}}\prodlab
  \stratlab[\assign(\var_1)]{\var_1}\prodlab\ldots \prodlab
  \stratlab[\assign(\var_k)]{\var_k},\tpath_{\iplay,|\fplay|-1} \modelst
  \psiout \et   \trp[f]{\psi}$, and we are done.

  For the other direction, assume that 
  $\unfold[\CKS_{\CGSi}]{\sstate_{\pos_\init}}\prodlab
  \stratlab[\assign(\var_1)]{\var_1}\prodlab\ldots \prodlab
  \stratlab[\assign(\var_k)]{\var_k},\noeud_\fplay \modelst
  \E(\psiout[f] \et   \trp[f]{\psi})$.
There exists a path $\tpath$ in $\unfold[\CKS_{\CGSi}]{\sstate_{\pos_\init}}\prodlab
  \stratlab[\assign(\var_1)]{\var_1}\prodlab\ldots \prodlab
  \stratlab[\assign(\var_k)]{\var_k}$ starting in
node $\noeud_\fplay$ that satisfies both $\psiout[f]$ and $\trp[f]{\psi}$.
By construction of $\CKS_{\CGSi}$ there exists an infinite play $\iplay$
such that $\iplay_{\leq |\fplay|-1}=\fplay$ and $\tpath=\tpath_{\iplay,|\fplay|-1}$.
By induction hypothesis, $\CGSi,\assign,\iplay,|\fplay|-1 \modelsSL \psi$.
Because $\tpath_{\iplay,|\fplay|-1}$ satisfies $\psiout[f]$, $\dom(\assign)\inter \Agf=\dom(f)$, and for all $a \in \dom(f)$,
  $\assign(a) = \assign(f(a))$, it is also the case that
  $\iplay\in\out(\assign,\fplay)$, 
hence  $\CGSi,\assign,\fplay \modelsSL \Eout\psi$.

The size of $\CKS_{\CGSi}$, $\tr[f]{\phi}$ and $\trp[f]{\psi}$ are easily verified.
 \end{proof}
 
 Applying Proposition~\ref{prop-redux} to the sentence $\Phi$, $\fplay=\pos_\init$, any assignment $\assign$, and
 the empty function $\emptyset$, we get:
 \[\CGSi \models \Phi \quad \mbox{if and only if}\quad
\unfold[\CKS_{\CGSi}]{s_{\pos_{\init}}} \models
 \tr[\emptyset]{\Phi}.\]

\head{Preserving hierarchy}
To complete the proof of Theorem~\ref{theo-SLi} it remains to check that $\tr[\emptyset]{\Phi}$ is a
hierarchical \QCTLsi formula, which is the case because
 $\Phi$ is hierarchical in
$\CGSi$ and for every two observations $\obs_{i}$ and $\obs_{j}$ in $\Obsf$ such that
$\obsint(\obs_{i})\subseteq\obsint(\obs_{j})$, by definition of $\trobs{\obs_{k}}$
we have that $\trobs{\obs_{i}}\subseteq \trobs{\obs_{j}}$.

\subsection{Complexity}
\label{sec-complexity-SL}

We now establish the complexity of model checking hierarchical
instances of \SLi. As we did for \QCTLsi, we first define the
simulation depth of \SLi state formulas. In the following inductive
definition,  $\obsint_\phi$  denotes the intersection of all
indistinguishability relations used in $\phi$:
$\obsint_\phi\egdef\inter_{\obs\in\phi} \obsint(\obs)$, with the
empty intersection being defined as the identity relation (perfect information). Also, for a
path formula $\psi$, $\max(\psi)$ is the set of maximal state
subformulas in $\psi$.

      \[
        \begin{array}{lcc}%
          \ndd(p) \egdef (0,\nd)& \hspace{5cm} & \ndd(\neg \phi) \egdef
                                  (\ndd_k(\phi),\alt)\\[7pt]
          \multicolumn{3}{l}{
          \ndd(\phi_1\ou\phi_2) \egdef \left
          (\max_{i\in\{1,2\}}\ndd_k(\phi_i), x \right ),}\\[5pt]
          \multicolumn{3}{r}{
          \mbox{where
          }x=    \begin{cases}
            \nd & \mbox{if }\ndd_x(\phi_1)=\ndd_x(\phi_2)=\nd\\
            \alt & \mbox{otherwise}
          \end{cases}
          }\\[17pt]
          \multicolumn{3}{l}{
          \ndd( \Estrato{\obs}\phi)\egdef (k,\nd),}\\[5pt]
          \multicolumn{3}{r}{
          \mbox{where }
          k=\begin{cases}
            \ndd_k(\phi) & \mbox{if }\ndd_x(\phi)=\nd \mbox{ and
            }\obsint(\obs)=\obsint_\phi\\
            \ndd_k(\phi)+1 & \mbox{otherwise}
          \end{cases}
                             }\\[7pt]
          \ndd(\bind{\var}\phi) \egdef \ndd(\phi)\\[7pt]
           \multicolumn{3}{l}{
          \bas{\ndd(\Eout\psi)\egdef
          \begin{cases}
            (0,\nd) &\mbox{if }\psi\in\LTL\\
            (\max_{\phi\in\max(\psi)}\ndd_k(\phi),\alt) &\mbox{otherwise}
          \end{cases}
                                                          }}
        \end{array}
      \]

      \begin{proposition}
        \label{prop-compl-SLi}
The model-checking problem for hierarchical instances of \SLi of
        simulation depth at most $k$  is  \kEXPTIME[(k+1)]-complete. 
      \end{proposition}

      \begin{proof}
        The upper bounds follow from the fact that the translated
        formulas in our reduction have essentially the same simulation
        depth as the original ones. However this is not quite right,
        because in the case where $\ndd_x(\phi)=\nd$ and
            $\obsint(\obs)=\obsint_\phi$ we have
            $\ndd(\Estrato{\obs}\phi)=(\ndd_k(\phi),\nd)$, while
            $\ndd(\tr[f]{\Estrato{\obs}\phi})=(\ndd_k(\tr[f]{\phi})+1,\nd)$:
            indeed, while it is the case that
      $\obsint(\obs)=\obsint_\phi$ implies that
      $\trobs{\obs}=\Iphi[{\tr[f]{\phi}}]$, the translation introduces
      a conjunction with $\phistrat$, and even when
      $\ndd_x(\tr[f]{\phi})=\nd$, we have $\ndd_x(\phistrat\et\tr[f]{\phi})=\alt$.
According to Proposition~\ref{prop-compl-QCTL}, this should thus
induce an additional exponential to check the translated
formula. However, this can be avoided by noticing that the fixed
formula $\phistrat= \A\always \bigou_{\mov\in\Mov}p_{\mov}^{\var}$  can be checked by a
simple \emph{deterministic} tree automaton with two states
$\tq_{\text{check}}$ and $\tq_{\text{rej}}$: the automaton starts in
state $\tq_{\text{check}}$, which is accepting (it has parity zero); when it visits a node $\noeud$ in state $\tq_{\text{check}}$,
if $\lab(\noeud)$ satisfies $\bigou_{\mov\in\Mov}p_{\mov}^{\var}$,
then the automaton sends state $\tq_{\text{check}}$ to all children of
$\noeud$,
otherwise it sends the  state $\tq_{\text{rej}}$ to all
children. State $\tq_{\text{rej}}$ is rejecting (it has parity one) and is a sink:
it sends itself to all children, independently on the label of the
visited node.
If we restrict \SLi to deterministic strategies, the same observation can be made:
the automaton that checks formula $\phistratdet =
\A\always\bigou_{\mov\in\Mov}(p_{\mov}^{\var}\et\biget_{\mov'\neq\mov}\neg
p_{\mov'}^{\var})$ 
is the same as the one described above, except that it checks whether
$\bigou_{\mov\in\Mov}(p_{\mov}^{\var}\et\biget_{\mov'\neq\mov}\neg
p_{\mov'}^{\var})$ is satisfied by the label of the current node.

Given two 
tree automata $\ATA_1$ and $\ATA_2$, one deterministic and one
nondeterministic, one can easily build a nondeterministic automaton
$\ATA_1\inter\ATA_2$ of size $|\ATA_1|\times |\ATA_2|$ that accepts
the intersection of their languages, so that in this case the
conjunction does not introduce alternation, and thus we do not need an
additional simulation before projecting to guess the strategy.
We could refine the notion of simulation depth to reflect this, but we
find that it would become very cumbersome for little added benefit, so
we keep this observation in this proof.

      The lower bounds are inherited from \QCTLsi thanks to the
polynomial      reduction presented in Section~\ref{sec-QCTL-to-SL},
which preserves simulation depth.
      \end{proof}
      
We point out that all instances of the model-checking problem for the
perfect-information fragment are hierarchical, and thus  this result
provides improved upper-bounds for \SL, which was
only known to be in
\kEXPTIME[k] for formulas of length at most $k$~\cite{DBLP:journals/tocl/MogaveroMPV14}.
Also the lower bounds for \QCTLsi are inherited directly from the
perfect-information fragment \QCTLs,
which reduces to the perfect-information fragment of \SLi following
the construction from Section~\ref{sec-QCTL-to-SL}.
Therefore the lower bounds  hold already for the perfect-information fragment  of \SLi.
Note however that this does  not provide lower bounds for the usual, linear-time variant
of Strategy Logic, where path quantifiers in \QCTLs formulas must be
simulated with strategy quantifications which increase the
simulation depth of the resulting Strategy Logic formulas.
The exact complexity of the linear-time variant is not known, even in
the perfect-information case. 


\halfline
\head{Simulation number} The intuition behind the alternation number
as considered in~\cite{DBLP:journals/tocl/MogaveroMPV14} is to refine the  classic alternation depth between existential
and universal quantifiers by observing that subsentences of a sentence $\Phi$
to model-check can be treated independently thanks to a
bottom-up labelling algorithm: innermost sentences are evaluated in
all states of the model and replaced in $\Phi$ by atomic propositions
that label the states where they hold. The alternation number of
$\Phi$ is the maximum alternation depth of the successive subsentences
that are treated by this bottom-up procedure, and it determines the
complexity of the overall model-checking procedure.

However, as discussed in Remark~\ref{rem-sentences}, the semantics of
the outcome quantifier  makes
sentences sensitive to the assignment in which they are evaluated. As
a result, to define the notion of alternation number in our setting,
we introduce a notion of \emph{independent subsentence}. Intuitively,
a subsentence $\phi$ of a sentence $\Phi$ is \emph{independent} if 
it redefines or unbinds the strategies of all players who are bound to
a strategy when $\phi$ is reached in the evaluation of $\Phi$.
More precisely, we say that an agent $\ag$ is \emph{bound} in a syntactic
subformula $\phi$ of $\Phi$ if the path that
leads to $\phi$ in $\Phi$'s syntactic tree  contains a binding
operator $\bind{\var}$ for $\ag$ which is not followed by an unbinding
$\unbind$ for her. A subsentence $\phi$ of $\Phi$ is
\emph{independent} if all agents that are bound in $\phi$ are either
rebound by an operator  $\bind{\var}$ or unbound by an operator $\unbind$ before any outcome
quantifier is met in $\phi$. In an independent subsentence $\phi$, the semantics
of the outcome quantifier does not depend on strategies that are
quantified outside $\phi$, and in fact a subsentence $\phi$ of $\Phi$ is
  independent if and only if the formula that corresponds to $\phi$ in
  $\tr[\emptyset]{\Phi}$ is a subsentence of $\tr[\emptyset]{\Phi}$. 

Similarly to what we did for \QCTLsi we now define the
\emph{simulation number} $\ndn(\Phi)$ of an \SLi sentence $\Phi$
as the maximum of the simulation depths for independent
subsentences, where strict independent subsentences are counted as atoms.

\begin{lemma}
  \label{lem-corres-number}
          For every hierarchical instance $(\CGSi,\Phi)$ of \SLi, 
$\ndn(\Phi)=\ndn(\tr[\emptyset]{\Phi})$.
\end{lemma}

The following then follows from Proposition~\ref{prop-redux}, 
Lemma~\ref{lem-corres-number} and
Proposition~\ref{prop-compl-QCTL-refined}.

\begin{proposition}
  \label{prop-compl-SLi-refined}
  The model-checking problem for hierarchical instances of \SLi of
        simulation number at most $k$  is  \kEXPTIME[(k+1)]-complete. 
\end{proposition}

We now compare the latter result with the complexity of model checking
\SLNG, the
nested goal fragment of Strategy Logic with perfect information (we
refer the interested reader to~\cite{DBLP:journals/tocl/MogaveroMPV14}
for a definition of this fragment). It is
established in~\cite{DBLP:journals/iandc/ChatterjeeHP10,DBLP:journals/tocl/MogaveroMPV14} that this
problem is in \kEXPTIME[(k+1)] for formulas of \emph{alternation number}
$k$. We remark that the simulation number of an \SLNG
formula translated in our branching-time version of \SL (this is done by adding outcome
  quantifiers between bindings and temporal operators) is equal to its
alternation number plus one, and thus
Proposition~\ref{prop-compl-SLi-refined} gives a \kEXPTIME[(k+2)]
upper bound for \SLNG formulas of alternation number
$k$. In~\cite{DBLP:journals/iandc/ChatterjeeHP10,DBLP:journals/tocl/MogaveroMPV14} the extra exponential
is avoided by resorting to universal and nondeterministic tree
automata, depending on whether the innermost strategy quantification
is existential or universal, to deal with temporal formulas.
Thus, the innermost strategy quantification can be dealt with without
incurring an exponential blowup.

The same thing cannot be done for
\SLi, for two reasons. The first one is that in general the innermost strategy quantification may have
imperfect information and thus require a narrowing of the automaton;
this operation introduces alternation, which needs to be removed
at the cost of one exponential before dealing with strategy quantification. The second reason is that
even when the innermost strategy has perfect information, the outcome
quantifier that we introduce in Strategy Logic allows the expression
of \CTLs formulas which cannot be dealt with by nondeterministic and
universal automata as is done in~\cite{DBLP:journals/iandc/ChatterjeeHP10,DBLP:journals/tocl/MogaveroMPV14}.



\section{Comparison with related logics}
\label{sec-SLi-comparison}


In this section we first show that \SLi subsumes \SL and the main
imperfect-information extensions of \ATL. Then we show that model checking
Coordination Logic (\CL) reduces to model checking hierarchical
instances of \SLi  where the
truth of all atomic propositions in the model is known by all agents
(or more precisely, all observations in the concurrent game structures
are fine enough to observe the truth value of all atomic propositions).

\subsection{Comparison with  \ATL}
\label{subsec:SLi-comparison}

The main difference between \SL and \ATL-like strategic logics is that
in the latter a strategy is always bound to some player, while in the
former bindings and quantifications are separated. This separation
adds expressive power, e.g., one can bind the same strategy to
different players. Extending \ATL with imperfect-information is done
by giving each player an indistinguishability relation that its
strategies must respect~\cite{BJ14}.  In \SLi
instead each strategy $x$ is assigned an
indistinguishability relation $o$ when it is quantified.  Associating
observations to strategies rather than players
allows us to obtain a logic \SLi that is a clean generalisation of
(perfect-information) \SL,
and subsumes imperfect-information extensions of \ATLs that
associate observations to players.
Concerning \SL, it is rather easy to see that every sentence in \SL has an
equivalent in the fragment of \SLi with deterministic strategies where all observation symbols are
interpreted as perfect information.
We now prove that \SLi also subsumes \ATLs with imperfect information.

\begin{proposition}
  \label{prop-subsume-ATL}
For every \ATLsi formula\footnote{See~\cite{BJ14} for the definition of \ATLsi, where subscript i
refers to ``imperfect information'' and subscript R to ``perfect
recall''. Also, we consider the so-called \emph{objective
    semantics} for \ATLsi.} $\phi$ there is an \SLi formula $\phi'$
  such that for every \CGSi $\CGSi$ there is a \CGSi $\CGSi'$ such that $\CGSi
  \models \phi$ if, and only if, $\CGSi' \models \phi'$.
\end{proposition}

We recall that an \ATLsi formula $\EstratATL \psi$  reads as ``there are
strategies for players in $\coal$ such that $\psi$ holds whatever
players in $\Agf\setminus\coal$ do''.
Formula $\phi'$ is built from $\phi$ by replacing each subformula of
the form $\EstratATL \psi$, where
$\coal=\{\ag_{1},\ldots,\ag_{k}\}\subset\Agf$ is a coalition of
players and $\Agf\setminus\coal=\{\ag_{k+1},\ldots,\ag_n\}$ with formula $\Estrato[\var_{1}]{\obs_{1}}\ldots
\Estrato[\var_{k}]{\obs_{k}} 
\bind[\ag_{1}]{\var_{1}}\ldots \bind[\ag_{k}]{\var_{k}}\bind[\ag_{k+1}]{\unb}\ldots\bind[\ag_n]{\unb}\Aout\, \psi'$, where
$\psi'$ is the translation of $\psi$. Then $\CGSi'$ is obtained from
$\CGSi$ by interpreting each $\obs_{i}$ as the equivalence relation
for player~$i$ in $\CGSi$, and interpreting $\obs_{p}$ as the identity
relation. 


Third, \SLi also subsumes the imperfect-information extension of \ATLs
with strategy context
(see~\cite{DBLP:journals/corr/LaroussinieMS15} for the definition of
\ATLssc with partial observation, which we refer to as \ATLssci). 

\begin{proposition}
  \label{prop-subsume-ATLsc}
For every \ATLssci formula $\phi$ there is an \SLi formula $\phi'$
  such that for every \CGSi $\CGSi$ there is a \CGSi $\CGSi'$ such that $\CGSi
  \models \phi$ if, and only if, $\CGSi' \models \phi'$.
\end{proposition}

The only difference between \ATLssci and \ATLsi is the following: in
\ATLsi, when a subformula of the form $\EstratATL \psi$ is met, we
quantify existentially on strategies for players in $\coal$ and
quantify universally on possible outcomes obtained by letting other
players behave however they want. Therefore, if any player in
$\Agf\setminus\coal$ had previously been assigned a strategy, it is
forgotten. In \ATLssci on the other hand, these strategies are stored
in a \emph{strategy context}, which is a \emph{partial} assignment $\assign$, defined for the
subset of players currently bound to a strategy.
A strategy context allows one to quantify universally only on strategies
 of players who are not in $\coal$ and who are not already
bound to a strategy. It is then easy to adapt the translation
presented for Proposition~\ref{prop-subsume-ATL}: it suffices not to
unbind agents outside the coalition from their strategies. $\CGSi'$ is defined as for Proposition~\ref{prop-subsume-ATL}.

\subsection{Comparison with Coordination Logic}
\label{subsec:SLi-comparison-CL}

There is a natural and simple translation of instances of the model-checking problem of 
\CL\cite{DBLP:conf/csl/FinkbeinerS10} into the hierarchical instances of \SLi. 
Moreover, the image of this translation consists
of instances of \SLi with a very restricted form: atoms mentioned in the \SLi-formula 
are observable by all observations of the \CGSi, i.e., for all
$\obs\in\Obsf$ and $p\in\APf$, $v \obseq v'$
implies that $p \in \val(v)$
iff $p \in \val(v')$.

\begin{proposition} \label{prop:CLtoSLi}
There is an effective translation that, given a \CL-instance $(\CKS,\phi)$ produces a hierarchical \SLi-instance $(\CGSi,\Phi)$ such that
\begin{enumerate}
 \item $\CKS \models \phi$ if, and only if, $\CGSi \models \Phi$,
 \item For all atoms $p\in\APf$ and observations $\obs \in \Obsf$,  $v \obseq v'$ implies
   that $p \in \val(v)$ iff $p \in \val(v')$. 
\end{enumerate}
\end{proposition}

To do this, one first translates \CL into (hierarchical) \QCTLsi, the
latter is defined in the next section. This step is a simple
reflection of the semantics of \CL in that of \QCTLsi. Then one
translates \QCTLsi into \SLi by a simple adaptation of the translation
of \QCTLs into
\ATLssc~\cite{DBLP:journals/iandc/LaroussinieM15}.


We briefly recall the syntax and semantics of \CL, and refer
to~\cite{DBLP:conf/csl/FinkbeinerS10} for further details.

\halfline
\head{Notation for trees}
Note that our definition for trees (see
Section~\ref{sec-QCTLsi-semantics}) differs slightly  from the one
in~\cite{DBLP:conf/csl/FinkbeinerS10}, where the root is the empty
word. Here we adopt this convention to stay closer to notations
in~\cite{DBLP:conf/csl/FinkbeinerS10}. Thus, $(Y,X)$-trees in \CL are
of the form $(\tau,l)$ where $\tau \subseteq X^*$ and $l:\tau \to 2^Y$.

For two disjoint sets $X$ and $Y$, we
identify $2^{X}\times 2^{Y}$ with $2^{X\union Y}$.
Let $X$ and $Y$ be  two sets with $Z=X\union Y$, and let  $M$
and $N$ be two disjoint sets.
Given an ${M}$-labelled $2^{Z}$-tree $\ltree=(\tree,\lab_{M})$ and an
${N}$-labelled $2^{Z}$-tree $\ltree'=(\tree',\lab_{N})$ with same
domain $\tree=\tree'$, we define $\ltree\dunion\ltree'\egdef(\tau,\lab')$,
where for every $\noeud\in\tree$, $\lab'(\noeud)=\lab_{M}(\noeud)\union\lab_{N}(\noeud)$.
Now, given a complete ${M}$-labelled
$2^{X}$-tree $\ltree=((2^{X})^{*},\lab_{M})$ and a complete ${N}$-labelled
$2^{Y}$-tree $\ltree'=((2^{Y})^{*},\lab_{N})$, we define $\ltree \oplus
\ltree'\egdef\liftI[2^{Z\setminus X}]{}{\ltree}\dunion\,
\liftI[2^{Z\setminus Y}]{}{\ltree'}$.

\halfline
\head{\CL Syntax}
Let $\coordvar$ be a set of \emph{coordination variables}, and let
$\stratvar$ be a set of \emph{strategy variables} disjoint from $\coordvar$. 
The syntax of \CL is given by the following grammar:
\[\phi ::= x \mid \neg \phi \mid \phi\ou \phi \mid \X \phi  \mid
\phi \until \phi \mid \strquantif C \exists s.\, \phi \]
where $x\in \coordvar\union\stratvar$, $C\subseteq\coordvar$ and
$s\in\stratvar$, and with the restriction that each coordination 
variable appears in at most one \emph{subtree quantifier} $\strquantif C \exists s.\,$, and similarly for strategy variables.

The notion of free and bound (coordination or strategy) variables is
as usual. The set of free coordination variables in $\phi$ is noted $\freecoord$.
A bound coordination variable $c$ is \emph{visible} to a strategy
variable $s$ if $s$ is  in the scope of the quantifier that introduces
$c$, and $\scope{s}$ is the union of the set of bound coordination variables
visible to $s$ and the set of free coordination variables (note that this union is disjoint). We will see,
in the semantics, that the meaning of a bound strategy variable $s$ is a strategy
$f_{s}:(2^{\scope{s}})^{*}\to 2^{\{s\}}$.
Free strategy variables are called \emph{atomic propositions}, and 
we denote the set of atomic propositions in $\phi$ by $\APphi$.

\halfline
\head{\CL Semantics}
%
%
A \CL formula $\phi$ is evaluated on a complete $\APphi$-labelled
$2^{\freecoord}$-tree $\ltree$. An $(\APphi,2^{\freecoord})$-tree
$\ltree=(\tree,\lab)$ satisfies a \CL formula $\phi$ if for every path $\tpath$ that
starts in the root we have $\ltree,\tpath,0\models\phi$, where the
satisfaction of a formula at position $i\geq 0$ of a path $\tpath$ is
defined inductively as follows:
\begingroup
  \addtolength{\jot}{-3pt}
\begin{alignat*}{3}
  \ltree,\tpath,i\modelst & 	\,p 			&& \mbox{ if }
  &&\quad p\in\lab(\tpath_{i})\\
  \ltree,\tpath,i\modelst & 	\,\neg \phi'		&& \mbox{ if } && \quad\ltree,\tpath,i\not\modelst \phi'\\
  \ltree,\tpath,i\modelst & 	\,\phi_{1} \ou \phi_{2}		&& \mbox{ if } &&\quad \ltree,\tpath,i \modelst \phi_{1} \mbox{ or    }\ltree,\tpath,i\modelst \phi_{2} \\
\ltree,\tpath,i\modelst & \,\X\phi' 				&& \mbox{ if } && \quad\ltree,\tpath,i+1\modelst \phi' \\ 
\ltree,\tpath,i\modelst & \,\phi_{1}\until\phi_{2} 		&&
\mbox{ if } && \quad\exists\, j\geq i \mbox{ s.t.
}\ltree,\tpath,j\modelst\phi_{2} \text{ and } \forall k \text{ s.t. }i\leq k <j,\;
\ltree,\tpath,k\modelst\phi_{1}\\
\ltree,\tpath,i\modelst & \,\strquantif C\exists s.\,\phi' \quad
&& \mbox{ if } && \quad\exists\,
f:(2^{\scope{s}})^{*}\to 2^{\{s\}} \mbox{ s.t. } \ltree_{\tpath_{i}}\oplus ((2^{\scope{s}})^{*},f)\modelst \phi',
\end{alignat*}
\endgroup
where $\ltree_{\tpath_{i}}$ is the subtree of $\ltree$ rooted in $\tpath_{i}$.

First, observe that in the last inductive case, $\ltree_{\tpath_{i}}$ being a
$2^{\freecoord}$-tree, $\ltree_{\tpath_{i}}\oplus
((2^{\scope{s}})^{*},f)$ is a $2^{\freecoord\union \scope{s}}$-tree.
By definition, $\scope{s}=\freecoord\union C=\freecoord[\phi']$. It
follows that $\freecoord\union\scope{s}=\scope{s}=\freecoord[\phi']$,
hence $\phi'$ is indeed evaluated on a $\freecoord[\phi']$-tree. 

\begin{remark}
Note that all strategies observe the value of all atomic
propositions. 
%
%
Formally, for every \CL-formula  $\phi$ of the form $\phi=\strquantif C_{1}
  \exists s_{1}.\,\ldots,\strquantif C_{i}\exists s_{i}.\,\phi'$
  evaluated on a $2^{\freecoord}$-tree $\ltree=(\tree,\lab)$, $\phi'$ is evaluated
  on a $2^{\freecoord\union C_{1}\union \ldots \union C_{i}}$-tree
  $\ltree'=(\tree',\lab')$ such that for every $p\in\APphi$, for every pair of nodes
  $\noeud,\noeud'\in\ltree'$ such that
  $\projI[2^{\freecoord}]{\noeud}=\projI[2^{\freecoord}]{\noeud'}$,
  it holds that $p\in\lab'(\noeud)$ iff $p\in\lab'(\noeud')$.
Thus, in \CL one cannot directly 
capture strategic problems  where atomic
propositions are not observable to all players. 
\end{remark}

The input to the model-checking problem for \CL consists of a \CL
formula $\phi$ and a finite representation of a
$(\APphi,2^{\freecoord})$-tree $t$. The standard assumption is to
assume $t$ is a regular tree, i.e., is the unfolding of a finite
structure.  Precisely, a \emph{finite representation} of a
$(\APphi,2^{\freecoord})$-tree $t = (\tau,\lab')$ is a structure
$\CKS = (\setstates,R,\lab,\sstate_\init)$ such that
\begin{itemize}
 \item $\setstates = 2^{\freecoord}$,
 \item $R = S \times S$, 
 \item $\lab:\setstates \to 2^{\APphi}$,
 \item $\sstate_\init \in \setstates$,
\end{itemize}
and $t = \unfold[\CKS]{\sstate_{\init}}$ is the unfolding of $\CKS$.

Thus, an \emph{instance} of the model-checking problem for \CL is a
pair $(\CKS,\Phi)$ where
 $\CKS = (\setstates,
\relation,\sstate_{\init},\lab)$ is a finite representation of an
$(\APphi,2^{\freecoord})$-tree and $\Phi$ is a \CL formula (over variables $\stratvar \cup
\coordvar$).  The \emph{model-checking problem for
  \CL} is the following decision problem: given an instance
$(\CKS,\Phi)$, return `Yes' if $\unfold[\CKS]{\sstate_{\init}}
\modelst \Phi$ and `No' otherwise.

We now describe a natural translation of \CL-instances to 
\SLi-instances. This translation:
\begin{enumerate}
\item reduces the model-checking problem of \CL to that of the hierarchical fragment of \SLi.
\item shows that \CL only produces instances in which all atoms are
  uniform with regard to all observations, {i.e., instances
    $(\CGSi,\Phi)$ such that for every $p \in \APf$ and
    $\obs \in \Obsf$, $v \obseq v'$ implies
    $p \in \val(v) \iff p \in \val(v')$. }
\end{enumerate}


We will present the translation in two steps: first from \CL-instances
into \QCTLsih-instances, and then from \QCTLsi-instances to
\SLi-instances such that \QCTLsih-instances translate to hierarchical
\SLi-instances. 

\subsubsection{Translating \CL to \QCTLsih}
Let $(\CKS,\Phi)$ be an instance of the model-checking problem for
\CL, where $\CKS=(\setstates,\relation,\lab,\sstate_{\init})$.
We will construct a \QCTLsih-instance $(\trCL{\CKS},\trCL{\phi})$ 
such that $\CKS \models \Phi$ iff $\trCL{\CKS} \models \trCL{\Phi}$.  
Let $\trCL{\APf}$ be the set of all strategy variables occurring in $\Phi$, 
let $\coordvar(\Phi)$ be the set of coordination variables that appear in $\Phi$, and assume, w.l.o.g., that
$\coordvar(\varphi) = [n]$ for some $n \in \setn$. Let $\hidden(\Phi)\egdef\coordvar(\Phi)\setminus\freecoord$. 

First, we define the \CKS $\trCL{\CKS}$ over $\trCL{\APf}$: the idea is to add in the structure
$\CKS$ the local states corresponding to coordination variables that are
not seen by all the strategies. 


Formally, $\trCL{\CKS} \egdef
(\trCL{\setstates},\trCL{\relation},\trCL{\sstate_{\init}},\trCL{\lab})$ where
\begin{itemize}
\item $\trCL{\setstates}=\prod_{c\in\coordvar(\Phi)}\setlstates_{c}$
  where $\setlstates_c = \{c_0,c_1\}$, 
\item $\trCL{\relation}=\trCL{\setstates}\times\trCL{\setstates}$, 
\item for every $\sstate\in\trCL{\setstates}$,
  $\trCL{\lab}(\sstate)=\lab(\projI[\freecoord]{\sstate})$, and
  \item $\trCL{\sstate_{\init}} \in \trCL{\setstates}$ is any
state $\sstate$ such that $\projI[\freecoord]{\sstate}=\sstate_{\init}$
\end{itemize}

Second, we define concrete observations corresponding to strategy variables in
$\Phi$. As explained in~\cite{DBLP:conf/csl/FinkbeinerS10}, and as
reflected in the semantics of \CL, the
intuition is that a strategy variable $s$ in formula $\Phi$ observes
coordination variables $\scope{s}$. Therefore, we simply define, for
each strategy variable  $s$ in $\Phi$, the concrete observation
$\cobs_{s}\egdef \scope{s}$. 

Finally, we define the \QCTLsi
formula $\trCL{\Phi}$. This is done by induction on $\Phi$ as follows (recall that we
take for atomic propositions in \QCTLsi the set of all strategy variables
in $\Phi$):
\begin{align*}
\trCL{x} 		& \egdef x  \\
\trCL{\neg \phi} 	& \egdef \neg \trCL{\phi}\\
\trCL{\phi_1\ou\phi_2} & \egdef \trCL{\phi_1}\ou\trCL{\phi_2} \\
\trCL{\X\phi} & \egdef \X\,\trCL{\phi} \\
\trCL{\phi_1\until\phi_2} & \egdef \trCL{\phi_1}\,\until\,\trCL{\phi_2}
\\
\trCL{\strquantif C \exists s.\,\phi} 		& \egdef \existsp[s]{\cobs_{s}}\A\trCL{\phi}
\end{align*}
In the last case, note that $C \subseteq \cobs_s = \scope{s}$.

Note that $\trCL{\Phi}$ is a hierarchical \QCTLsi-formula.
Also, one can easily check that the following holds:
\begin{lemma}
  \label{lem-trans-CL}
  $\unfold[\CKS]{\sstate_{\init}}\models\Phi\quad\mbox{iff}\quad
  \unfold[\trCL{\CKS}]{\trCL{\sstate_{\init}}}\models \A\trCL{\Phi}$.
\end{lemma}

Importantly, we notice that $\A\trCL{\Phi}\in\QCTLsih$, and that:
\begin{lemma}
  \label{lem-observable}
  For every $x\in\APphi$ and every $s$ quantified in $\Phi$,
  $\unfold[\trCL{\CKS}]{\trCL{\sstate_{\init}}}$ is  $\cobs_{s}$-uniform in $x$.
\end{lemma}

\subsubsection{Translation from \QCTLsi to \SLi}
\label{sec-QCTL-to-SL}
We now present a translation of \QCTLsi-instances to \SLi-instances.
It is a simple adaptation of 
 the reduction from the model-checking problem for \QCTLs
to the model-checking problem for \ATLssc presented
in~\cite{DBLP:journals/iandc/LaroussinieM15}.

Let $(\CKS,\Phi)$ be an instance of the model-checking problem for
\QCTLsi, where $\CKS=(\setstates,\relation,\val,\sstate_{\init})$ and
$\setstates\subseteq \prod_{i\in [n]}\setlstates_{i}$.  We assume,
without loss of generality, that every atomic proposition is
quantified at most once, and that if it appears quantified it does not
appear free. Also, let $\APq(\Phi)=\{p_{1},\ldots,p_{k}\}$ be the set
of atomic propositions quantified in $\Phi$, and for $i\in [k]$, let
$\cobs_{i}$ be the concrete observation associated to the quantifier
on $p_{i}$.

We build the \CGSi
$\CGSi^{\CKS}\egdef(\Act,\setpos,\trans,\val',\pos_\init,\obsint)$ over agents
$\Agf\egdef\{a_{0},a_{1},\ldots,a_{k}\}$, observations
$\Obsf\egdef\{\obs_{0},\obs_{1},\ldots,\obs_{k}\}$ and atomic propositions
$\APf\egdef\APq(\Phi) \union\{p_{\setstates}\}$, where
$p_{\setstates}$ is a fresh atomic proposition.
Intuitively, agent $a_{0}$ is in
charge of choosing transitions in $\CKS$, while agent $a_{i}$ for
$i\geq 1$ is in charge of choosing the valuation for $p_{i}\in\APq(\Phi)$.

To this aim, we let  \[\setpos\egdef
\begin{array}{l}
  \{\pos_{\sstate}\mid\sstate\in\setstates\} \;\union \\
   \{\pos_{\sstate,i}\mid
   \sstate\in\setstates \mbox{ and }i\in [k]\}\;\union\\
   \{\pos_{p_{i}}\mid 0\leq i \leq k\}\;\union\\
   \{\pos_{\perp}\}
\end{array}
\]
      and
      \[\Act\egdef
      \{\mov^{\sstate}\mid\sstate\in\setstates\}\union \{\mov^{i}\mid
      0\leq i \leq k\}.\]
 In
positions of the form $\pos_{\sstate}$ with $\sstate\in\setstates$, transitions are determined by
the action of agent $\ag_{0}$. First, she can choose to
simulate a transition in $\CKS$: for every joint action
$\jmov\in\Act^{\Agf}$ such that $\jmov_{0}=\mov^{\sstate'}$,
\[\trans(\pos_{\sstate},\jmov)\egdef
\begin{cases}
  \pos_{\sstate'} & \text{if }\relation(\sstate,\sstate')\\
  \pos_\perp & \text{otherwise}.
\end{cases}\]
She can also choose to move to a position in which agent $\ag_{i}$
will choose the valuation for $p_{i}$ in the current node:
for every joint action
$\jmov\in\Act^{\Agf}$ such that $\jmov_{0}=\mov^{i}$,
\[\trans(\pos_{\sstate},\jmov)\egdef
\begin{cases}
  \pos_{\sstate,i} & \text{if }i\neq 0\\
  \pos_\perp & \text{otherwise}.
\end{cases}\]

Next, in a position of the form $\pos_{\sstate,i}$, agent $a_{i}$
determines the transition, which codes the labelling of $p_{i}$ in the
current node: choosing $\act^{i}$ means that $p_{i}$ holds in the
current node, choosing any other action codes that $p_{i}$ does not hold.
Formally, for every joint action
$\jmov\in\Act^{\Agf}$,
\[\trans(\pos_{\sstate,i},\jmov)\egdef
\begin{cases}
  \pos_{p_{i}} & \text{if }\jmov_{i}=\mov^{i}\\
  \pos_\perp & \text{otherwise}.
\end{cases}\]

Positions of the form $\pos_{p_i}$ and $\pos_{\perp}$ are
sink positions.

The labelling function $\lab'$ is defined as follows:
\[\val'(\pos)\egdef
    \begin{cases}
      \lab(\sstate)\union \{p_{\setstates}\} & \mbox{if
      }\pos=\pos_{\sstate} \mbox{ for some } \sstate\in\setstates\\
      \emptyset & \mbox{if }\pos\in\{\pos_{\sstate,i} \mid
      \sstate\in\setstates, i\in [k]\} \union \, \{\pos_{p_0},\pos_{\perp}\}\\
      \{p_{i}\} & \mbox{if }\pos=\pos_{p_{i}}\text{ with }i\in [k]
    \end{cases}
  \]

  Finally we let $\pos_{\init}\egdef \pos_{\sstate_{\init}}$ and we
  define the observation interpretation as follows:
  \[\obsint(\obs_{0})\egdef      \{(\pos,\pos)\mid \pos\in\setpos\},\]
  meaning that agent $a_{0}$ has perfect information, and
 for $i\in [k]$, $\obsint(\obs_{i})$ is the smallest reflexive
 relation such that
\[\obsint(\obs_{i})\supseteq \bigunion_{\sstate,\sstate'\in\setstates}\{(\pos_{\sstate},\pos_{\sstate'}),(\pos_{\sstate,i},\pos_{\sstate',i})\mid \sstate
\oequiv[\cobs_{i}] \sstate' \}.\] We explain the latter
definition. First, observe that for every finite play $\fplay$ in
$\CGS^{\CKS}$ that stays in
$\setpos_{\setstates}=\{\pos_{\sstate}\mid\sstate\in\setstates\}$,
writing $\fplay=\pos_{\sstate_{0}}\ldots \pos_{\sstate_{n}}$, one can
associate a finite path $\spath_{\fplay}=\sstate_0\ldots\sstate_{n}$ in
$\CKS$. This mapping actually defines a bijection between the set of
finite paths in $\CKS$ that start in $\sstate_{\init}$ and the set of
finite plays in $\CGS^{\CKS}$
that remain in $\setpos_{\setstates}$.

Now, according to
the definition of the transition function, a strategy $\strat_{i}$
for agent $i$ with $i\in [k]$ is only relevant on finite plays of
the form $\fplay=\fplay'\cdot \pos_{\sstate,i}$, where
$\fplay'\in\setpos_{\setstates}^{*}$, and $\strat_{i}(\fplay)$
is meant to determine whether $p_{i}$ holds in $\spath_{\fplay'}$. If
$\strat_{i}$ is $\obs_{i}$-uniform, by definition of
$\obsint(\obs_{i})$, it determines an $\cobs_{i}$-uniform labelling
for $p_i$ in $\unfold[\CKS]{\sstate_{\init}}$. Reciprocally, an
$\cobs_{i}$-uniform labelling for $p_i$ in
$\unfold[\CKS]{\sstate_{\init}}$ induces an
$\obsint(\obs_{i})$-strategy for agent $\ag_{i}$.
It remains to transform $\Phi$ into an \SLi-formula.
\newcommand{\trQCTL}[2][\cobs_i]{\widetilde{#2}}

We define the \SLi formula $\trQCTL{\Phi}$ by induction on
$\Phi$ as follows:
\begin{align*}
\trQCTL{p}         & \egdef
\begin{cases}
\Eout \X \X p &
  \text{if }p=p_{i}\\
  p & \text{otherwise}
\end{cases}
\\
\trQCTL{\neg \phi}     & \egdef \neg \trQCTL{\phi}\\
\trQCTL{\phi_1\ou\phi_2} & \egdef \trQCTL{\phi_1}\ou\trQCTL{\phi_2} \\
\trQCTL{\E \psi}         & \egdef 
 \Eout(\always p_{\setstates}\et\trQCTL{\psi})\\
\trQCTL{\existsp[p_{i}]{\cobs_{i}} \phi}         & \egdef
\Estrato[\var_{i}]{\obs_{i}}\bind[\ag_i]{\var_i}\trQCTL[\cobs_{i}]{\phi}.
\end{align*}
The cases for path formulas are obtained by distributing over the operators.

Observe that player 0 is never bound to a strategy. In the case for atomic
propositions, the existential
quantification on outcomes  thus lets
player 0 choose to move to a position where agent $i$ fixes the value
for $p_i$ according to his strategy, fixed by the strategy quantifier
in the translation for formulas of the form
$\existsp[p_{i}]{\cobs_{i}} \phi$. In the translation of formulas of
the form $\E\psi$, the existential quantification on outcomes lets
player 0 choose a path in the original \CKS $\CKS$.

We have the following:
\begin{lemma}
  \label{lem-redux-QCTL}
  $\CKS\models\Phi\quad\text{if and only if}\quad \CGS^{\CKS}\models\trQCTL[{[n]}]{\Phi}$.
\end{lemma}

We observe that
if $\Phi$ is hierarchical, then $(\trQCTL[{[n]}]{\Phi},\CGS^{\CKS})$
is a hierarchical instance, and:

\begin{lemma}
  \label{lem-obs-SLi}
  For every $p\in\APfree(\Phi)$ and for every $i\in [k]$, if
  $\unfold{\sstate_{\init}}$ is $\cobs_{i}$-uniform in $p$ then
  $\pos\obseq[\obs_{i}]\pos'$ implies that $p\in\val(\pos)$ iff $p\in\val(\pos')$.
\end{lemma}

Combining Lemma~\ref{lem-trans-CL} with Lemma~\ref{lem-redux-QCTL} we
get a reduction from the model-checking problem for \CL to that for
the hierarchical fragment of \SLi, and Lemma~\ref{lem-observable}
together with Lemma~\ref{lem-obs-SLi} show that in the models produced by
this reduction, all atomic propositions are observable to all players.
This implies that in \CL one cannot reason about strategic problems
with unobservable objectives. As a result it does not fully capture
 classic distributed synthesis~\cite{PR90,kupermann2001synthesizing},
 where the specification can talk about all variables, hidden and visible.
It also shows that \CL does not capture in a natural way \ATL with
imperfect information as defined in~\cite[Section
7.1]{DBLP:journals/jacm/AlurHK02}, where imperfect information of
agents is modelled by defining which atomic propositions they can
observe. This, as well as unobservable objectives, can be naturally modelled in \SLi.




\section{Applications}
\label{sec-applications}

In this section we apply Theorem~\ref{theo-SLi} to
decide the existence of Nash Equilibria in hierarchical games of
imperfect information. We then use a similar approach to obtain decidability results
for the rational synthesis problem. In this section,  for a tuple of
agents $\bm{\ag}=(\ag_i)_{i\in[m]}$ and tuple of strategy variables
$\vecvar[\var]=(\var_i)_{i\in[m]}$, we let
$\vecbind[\bm{\ag}]{\vecvar[\var]}$ be a macro
for $\bind[\ag_1]{\var_1}\ldots\bind[\ag_m]{\var_m}$, and similarly
for the unbinding operator $\unbind[\bm{\ag}]$ which stands for $\unbind[\ag_1]\ldots\unbind[\ag_m]$.

\subsection{Existence of Nash Equilibria in games with hierarchical observations}
\label{sec-NE-hier}

A Nash equilibrium in a game is a tuple of strategies such that no
player has an incentive to deviate. Let $\Agf = \{\ag_i : i \in
[n]\}$. Assuming that agent $\ag_i$ has observation $\obs_i$ and \LTL goal $\psi_i$, the following  \SLi formula 
 expresses the
existence of a Nash equilibrium:
\begin{align*}
 \Phi_\textsc{NE} \egdef & \Estrato[x_1]{\obs_1} \dots  \Estrato[x_n]{\obs_{n}}  \bind[\bm{\ag}]{\bm{\var}} \bigwedge_{i \in [n]} \Big[
 \Big( \Estrato[y_i]{\obs_i} (a_i,y_i)\, \A\psi_i \Big) \to \A\psi_i\Big ] 
\end{align*}
where $\bm{\ag}=(\ag_i)_{i\in[n]}$ and $\bm{\var}=(\var_i)_{i\in[n]}$.

\basl{Nash equilibria do not always exist when one restricts attention
  to pure strategies, as we do in this work. This is the case already
  in finite games, and by extension also in the infinite concurrent games
  played on graphs that we consider.  This motivates the study of the
  Nash equilibria existence problem in such games. In the perfect
  information case, the problem has been solved for $\omega$-regular
  objectives, as well as more complex semi-quantitative
  objectives~\cite{DBLP:journals/corr/BouyerBMU15}.  When moving to
  imperfect information, for two players the problem is decidable for LTL
  objectives~\cite{DBLP:journals/iandc/GutierrezPW18} and parity
  objectives~\cite{filiot2018rational}. However, as for distributed
  synthesis, existence of
  Nash equilibria becomes undecidable for more than two players. This
result is proved in~\cite{DBLP:conf/fossacs/Bouyer18} for constrained Nash equilibria (when one specifies for
  each player whether her objective is satisfied or not),  and 
 in~\cite{DBLP:journals/iandc/GutierrezPW18} for
  unconstrained equilibria. In both cases the proof proceeds by reduction from the distributed
  synthesis problem~\cite{peterson2001lower,PR90}.

  The only known decidable cases for more than two players assume that
  all players receive the same
  information. In~\cite{DBLP:conf/fossacs/Bouyer18} the problem is
  solved on games where players observe the evolution of the game via
  \emph{public signals} and objectives are given by visible parity
  conditions or mean-payoff functions.
  In~\cite{DBLP:conf/ijcai/BelardinelliLMR17}, an epistemic extension
  of strategy logic is used to solve the existence of Nash equilibria
  on games with \emph{broadcast actions} for objectives given as
  formulas from epistemic temporal logic. A stronger notion of Nash
  equilibria, called \emph{locally consistent equilibria}, is studied
  in~\cite{ramanujam2010communication}. In a locally consistent
  equilibrium, each player's strategy has to be a best response not
  only to other players' strategies in the equilibrium, but also to
  all strategies that are indistinguishable from those in the
  equilibrium. It is proved in~\cite{ramanujam2010communication} that
  the existence of such equilibria is decidable on a model of games
  close in spirit to those with public signals studied
  in~\cite{DBLP:conf/fossacs/Bouyer18}.

  Here we show that the
  existence of Nash equilibria is decidable for $n$ players when
  observations are hierarchical and objectives are given as \LTL formulas. Note that this result is orthogonal
  to those described above, which all allow in a way or another some
  non-hierarchical information: in~\cite{DBLP:conf/fossacs/Bouyer18}
  players know their own actions in addition to the public signals,
  in~\cite{ramanujam2010communication} they know their private local
  state, and in~\cite{DBLP:conf/ijcai/BelardinelliLMR17} they can have
  incomparable initial knowledge of the situation.}


\begin{definition}
A \CGSi $\CGSi$ presents \emph{hierarchical
  observation}~\cite{DBLP:journals/acta/BerwangerMB18} if
the ``finer-than'' relation is a total ordering, i.e., if for all
$\obs,\obs' \in \Obsf$, either $\obsint(\obs) \subseteq \obsint(\obs')$
or $\obsint(\obs') \subseteq \obsint(\obs)$.  
\end{definition}

Let $\CGSi$ be a \CGSi with hierarchical observation, and since all agents have symmetric roles in the problem considered,
assume without loss of generality that $\obsint(\obs_n)\subseteq\ldots\subseteq\obsint(\obs_1)$.

Because of the nested strategy quantifiers $\Estrato[y_i]{\obs_i}$, the instance $(\CGSi,\Phi_\textsc{NE})$  is \emph{not}
 hierarchical  even if $\CGSi$ yields hierarchical observation (unless $\obsint(\obs_i) = \obsint(\obs_j)$ for all
 $i,j\in[n]$). \bas{However, considering
the special observation symbol $\obs_p$ that is always interpreted as
the identity relation (and thus represents perfect observation),
and letting
\begin{align*} 
\Phi' \egdef & \Estrato[x_1]{\obs_1} \dots  \Estrato[x_n]{\obs_{n}}  \bind[\bm{\ag}]{\bm{\var}} \bigwedge_{i \in [n]} \Big[
 \Big ( \Estrato[y_i]{\obs_{p}} (a_i,y_i)\, \Eout\psi_i \Big ) \to \Eout\psi_i \Big],
\end{align*}
we have that $\Phi'$ forms a hierarchical instance with any \CGSi 
that presents hierarchical observation. Besides, we can prove that for deterministic strategies, $\Phi'$ is
equivalent to $\Phi_\textsc{NE}$:

\begin{lemma}
  \label{lem-eq-formulas-NE}
  If we consider deterministic strategies, then $\Phi_\textsc{NE}\equiv \Phi'$.
\end{lemma}
}
\begin{proof}
\bas{Concerning the universal versus existential 
  quantification on outcomes, it
  is enough to observe that  assigning a
  deterministic strategy to each agent determines a unique
  outcome. Next, to  change each inner $\obs_i$ for $\obs_p$,
we exploit the fact that} in a one-player game of
  partial observation (such a game occurs when all but one player have
  fixed their strategies),
  the player has a strategy enforcing some goal iff she has a
  uniform strategy enforcing that goal.  

To
see this, it is enough to establish that for every  \CGSi $\CGSi$ and
position $\pos$,
\[ \CGSi,\assign,\pos \models \Estrato[y_i]{\obs_{p}} (a_i,y_i)\,\Eout \psi_i \equivaut \Estrato[y_i]{\obs_i} (a_i,y_i)\,\Eout \psi_i,
\]
for every $i \in [n]$ and every assignment $\assign$ such that
$\assign(a_j)$ is defined for all $j$.

To this end, fix $i$ and $\assign$. The right-to-left implication is
immediate (since $\obs_{p}$ is finer than $\obs_i$).  For the converse, let
$\sigma$ be an $\obs_p$-strategy (\ie, a perfect-information strategy) such
that $\CGSi',\assign',\pos_{\init} \models
\psi_i$, where $\assign'=\assign[y_i \mapsto \sigma,\ag_i \mapsto \sigma]$.
Because we consider deterministic strategies and $\assign'$ assigns a
strategy to each agent, it defines a unique outcome $\pi$ from the initial
position, \ie, $\out(\assign',\pos_\init)=\{\pi\}$. We construct an
$\obs_i$-strategy $\sigma'$  such that if $\ag_i$ uses it instead of
$\sigma$, we obtain the same outcome $\pi$, \ie,
$\out(\assign'',\pos_\init)=\{\pi\}$, where
$\assign''=\assign[y_i\mapsto \sigma',\ag_i\mapsto\sigma']$.
 This can be done as
follows: if $\fplay \obseq[\obs_{i}] \pi_{\leq |\fplay|-1}$ then define
$\sigma'(\fplay) \egdef \sigma(\pi_{\leq |\fplay|-1})$, and otherwise
let $\sigma'(\rho) \egdef c$ for some fixed action $c \in \Act$.
It is easy to see that $\sigma'$ is an $\obs_i$-strategy and that
$\assign''$ produces the same outcome as $\assign$ from $\pos_\init$.
\end{proof}

\bas{
\begin{corollary}
  \label{cor-NE}
If we consider deterministic strategies, then the existence of Nash
Equilibria in games with hierarchical observation and $k$ different observations is in \kEXPTIME[(k+1)].
\end{corollary}
}
\begin{proof}
  \bas{
  Deciding the existence of a Nash Equilibrium in a \CGSi $\CGSi$
  amounts to model-checking formula $\Phi_\textsc{NE}$ in $\CGSi$,
  which by Lemma~\ref{lem-eq-formulas-NE} is equivalent to
  model-checking $\Phi'$ in $\CGSi$ if we restrict to deterministic
  strategies.  Because $\Phi'$ forms hierarchical instances with games
  that yield hierarchical observation, by Theorem~\ref{theo-SLi} we
  can model check it on such games. Now because each $\psi_i$ is an
  \LTL formula, we have that
  \begin{align*}
  \ndd\left(\Estrato[y_i]{\obs_{p}} (a_i,y_i)\, \Eout\psi_i\right ) &= (0,\nd),\\
  \ndd\left (\bigwedge_{i \in [n]} \Big[ \Big( \Estrato[y_i]{\obs_{p}} (a_i,y_i)\,
  \Eout\psi_i \Big) \to \Eout\psi_i \Big]\right ) & =(0,\alt),    
  \end{align*}
 and finally we obtain that  $\ndd(\Phi')=(k,\nd)$,
  where $k$ is the number of different observations in $\CGSi$, \ie,
  $k=|\{\obsint(\obs_1),\ldots,\obsint(\obs_n)\}|$.
  By Proposition~\ref{prop-compl-SLi}, we can model check $\Phi'$ on $\CGSi$
  in time $(k+1)$-exponential, which concludes.}
\end{proof}

We now show that, using the same trick, our main result can be applied
to solve a more general problem called \emph{rational synthesis}.


\subsection{Rational distributed synthesis in games with hierarchical observations}
\label{sec-rat-synth}

In classic synthesis, the environment is considered monolithic
and ``hostile'', in the sense that the system to be synthesised should be
able to deal with all possible behaviours of the environment, even the
most undesirable ones. This is a very strong requirement that can
not always be met. When the environment can be considered rational,
and its objective is known, it is reasonable to relax this requirement
by asking that the system to synthesise behave well against the
\emph{rational} behaviours of the environment. 
This problem is known as the \emph{rational synthesis} problem~\cite{fisman2010rational,DBLP:journals/amai/KupfermanPV16,DBLP:conf/icalp/ConduracheFGR16,filiot2018rational}.
In the setting considered in the works
above-mentioned, the system is seen as an agent $a$  and the environment is composed of
several  components, say $\{e_1,\ldots,e_m\}$, that are assumed to be
rational and follow individual
objectives. While~\cite{DBLP:conf/icalp/ConduracheFGR16} and~\cite{filiot2018rational}  consider
various types of objectives such as reachability, safety or parity,
here we consider \LTL
objectives as is done in
\cite{fisman2010rational,DBLP:journals/amai/KupfermanPV16}: the
specification for the system is an \LTL formula $\psi_g$, and the
objective of each component $e_i$ of the environment is an \LTL formula $\psi_i$.
However note that the decidability results we establish would also
hold for arbitrary omega-regular objectives.

\subsubsection{Rational synthesis: state of the art}
\label{sec-rat-syn-sota}
Two variants of the rational synthesis problem have been considered:
the \emph{cooperative} one, in which it is possible to tell the  environment
how to behave, as long as the suggested behaviour for each component
forms an equilibrium, and the \emph{non-cooperative} one, in which
the components of the environment may have any behaviour that forms an equilibrium.
The existence of a solution to these problems can be expressed by the
formulas $\Phiratc$ and $\Phiratnc$, respectively, defined as follows:
\begin{align*}
\Phiratc&\egdef  \Estrato[\var]{\obs_p}\Estrato[\varb_1]{\obs_p}\ldots\Estrato[\varb_{m}]{\obs_p}\bind[\ag]{\var}\vecbind{\vecvar}
\,\phiauxgen \wedge \A\psi_g\\
\Phiratnc&\egdef \Estrato[\var]{\obs_p}\Astrato[\varb_1]{\obs_p}\ldots\Astrato[\varb_{m}]{\obs_p}\bind[\ag]{\var}\vecbind{\vecvar}
\,\phiauxgen\to \A\psi_g 
 \end{align*}
 where $\bm{e}=(e_i)_{i\in[m]}$, $\bm{\varb}=(\varb_i)_{i\in[m]}$, and
 $\phiauxgen$ expresses that $\vecvar$ forms an equilibrium for the
 environment.  Also, as in the previous section, $\obs_p$ represents
 the perfect-information observation.  Three different kinds of
 equilibria are considered in~\cite{DBLP:journals/amai/KupfermanPV16}:
 profiles of dominant strategies, Nash equilibria, and subgame-perfect
 equilibria. 
 Here we only consider Nash
 equilibria, because subgames of games with imperfect information
 should start in situations where all players have perfect information
 of the state, which we do not know how to express in \SLi; and for
 dominant strategies, the natural formula to express them does not
 give rise to non-trivial decidable cases in the imperfect-information
 setting that we introduce later. The rational synthesis problem for
 Nash equilibria 
 is obtained by replacing $\phiauxgen$ in the above formula with:
 \begin{align*}
   \phiauxNE&\egdef    \bigwedge_{i\in [m]}\Big[
 \Big ( \Estrato[\varb'_i]{\obs_{p}} (e_i,\varb'_i)\, \A\psi_i \Big ) \to
              \A\psi_i \Big]
 \end{align*}


\basl{It is proved in~\cite{DBLP:journals/amai/KupfermanPV16} that
  these problems are decidable for perfect information. Concerning
  imperfect information,  because the existence of Nash equilibria is
  undecidable for three players, the problem is undecidable when
  the environment consists of at least three
  components~\cite{filiot2018rational}.
  Three decidable cases are known: when the environment consists of a
  single component~\cite{filiot2018rational},  when actions of all components are
  public~\cite{DBLP:conf/ijcai/BelardinelliLMR17},
  and when only the system has imperfect information
  while the (finitely many) components of the environment are
  perfectly informed~\cite{filiot2018rational}.

We now extend the latter result by defining  a generalisation of
the rational synthesis problem that we call \emph{rational distributed 
  synthesis}, and solving it in the case of hierarchical
information. The case where the environment is perfectly informed and
the system consists of a single component, solved in~\cite{filiot2018rational}, is a particular case of our
Corollary~\ref{cor-ncRS} below\footnote{We only consider \LTL
  objectives, but  our automata construction
can be adapted  to handle all $\omega$-regular objectives.}.
However the other decidability result established
in~\cite{filiot2018rational} does not assume hierarchical information,
and thus cannot be derived from the results we now present.}

\subsubsection{Rational distributed synthesis}
While for perfect information, distributed synthesis amounts to
synthesis for a single meta-component \basl{which tells each component
  what to do}, 
in the context of imperfect information it makes sense to consider
that the system to be synthesised is composed of various components
$\{a_1,\ldots,a_n\}$ with different observation power, say $\obs_i$
for component $a_i$. We also let $\obs^e_i$ be the observation of the
environment's component $e_i$, for $i\in [m]$.
 
We consider the imperfect-information variants of cooperative and
non-cooperative rational synthesis defined by the following \SLi formulas:
\begin{align*}
  \Phiratci&\egdef  \Estrato[\var_1]{\obs_1}\ldots\Estrato[\var_n]{\obs_n}\Estrato[\varb_1]{\obs^e_1}\ldots\Estrato[\varb_{m}]{\obs^e_m}\vecbind[\bm{\ag}]{\bm{\var}}\vecbind{\vecvar}
             \,\phiauxgen \wedge \A\psi_g\\
  \Phiratnci&\egdef \Estrato[\var]{\obs_1}\ldots\Estrato[\var_n]{\obs_n}\Astrato[\varb_1]{\obs^e_1}\ldots\Astrato[\varb_{m}]{\obs^e_m}\vecbind[\bm{\ag}]{\bm{\var}}\vecbind{\vecvar}
              \,\phiauxgen\to \A\psi_g 
 \end{align*}

The formula for Nash equilibrium is
adapted as follows:
 \begin{align*}
   \phiauxNEi&\egdef    \bigwedge_{i\in [m]}\Big[
 \Big ( \Estrato[\varb'_i]{\obs^e_{i}} (e_i,\varb'_i)\, \A\psi_i \Big ) \to
              \A\psi_i \Big]
 \end{align*}

The only difference with the perfect-information case 
 is that we use the observation of the different
 components of the environment instead of the perfect-information
 observation.
 
 We call the problems expressed by formulas $\Phiratci$ and
 $\Phiratnci$  \emph{cooperative rational distributed
   synthesis} and \emph{non-cooperative rational distributed
   synthesis}, respectively.
 As in the previous section on the existence of Nash equilibria, one
 can see that even if there is a total hierarchy on all observations,
 these formula do not yield hierarchical instances unless all
 observations are the same. However, the trick
 applied in the proof of Corollary~\ref{cor-NE} also applies here,
 both for Nash equilibria and subgame-perfect equilibria, \ie, we can
 replace each $\obs^e_i$ with $\obs_p$ in $\phiauxNEi$  without affecting the semantics of formulas $\Phiratci$
 and $\Phiratnci$.  As a result, when there is a hierarchy on
 observations $\obs_1,\ldots,\obs_n,\obs^e_1,\ldots,\obs^e_m$, the
 cooperative rational distributed synthesis is decidable. 
 \begin{corollary}
   \label{cor-cRS}
   If we consider deterministic strategies and hierarchical
   observations, then  cooperative rational distributed synthesis
   is decidable. 
 \end{corollary}

 For the
 non-cooperative variant,  one cannot switch universal
 quantifications on strategies for the environments with 
 existential quantifications for the system in order to obtain
 hierarchical instances, as the resulting formula would then
 capture a different problem. As a consequence, in
 addition to a hierarchy on observations $\obs_1,\ldots,\obs_n,\obs^e_1,\ldots,\obs^e_m$, we need
 to have that the components of the environment observe better than
 the components of the system or, in other words, that the least
 informed component of the environment observes better than the best
 informed component of the system. When it is the case, we say
 that the environment is \emph{more informed} than the system.
 \begin{corollary}
   \label{cor-ncRS}
Non-cooperative rational distributed synthesis
   is decidable   for deterministic strategies and hierarchical
   observations where the environment is more informed than the system. 
 \end{corollary}

This result applies for instance  when there is hierarchical
information amongst the components of the system, and the environment
has perfect information. Note that when the system consists of a
single component, this corresponds to the second decidability result
in~\cite{filiot2018rational}.
As we mentioned in the introduction,
considering that the opponent has perfect information is something
classic in two-player games with imperfect information, as doing so
ensures that the strategy one synthesises is winning no matter how much
the opponent observes. In Reif's words, this amounts to considering the
possibility that the opponent may ``cheat'' and use information that
it normally does not have access to~\cite{reif1984complexity}. The
non-cooperative rational synthesis problem is not precisely a
two-player game, but it resembles one in the sense that the system as
a whole (composed of its various components $a_1,\ldots,a_n$) should
win against any ``rational'' behaviour of the environment as a
whole. In this view, considering that the components of the environment have perfect
information thus yields a distributed system that is robust to
possible leaks of hidden information to the environment.

\begin{remark}
When all components of the environment have perfect
information, $\Phiratci$ and $\Phiratnci$ already form hierarchical instances with games
where there is hierarchical observation amongst the system's components, and one does
not need to resort to the trick used in the proof of
Corollary~\ref{cor-NE}. A consequence is that in that case,
corollaries~\ref{cor-cRS} and \ref{cor-ncRS} also hold for nondeterministic strategies.  
\end{remark}



\section{Conclusion}
\label{sec-outlook}

We introduced \SLi, a logic for reasoning about strategic behaviour in
multi-player games with imperfect information.  The syntax specifies
the observations with which strategies have to work, and thus allows one to reason about strategic problems
in settings where agents can change observation power, for instance by
being eventually granted access to previously hidden information.
Moreover our logic contains an outcome quantifier and an unbinding
operator which simplify the semantics, make it easier
to express branching-time properties, allow us to naturally consider
nondeterministic strategies, and make the correspondence with \QCTLsi
tighter, enabling us to derive precise complexity results for the
model-checking of \SLi.

We isolated the class of hierarchical
formula/model pairs $(\Phi,\CGSi)$ and proved that for such instances
one can decide whether $\CGSi \models \Phi$.  The proof reduces
(hierarchical) instances of \SLi to (hierarchical) formulas of \QCTLsi, a
low-level logic that we introduced, and that serves as a natural
bridge between \SLi and automata constructions. 
We also studied in detail the complexity of the model-checking
problems solved in this work. To do so we introduced a new measure on formulas
called \emph{simulation depth}. This measure, though being a
purely syntactic notion, reflects the complexity of automata
constructions required to treat a given formula.

Since one can alternate quantifiers in \SLi, our decidability result
goes beyond synthesis and can be used to easily obtain the
decidability of many strategic problems. In this work we applied it to
the problem of existence of Nash equilibria in games
with hierarchical observation, and to
 the imperfect-information generalisations of rational synthesis
that we called (cooperative and non-cooperative) \emph{rational distributed synthesis}.
Our result has also been used to prove that the existence of
admissible strategies in games with hierarchical information is
decidable~\cite{brenguier2017admissibility}.

An interesting direction for future work
would be to try and adapt the notion of hierarchical instances to allow for situations in which hierarchies can change along a 
play, as done in~\cite{DBLP:journals/acta/BerwangerMB18}. We would also
like to consider alternatives to the
synchronous perfect recall setting considered here, such as the
classic asynchronous perfect recall setting~\cite{fagin1995reasoning,DBLP:conf/mfcs/Puchala10}, or the more recent notion
of causal knowledge~\cite{genest2015knowledge}. Finally, it is often
interesting in presence of imperfect information to introduce
epistemic operators to reason explicitely about what agents know. We
already generalised the main result of this work to an extension of
\SLi with such operators~\cite{maubert2018reasoning}; we would like to
see if this can be used to reason about subgame-perfect equilibria in
games with imperfect information, which do not seem to be easy
to characterise in \SLi, as mentioned in
Section~\ref{sec-rat-syn-sota}. Indeed, in games with imperfect
information, the notion of subgame specifies that the initial
situation should be known to all
players~\cite{selten1965spieltheoretische}, a property that epistemic
logics are meant to be able to express.


\begin{acks}
  We thank anonymous reviewers for their valuable comments on a previous version of this work.
  This project has received funding from the
  \grantsponsor{h2020}{European Union's Horizon 2020 research and
    innovation
    programme}{https://ec.europa.eu/programmes/horizon2020/en} under
  the Marie Sklodowska-Curie grant agreement No
  \grantnum{h2020}{709188}.
\end{acks}



\newpage
\appendix
\section{Proof of Proposition~\ref{prop-size-automata}}
\label{sec-appendix-a}

First,  for every \LTL formula $\psi$ one can build a parity word
automaton   $\autopsi$ with two colours and $2^{O(|\psi|)}$
states~\cite{vardi1994reasoning}. Let  $\kpsi\in\setn$ be such that
the number of states of $\autopsi$ is bounded by $2^{\kpsi|\psi|}$.

We also state a more precise version of Theorem~\ref{theo-simulation}:
for every \ATA $\auto$ with $n$ states and $l$ colours, one can build
an \NTA $\NTA$ with at most $2^{O(nl\log(nl))}$ states and $O(nl)$
colours such that $\lang(\ATA)=\lang(\NTA)$~\cite{DBLP:journals/tcs/MullerS95,loeding}. We let $\ka,\kb\in\setn$
be such that the number of states of $\NTA$ is bounded by $2^{\ka
  nl\log(nl)}$ and the number of colours by $\kb nl$.

Proposition~\ref{prop-size-automata} follows directly from the following.

\begin{proposition}
        Let $\Phi$ be a $\QCTLsih$ formula, $\CKS$ a \CKS, and let $\APq=\APq(\Phi)$.  
        For every subformula $\phi$ of $\Phi$ and state
        $\sstate\in\CKS$, it holds that:
        \begin{itemize}
        \item if
        $\ndd_k(\phi)=0$, $\bigauto[\sstate]{\phi}$ has at most
        $\gsphi$ states  and
        2 colours,
        \item 
        if $\ndd_k(\phi)\geq 1$, 
        $\bigauto[\sstate]{\phi}$ has at most
        $\tower{\ndd_k(\phi)}{\gsphi \log \gsphi}$  states, and its
        number of colours is at most
        $\tower{\ndd_k(\phi)-1}{\gsphi \log \gsphi}$, 
      \end{itemize}
              with
              $\gsphi=(4\ka+2\kb)^{\rEd(\phi)}|\phi||\CKS|^{\Ed(\phi)}2^{\kpsi|\phi|\Ed(\phi)}$.
              
    In addition, if $\bigauto{\phi}$ has state set $\tQ$, for each
    $\tq\in\tQ$ and $a\in 2^\APq$, we have
    $|\tdelta(\tq,a)|\leq |\CKS| |\tQ|^{|\CKS|} 2^{\paraphi|\phi|}$,
  where $\paraphi=1+\Ed(\phi)$.
\end{proposition}

\begin{proof}
  We prove the result by induction on $\phi$.
  
\halfline
$\bm{\phi=p:}$ in this case $\ndd_k(\phi)=\rEd(\phi)=\Ed(\phi)=0$. By construction,
$\bigauto[\sstate]{\phi}$ has one state $\tq_\init$ and two
colours, so that the first part of the claim holds. In addition,  each formula of its transition function is of size one,
so that the second part of the claim also holds.

\halfline
$\bm{\phi=\neg \phi':}$ Complementing an \ATA does not change the
number of states,  number
of  colours or size of formulas in the transition function, so that the result follows by induction
hypothesis and the fact that $|\phi'|\leq |\phi|$ and $\Ed(\phi)=\Ed(\phi')$.

\halfline
$\bm{\phi=\phi_1 \ou \phi_2:}$
To establish the claim about number of states and colours we split
cases. First  we consider the case where $\ndd_k(\phi)=0$. In that case we
also have $\ndd_k(\phi_1)=\ndd_k(\phi_2)=0$. By induction hypothesis, for $i\in \{1,2\}$,
$\bigauto[\sstate]{\phi_i}$ has at most  $\gsphi[\phi_i]$  states and
        $2$ colours. These automata are then narrowed down, but the
        narrowing operation leaves the size of formulas in the
        transition function unchanged (in fact they may become smaller,
        but not bigger, see~\cite{kupferman1999church}). Therefore, by construction $\bigauto{\phi}$ has at most
        $1+\gsphi[\phi_1]+\gsphi[\phi_2]$ states and two colours.

        Now we have that
        \begin{align*}
          1+\gsphi[\phi_1]+\gsphi[\phi_2] &= 1 + \sum_{i\in
                                            \{1,2\}}(4\ka+2\kb)^{\rEd(\phi_i)}|\phi_i||\CKS|^{\Ed(\phi_i)}2^{\kpsi|\phi_i|\Ed(\phi_i)}\\
          &= 1 +
            (4\ka+2\kb)^{\rEd(\phi)}|\phi||\CKS|^{\Ed(\phi)}\sum_{i\in
            \{1,2\}}2^{\kpsi|\phi_i|\Ed(\phi)}\\
          &\leq 1 +
            (4\ka+2\kb)^{\rEd(\phi)}|\phi||\CKS|^{\Ed(\phi)}2^{\kpsi(|\phi_1|+|\phi_2|)\Ed(\phi)}\\
          1+\gsphi[\phi_1]+\gsphi[\phi_2]          &\leq
            (4\ka+2\kb)^{\rEd(\phi)}|\phi||\CKS|^{\Ed(\phi)}2^{\kpsi(|\phi_1|+|\phi_2|+1)\Ed(\phi)}
        \end{align*}
        We get that
        \begin{equation}
          \label{eq-boum}
        1+\gsphi[\phi_1]+\gsphi[\phi_2] \leq \gsphi          
        \end{equation}
        which concludes the claim about the number of states.

        \halfline Now for the case where $\ndd_k(\phi)\geq 1$. By
        definition of nondeterminisation depth, for at least
        one $i\in \{1,2\}$ we have $\ndd_k(\phi_i)\geq 1$. Also,  the
        number of colours used in $\bigauto{\phi}$ is the maximum
        between the number of colours used in $\bigauto{\phi_1}$ and
        those used in $\bigauto{\phi_2}$. By induction hypothesis it is the case that
        $\bigauto{\phi_i}$ has at most
        $\tower{\ndd_k(\phi_i)-1}{\gsphi[\phi_i]\log \gsphi[\phi_i]}$
        colours if $\ndd_k(\phi_i)\geq 1$, or 2 if
        $\ndd_k(\phi_i)=0$. 
        Therefore,  the number of
        colours in $\bigauto{\phi}$ is at most
        $\tower{\ndd_k(\phi_i)-1}{\gsphi[\phi_i]\log
          \gsphi[\phi_i]}$ for some $i$, which is less than
        $\tower{\ndd_k(\phi)-1}{\gsphi[\phi]\log \gsphi[\phi]}$.

        For the number of states $|\tQ|$ in $\bigauto{\phi}$, we have
        that $|\tQ|=1+|\tQ_1|+|\tQ_2|$, where $\tQ_i$ is the set of
        states of $\bigauto{\phi_i}$. By induction hypothesis we get
        \begin{align*}
          |\tQ|&\leq 1 + \sum_{i\in\{1,2\}}\tower{\ndd_k(\phi_i)}{\gsphi[\phi_i]\log\gsphi[\phi_i]}\\
          &\leq 1 +
            \tower{\ndd_k(\phi)}{\sum_{i\in\{1,2\}}\gsphi[\phi_i]\log\gsphi[\phi_i]}\\
               &\leq
                 \tower{\ndd_k(\phi)}{(\sum_{i\in\{1,2\}}\gsphi[\phi_i]+1)\log\gsphi[\phi]}\\
         |\tQ| & \leq \tower{\ndd_k(\phi)}{\gsphi\log\gsphi[\phi]} 
            \mbox{\hspace{2cm} (using Equation~\eqref{eq-boum})}
        \end{align*}
        which concludes the claim about the number of states.

\halfline
        Concerning the size of formulas in the transition function,
for all  states from $\bigauto{\phi_1}$ and $\bigauto{\phi_2}$
the transition function is unchanged and the result thus holds by
induction  hypothesis. For the remaining state $\tq_\init$, we have by
definition $\tdelta(\tq_\init,a)=\tdelta^1(\tq_\init^1,a)\ou
\tdelta^2(\tq_\init^2,a)$ and thus $|\tdelta(\tq_\init,a)|=|\tdelta^1(\tq_\init^1,a)|+
|\tdelta^2(\tq_\init^2,a)|+1$. By induction hypothesis we get that
\begin{align*}
  |\tdelta(\tq_\init,a)| &\leq |\CKS|
                         |\tQ_1|^{|\CKS|}2^{\paraphi(\phi_1)|\phi_1|}+|\CKS|
                         |\tQ_2|^{|\CKS|}2^{\paraphi(\phi_2)|\phi_2|}+1\\
  &\leq |\CKS|
    2^{\paraphi(\phi)(|\phi_1|+|\phi_2|)}(|\tQ_1|^{|\CKS|}+|\tQ_2|^{|\CKS|})\\
                       &\leq |\CKS| 2^{\paraphi(\phi)|\phi|}(|\tQ_1|+|\tQ_2|)^{|\CKS|}
\end{align*}
And thus $|\tdelta(\tq_\init,a)|\leq  |\CKS| 2^{\paraphi(\phi)|\phi|}|\tQ|^{|\CKS|}$
as required.

\halfline
$\bm{\phi=\E\psi:}$
The word automaton built for the \LTL skeleton of $\psi$ is in fact a
B\"uchi automaton, and thus uses only two colours.
The number of colours used by $\bigauto{\phi}$ is therefore the maximum 
number of colours used by the automata $\bigauto{\phi_i}$ built for
the maximal state subformulas $\phi_i$ in $\psi$, and the result
follows by induction hypothesis.

Concerning the number of states, let $|\tQ_\phi|$ (resp. $|\tQ_i|$, $|\tQ_\psi|$) be the number of
states in $\bigauto{\phi}$ (resp. $\bigauto{\phi_i}$, $\autopsi$). Note that the
number of states in $\bigauto[\sstate']{\phi_i}$ does not depend on
$\sstate'$. Recall that  $\max(\psi)=\{\phi_1,\ldots,\phi_n\}$ is the
  set of maximal state subformulas of $\psi$, and let
$\psi'$ be the \LTL skeleton of $\psi$, \ie, the \LTL formula obtained
from $\psi$ by replacing maximal state
subformulas $\phi_i$ with propositions $p_{\phi_i}$.  We thus have
\begin{align*}
  |\tQ|&= |\tQ_\psi||\CKS| + 2|\CKS| \sum_{i\in [n]}|\tQ_i|\\
  &\leq 2^{\kpsi|\psi'|}|\CKS| + 2|\CKS|
    \sum_{i\in [n]}\tower{\ndd_k(\phi_i)}{\gsphi[\phi_i]\log
    \gsphi[\phi_i]}\\
|\tQ|  &\leq 2^{\kpsi|\psi'|}|\CKS| \left (1+ 
    \tower{\ndd_k(\phi)}{\sum_{i\in [n]}\gsphi[\phi_i]\log
     \gsphi[\phi_i]}\right )
\end{align*}
And thus
\begin{equation}
  \label{eq:bouma}
   |\tQ|  \leq 2^{\kpsi|\psi'|}|\CKS| \left (1+ 
    \tower{\ndd_k(\phi)}{\log
     \gsphi\sum_{i\in [n]}\gsphi[\phi_i]}\right )
\end{equation}
Now observe that for each $i\in [n]$ we have that $\Ed(\phi_i)\leq
\Ed(\phi)-1$, and $\rEd(\phi_i)=\rEd(\phi)$. Therefore,
\begin{align*}
  \sum_{i\in [n]}\gsphi[\phi_i] &=
                                  (4\ka+2\kb)^{\rEd(\phi)}\sum_{i\in
                                  [n]}|\phi_i||\CKS|^{\Ed(\phi_i)}2^{\kpsi
                                  |\phi_i|\Ed(\phi_i)}\\
&\leq (4\ka+2\kb)^{\rEd(\phi)}|\CKS|^{\Ed(\phi)-1}(\sum_{i\in
  [n]}|\phi_i|)2^{\kpsi(\Ed(\phi)-1)\sum_{i\in [n]}|\phi_i|}
\end{align*}
Using this in Equation~\eqref{eq:bouma} we get
\begin{align*}
   |\tQ|  &\leq 2^{\kpsi|\psi'|}|\CKS| \left (1+ 
    \tower{\ndd_k(\phi)}{(4\ka+2\kb)^{\rEd(\phi)}|\CKS|^{\Ed(\phi)-1}(\sum_{i\in
  [n]}|\phi_i|)2^{\kpsi(\Ed(\phi)-1)\sum_{i\in [n]}|\phi_i|}\log
     \gsphi}\right
            )\\
   &\leq 2^{\kpsi|\psi'|}\left (1+ 
    \tower{\ndd_k(\phi)}{(4\ka+2\kb)^{\rEd(\phi)}|\CKS|^{\Ed(\phi)}(\sum_{i\in
  [n]}|\phi_i|)2^{\kpsi(\Ed(\phi)-1)\sum_{i\in [n]}|\phi_i|}\log
     \gsphi}\right
     )\\
     &\leq 2^{\kpsi|\psi'|}\tower{\ndd_k(\phi)}{(4\ka+2\kb)^{\rEd(\phi)}|\CKS|^{\Ed(\phi)}(1+\sum_{i\in
  [n]}|\phi_i|)2^{\kpsi(\Ed(\phi)-1)\sum_{i\in [n]}|\phi_i|}\log
     \gsphi}\\
  & \leq \tower{\ndd_k(\phi)}{(4\ka+2\kb)^{\rEd(\phi)}|\CKS|^{\Ed(\phi)}|\phi|2^{\kpsi B}\log
     \gsphi},
\end{align*}
where $B=(\Ed(\phi)-1)\sum_{i\in [n]}|\phi_i|+|\psi'|$.
To conclude it only remains to show that $B\leq |\phi|\Ed(\phi)$.
Because $\phi=\E\psi$, it holds that $\Ed(\phi)\geq 1$. If $\Ed(\phi)=1$, we have
$B=|\psi'|\leq |\phi|\Ed(\phi)$. Now if $\Ed(\phi)\geq 2$, we have
 \[B=(\Ed(\phi)-2)\sum_{i\in [n]}|\phi_i|+|\psi'|+\sum_{i\in [n]}|\phi_i|\]
Clearly, $\sum_{i\in [n]}|\phi_i|\leq |\phi|$, and $|\psi'|+\sum_{i\in
  [n]}|\phi_i|\leq 2|\phi|$, and the result follows. Note that it
could seem that $|\psi'|+\sum_{i\in [n]}|\phi_i|\leq |\phi|$. It is
true if one defines the size of a formula as the number of connectors,
but not if one also 
 counts atomic propositions, as we do here. However it is true that $|\psi'|+\sum_{i\in
  [n]}|\phi_i|\leq 2|\phi|$, independently of the definition of
formulas' size.

\halfline
It remains to establish the claim about the size of transition
formulas. By definition, for every state $q$ of $\bigauto{\phi}$ that comes
from some $\ATA^{i}_{\sstate'}$ or $\compl{\ATA^{i}_{\sstate'}}$, the
transition function is unchanged and thus the result follows by
induction hypothesis and the fact that narrowing and complementation  do
not increase the size of formulas in transition functions.
Now for the remaining states, for each $(\qpsi,\sstate')\in \tQ$ and every
$a\in 2^{\APq(\Phi)}$,  we have
\begin{align*}
  |\tdelta((\qpsi,\sstate'),a)|&\leq \sum_{a'\in 2^{\max(\psi)}}     \left (
                  |\tdelta_{\psi}((\qpsi,\sstate'),a')| + 1 +
      \sum_{\phi_i\in a'}(|\tdelta^{i}_{\sstate'}(\tq^{i}_{\sstate'},a)|+1)+ 
     \sum_{\phi_i\notin
      a'}(|\compl{\delta^{i}_{\sstate'}}(\compl{\tq^{i}_{\sstate'}},a)|+1)
    \right )
\end{align*}
Now by induction hypothesis, and because complementation does not
increase the size of formulas, we get:
\begin{equation}
  \label{eq-boumb}
  |\tdelta((\qpsi,\sstate'),a)|\leq \sum_{a'\in 2^{\max(\psi)}}     \left (
                  |\tdelta_\psi((\qpsi,\sstate'),a')|+
2 \sum_{i\in[n]}|\CKS|2^{\paraphi(\phi_i)|\phi_i|}|\tQ_i|^{|\CKS|}
    \right )+2^{|\max(\psi)|}+2|\max(\psi)|2^{|\max(\psi)|},
\end{equation}
where $|\tQ_i|$ is the number of states in automaton $\bigauto[\sstate']{\phi_i}$.
Now by definition,
\begin{align*}
|\tdelta_\psi((\qpsi,\sstate'),a')|& = \left (\sum_{\tq'\in\Deltapsi(\qpsi,a')}\sum_{
    \sstate''\in\relation(\sstate')}
                               1\right )+|\Deltapsi(\qpsi,a')||\relation(\sstate')|-1\\
|\tdelta_\psi((\qpsi,\sstate'),a')|  &\leq
                                       2|\Deltapsi(\qpsi,a')||\relation(\sstate')|-1
\end{align*}
We thus have
\begin{equation}
  \label{eq:bouya}
  |\tdelta_\psi((\qpsi,\sstate'),a')|  \leq 2|\tQ_{\psi'}||\CKS|-1
\end{equation}
where $\tQ_{\psi'}$ is the set of states of the word automaton  $\autopsi$. 
Using this in Equation~\ref{eq-boumb} we get:
\begin{align*}
  |\tdelta((\qpsi,\sstate'),a)|&\leq 2^{|\max(\psi)|} \left (
2|\tQ_{\psi'}||\CKS|-1 + 
2 \sum_{i\in[n]}|\CKS|2^{\paraphi(\phi_i)|\phi_i|}|\tQ_i|^{|\CKS|}
                  \right )+2^{|\max(\psi)|}+2|\max(\psi)|2^{|\max(\psi)|}\\
    |\tdelta((\qpsi,\sstate'),a)|&\leq 2^{|\max(\psi)|+1}|\CKS| \left (
|\tQ_{\psi'}| + 
 \sum_{i\in[n]}2^{\paraphi(\phi_i)|\phi_i|}|\tQ_i|^{|\CKS|}
    \right )+2|\max(\psi)|2^{|\max(\psi)|}
\end{align*}
But for natural numbers $\{a_i,b_i\}_{i\in [n]}$, it holds that
\[ \sum_{i\in[n]}2^{a_i}b_i =
  2^{\sum_{i\in[n]}a_i}\sum_{i\in[n]}b_i-\sum_{i\in[n]}2^{a_i}(2^{\sum_{j\neq
      i}a_j}-1)b_i\]
Applying this to $a_i=\paraphi(\phi_i)|\phi_i|$ and $b_i=|\tQ_i|^{|\CKS|}$
we obtain
\[ \sum_{i\in[n]}2^{\paraphi(\phi_i)|\phi_i|}|\tQ_i|^{|\CKS|} =
  2^{\sum_{i\in[n]}\paraphi(\phi_i)|\phi_i|}\sum_{i\in[n]}|\tQ_i|^{|\CKS|}-\sum_{i\in[n]}2^{\paraphi(\phi_i)|\phi_i|}(2^{\sum_{j\neq
      i}\paraphi(\phi_j)|\phi_j|}-1)|\tQ_i|^{|\CKS|}\]
We thus get that
\[    |\tdelta((\qpsi,\sstate'),a)|\leq 2^{|\max(\psi)|+1}|\CKS| \left (
|\tQ_{\psi'}| +   2^{\sum_{i\in[n]}\paraphi(\phi_i)|\phi_i|}\sum_{i\in[n]}|\tQ_i|^{|\CKS|}
    \right )+C,
  \]
  with
  \begin{align*}
    C&=2|\max(\psi)|2^{|\max(\psi)|}-2^{|\max(\psi)|+1}|\CKS|\sum_{i\in[n]}2^{\paraphi(\phi_i)|\phi_i|}(2^{\sum_{j\neq
       i}\paraphi(\phi_j)|\phi_j|}-1)|\tQ_i|^{|\CKS|}\\
    &= 2^{|\max(\psi)|}\left(2|\max(\psi)|-2|\CKS|\sum_{i\in[n]}2^{\paraphi(\phi_i)|\phi_i|}(2^{\sum_{j\neq
      i}\paraphi(\phi_j)|\phi_j|}-1)|\tQ_i|^{|\CKS|}\right )
  \end{align*}
  If $n=|\max(\psi)|>1$, \ie, there are at least two maximal state
  subformulas, then $\sum_{j\neq
    i}\paraphi(\phi_j)|\phi_j|>0$, hence
  $2|\CKS|\sum_{i\in[n]}2^{\paraphi(\phi_i)|\phi_i|}(2^{\sum_{j\neq
      i}\paraphi(\phi_j)|\phi_j|}-1)|\tQ_i|^{|\CKS|}\geq
  4n=4|\max(\psi)|$, which implies that $C\leq 0$, and thus
  \begin{align*}    
    |\tdelta((\qpsi,\sstate'),a)|&\leq 2^{|\max(\psi)|+1}|\CKS| \left (
|\tQ_{\psi'}| +   2^{\sum_{i\in[n]}\paraphi(\phi_i)|\phi_i|}\sum_{i\in[n]}|\tQ_i|^{|\CKS|}
                    \right )\\
   &\leq   2^{|\max(\psi)|+1}|\CKS| 
 2^{\sum_{i\in[n]}\paraphi(\phi_i)|\phi_i|}\left
     (|\tQ_{\psi'}|^{|\CKS|} + \sum_{i\in[n]}|\tQ_i|^{|\CKS|}\right)\\
       &\leq   |\CKS| 
 2^{|\max(\psi)|+1+(\paraphi(\phi)-1)\sum_{i\in[n]}|\phi_i|}
         \left (|\tQ_{\psi'}| + \sum_{i\in[n]}|\tQ_i|\right)^{|\CKS|}\\
           &\leq   |\CKS|  2^{|\phi|+(\paraphi(\phi)-1)|\phi|}
             |\tQ|^{|\CKS|}\\
    |\tdelta((\qpsi,\sstate'),a)|&\leq    |\CKS|  2^{\paraphi(\phi)|\phi|} |\tQ|^{|\CKS|}
  \end{align*}

  It remains to consider the case where $\max(\psi)=\{\phi_1\}$. In
  that case there are only two letters in the alphabet
  $2^{\max(\psi)}$, which are $\emptyset$ and $\{\phi_1\}$. 
  The transition formulas then simplify and one gets that
  \begin{align*}
    |\tdelta((\qpsi,\sstate'),a)|&\leq
                                   |\tdelta_{\psi}((\qpsi,\sstate'),\emptyset)|
                                   + 1 +
                                   |\compl{\tdelta^{1}_{\sstate'}}(\compl{\tq^{1}_{\sstate'}},a)|+1+
                                   |\tdelta_{\psi}((\qpsi,\sstate'),\{\phi_1\})|
                                   + 1 +
                                   |\tdelta^{1}_{\sstate'}(\tq^{1}_{\sstate'},a)|
  \end{align*}
  Using Equation~\eqref{eq:bouya} and the induction hypothesis we get
  \begin{align*}
    |\tdelta((\qpsi,\sstate'),a)| &\leq
      4|\tQ_{\psi'}||\CKS|-2+2|\CKS|2^{\paraphi(\phi_1)|\phi_1|}|\tQ_1|^{|\CKS|}+3\\
    &\leq 1 + 2
      |\CKS|(2|\tQ_{\psi'}|+2^{\paraphi(\phi_1)|\phi_1|}|\tQ_1|^{|\CKS|})\\
    & \leq 1 + 2 |\CKS|
      2^{(\paraphi(\phi)-1)|\phi_1|}(|\tQ_{\psi'}|^{|\CKS|}+|\tQ_1|^{|\CKS|})\\
        & \leq 1 +  |\CKS|
          2^{\paraphi(\phi)|\phi|}(|\tQ_{\psi'}|^{|\CKS|}+|\tQ_1|^{|\CKS|})\\
    |\tdelta((\qpsi,\sstate'),a)|            & \leq  |\CKS| 2^{\paraphi(\phi)|\phi|}|\tQ|^{|\CKS|}
  \end{align*}

  \halfline
  $\bm{\phi=\exists}^{\bm{\cobs}}\bm{p.\,\phi':}$
  We first establish the claim for states and colours, and we start with the case $\ndd_k(\phi)=\ndd_k(\phi')$. By definition we
  necessarily have that $\ndd_x(\phi')=\nd$, \ie, $\bigauto{\phi'}$ is
  nondeterministic, and $\cobs=\Iphi[\phi']$,
  therefore there is no need to use narrowing or nondeterminisation
  here. $\bigauto{\phi}$ is obtained by directly projecting
  $\bigauto{\phi'}$, an operation that does not change the number of
  states or colours, so that  the claim for states and colours follows
  directly by induction hypothesis.


  Now we consider the case where $\ndd_k(\phi)\neq \ndd_k(\phi')$,
  which implies that
  $\ndd_k(\phi)\geq 1$. Let $n$ be the number of states and $l$ the
  number of colours in
  $\bigauto{\phi'}$.  In this case $\bigauto{\phi'}$ is first narrowed
  down, which does not change number of states or colours. The
  resulting automaton is then
  nondeterminised, yielding an automaton with at most $2^{\ka nl\log
    nl}$ states and $\kb nl$ colours.

  Again, we split cases: if $\ndd_k(\phi')=0$, by induction
  hypothesis, $n\leq \gsphi[\phi']$ and $l=2$. For the number of
  colours, observing that $\rEd(\phi)=\rEd(\phi')+1$, we have
  \begin{align*}
   \kb n l \leq 2 \kb \gsphi[\phi']&= 2 \kb (4\ka+2\kb)^{\rEd(\phi')}|\phi'||\CKS|^{\Ed(\phi')}2^{\kpsi|\phi'|\Ed(\phi')}\\
    &\leq
      (4\ka+2\kb)^{\rEd(\phi)}|\phi||\CKS|^{\Ed(\phi)}2^{\kpsi|\phi|\Ed(\phi)}\\
    \kb n l & \leq \tower{\ndd_k(\phi)-1}{\gsphi \log \gsphi}
  \end{align*}
For the number of states, we have that 
\begin{align*}
  2^{\ka nl\log nl} &\leq 2^{2\ka \gsphi[\phi']\log(2\gsphi[\phi'])}
                      \leq \tower{\ndd_k{\phi}}{\gsphi[\phi]\log(\gsphi[\phi])}
\end{align*}

Now for the final case, if 
$\ndd_k(\phi)=\ndd_k(\phi')+1$ and $\ndd_k(\phi')\geq 1$, by induction
hypothesis $n\leq \tower{\ndd_k(\phi')}{\gsphi[\phi']\log
  \gsphi[\phi']}$   and $l\leq
\tower{\ndd_k(\phi')-1}{\gsphi[\phi']\log \gsphi[\phi']}$. For the
number of colours in $\bigauto{\phi}$ we thus get
\begin{align*}
  \kb n l &\leq \kb \tower{\ndd_k(\phi')-1}{\gsphi[\phi']\log
            \gsphi[\phi']2^{\gsphi[\phi']\log \gsphi[\phi']}}\\
  &\leq  \tower{\ndd_k(\phi')}{2\kb\gsphi[\phi']\log
    \gsphi[\phi']}\\
  \kb n l  &\leq  \tower{\ndd_k(\phi)-1}{\gsphi[\phi]\log
            \gsphi[\phi]}
\end{align*}
Concerning the number of states, we observe that
\begin{align*}
  nl &\leq \tower{\ndd_k(\phi')-1}{\gsphi[\phi']\log
       \gsphi[\phi']2^{\gsphi[\phi']\log \gsphi[\phi']}}\\
nl  &\leq\tower{\ndd_k(\phi')}{2\gsphi[\phi']\log
       \gsphi[\phi']}\\
\ka nl \log nl &\leq \tower{\ndd_k(\phi')-1}{2\ka\gsphi[\phi']\log
             \gsphi[\phi']2^{2\gsphi[\phi']\log \gsphi[\phi']}}\\
\ka n l \log n l  &\leq\tower{\ndd_k(\phi')}{4\ka\gsphi[\phi']\log
    \gsphi[\phi']}\\
\ka n l \log n l    &\leq\tower{\ndd_k(\phi')}{\gsphi[\phi]\log
                      \gsphi[\phi]}\\
2^{\ka n l \log n l}    &\leq\tower{\ndd_k(\phi)}{\gsphi[\phi]\log
       \gsphi[\phi]}
\end{align*}

It only remains to establish the claim for the size of transition
formulas. Since $\bigauto{\phi}$ is nondeterministic, formulas  $\tdelta(\tq,a)$ are written in
disjunctive normal form and for every direction $\dir\in \SI[\phi]$ 
each disjunct contains exactly one element of $\{\dir\}\times \tQ$,
where $\tQ$ is the set of states in $\bigauto{\phi}$. As a
result, each formula $\tdelta(\tq,a)$ is of size 
\begin{align*}
  |\tdelta(\tq,a)| &\leq |\tQ|^{|\SI[\phi]|}(2
                     |\SI[\phi]|-1)+|\tQ|^{|\SI[\phi]|}-1\\
                   &\leq 2 |\SI[\phi]||\tQ|^{|\SI[\phi]|}\\
  |\tdelta(\tq,a)|  &\leq 2^{\paraphi(\phi)|\phi|}|\CKS||\tQ|^{|\CKS|}
\end{align*}
\end{proof}


\end{document}